\documentclass[12pt]{article}

\usepackage[top=1.2in, bottom=1.2in, left=1in, right=1in]{geometry}

\usepackage[sort&compress,authoryear]{natbib}
\usepackage{BOONDOX-cal}
\usepackage{epsfig, color}
\usepackage{graphicx}
\usepackage{amsmath,amsthm,amsfonts,amssymb,bm}
\usepackage[figuresright]{rotating}

\usepackage{makeidx}
\usepackage{listings}
\usepackage[colorlinks,
linkcolor=black,
anchorcolor=black,
citecolor=black
]{hyperref}
\usepackage{enumerate}

\usepackage{rotating}
\usepackage{graphicx}
\usepackage{tikz}
\usepackage{multirow}
\usepackage{subfigure}
\usepackage{floatrow}
\usepackage{cases}
\newtheorem{proposition}{\hspace{6mm}Proposition}[section]
\newtheorem{lemma}{\hspace{6mm}Lemma}[section]

\def\be{\begin{equation}}
\def\ee{\end{equation}}

\title{Optimal life insurance and annuity decision under money illusion}
\author{
    Wenyuan Li\thanks{\rm Department of Statistics and Actuarial Science, The University of Hong Kong, Pokfulam, Hong Kong. Email: wylsaas@hku.hk} 
    \and 
    Pengyu Wei\thanks{\rm Insurance Risk and Finance Centre, Division of Banking \& Finance, Nanyang Business School, Nanyang Technological University, Singapore. E-mail: pengyu.wei@ntu.edu.sg}
}
\date{\today}

\begin{document}

\maketitle

\begin{abstract}
This paper investigates the optimal consumption, investment, and life insurance/annuity decisions for a family in an inflationary economy under money illusion. The family can invest in a financial market that consists of nominal bonds, inflation-linked bonds, and a stock index. The breadwinner can also purchase life insurance or annuities that are available continuously.  The family's objective is to maximize the expected utility of a mixture of nominal and real consumption, as they partially overlook inflation and tend to think in terms of nominal rather than real monetary values. We formulate this life-cycle problem as a random horizon utility maximization problem and derive the optimal strategy. We calibrate our model to the U.S. data and demonstrate that money illusion increases life insurance demand for young adults and reduces annuity demand for retirees. Our findings indicate that the money illusion contributes to the annuity puzzle and highlights the role of financial literacy in an inflationary environment.
\end{abstract}

% \noindent {\textbf{JEL classification:} C61; D91; G11; G22; H55}\\
\section{Introduction}

Life insurance and annuity decision-making are crucial in studying the life-cycle model, aiding individuals in managing the mortality and financial risks in long-term investments. \cite{yaari1965uncertain} pioneers adding the random horizon to the life cycle model, highlighting the separability of consumption and bequest decisions when insurance is an option. \cite{fischer1973life} examines the life-cycle model in discrete time, observing that individuals with labor income are inclined to purchase life insurance early in life and sell it later on, and proposes annuity purchases to solve the short-sale constraints of life insurance. \cite{richard1975optimal} delves into an individual's life and annuity decisions within a continuous-time framework, finding that those anticipating high future income are likely to purchase life insurance regardless of relative preference in consumptions and bequests. \cite{pliska2007optimal} extend \cite{richard1975optimal}'s bounded lifetime assumption to an unbounded scenario, exploring life insurance demand across various economic factors and drawing significant economic insights. For pertinent work, we refer to \cite{hakansson1969optimal}, \cite{huang2008portfolio}, \cite{koijen2011optimal}, \cite{ekeland2012time}, \cite{kwak2014optimal}, \cite{wei2020optimal}, \cite{bernard2021optimal}, \cite{fischer2023optimal}, \cite{li2023optimal}, \cite{kizaki2024multi}, and \cite{wang2024optimal}, among others. However, none of the existing literature considers the effects of money illusion on an individual's consumption, investment, and insurance demands.  

The concept of money illusion, whereby individuals assess their utility based on nominal rather than real monetary value, was introduced in \cite{Fisher1928} and is substantiated by empirical research \citep[e.g.][]{shafir1997money,fehr2001does,fehr2007money,svedsater2007money}. In the economics literature, \cite{brunnermeier2008money} demonstrates that reducing inflation can significantly increase house prices for individuals affected by money illusion. \cite{miao2013economic} explores economic growth and determines that money illusion influences an agent's perception of real wealth growth and risk. \cite{basak2010equilibrium} highlight the substantial impact of money illusion on stock equilibrium prices.  \cite{lioui2023money} find that individuals influenced by money illusion tend to shift from inflation-indexed to nominal bonds. For more related work, we refer to \cite{david2013ties}, \cite{chen2009stock}, \cite{chen2013sources}, \cite{he2014myopic}, and \cite{eisenhuth2017money}. In actuarial science literature, \cite{wei2023optimal} study the money illusion effect within a defined-contribution plan framework. They reveal that the money illusion decreases an individual's holdings of inflation-linked bonds and results in substantial welfare loss. \cite{donnelly2024money} extend \cite{wei2023optimal}'s work with a minimum guarantee at retirement time. They find that the minimum guarantee can significantly reduce welfare loss caused by money illusion. However, both papers focus solely on the pre-retirement arrangement, overlooking post-retirement management. Furthermore, neither paper considers the individual's mortality risk, leaving research gaps to study the individual's insurance and annuity demands under the money illusion.

We consider the money illusion effect in a life cycle model. The family comprises the breadwinner and the rest of the family, both affected by money illusion in their consumption decisions. For the financial market, we adopt a two-factor model proposed by \cite{koijen2011optimal} to describe time-varying real interest rates, inflation rates, and risk premiums. The family can invest a part of the wealth in a stock index, nominal bonds, inflation-linked bonds, and a cash account. Simultaneously, they allocate the other part of the wealth to purchase life insurance to manage the breadwinner's mortality risk. Specifically, the family continuously pays the premium while the breadwinner is alive. Upon the breadwinner's death, the rest of the family receives a death benefit comprising the wealth and a life insurance payment, after which they continue managing their investments independently. We formulate this life-cycle problem as a random horizon utility maximization problem. The family's consumption utility combines nominal and real consumption, reflecting their inclination towards nominal monetary values over real values. In addition, the timeline is divided into three scenarios: pre-breadwinner's death and retirement, post-breadwinner's retirement but pre-breadwinner's death, and post-breadwinner's death. We derive the corresponding Hamilton-Jacobi-Bellman (HJB) equations and the optimal strategies. Under the standard constant relative risk aversion (CRRA) utility function, we obtain explicit solutions for both value functions and optimal strategies. Global existence conditions are presented to ensure that the explicit solutions do not explode in finite time. Furthermore, we prove that the explicit solutions exactly solve the HJB equations via the verification theorem.

We calibrate our model to the U.S. data and numerically illustrate the money illusion effect on the family's consumption, investment, and insurance demands. Through dynamic analysis, we plot the expected trajectories of optimal strategies evolving over time. We observe that both the high-risk-aversion breadwinner and the rest of the family will shift from an increasing to a decreasing consumption pattern when considering the money illusion. Regarding investment strategies, the family will short fewer short-term nominal bonds in exchange for longing for more long-term nominal bonds before late age (age 80), with an opposite trend following this threshold. The demand for inflation-linked bonds diminishes due to the influence of money illusion, while the impact on stock index investments remains minimal. For insurance strategies, the breadwinner purchases life insurance at working age and switches to annuities near retirement. Introducing the money illusion leads to a higher life insurance demand and a lower annuity demand for the breadwinner. The static analysis shows that the optimal consumption and annuity strategies follow an upward ``U-shape'' pattern with two factors (interest rate and inflation rate factors). In contrast, the life insurance strategy displays a downward ``U-shape'' pattern. This is because life insurance safeguards future income and substitutes for current consumption, contrasting with annuities bolstering current consumption through income streams. Lastly, we investigate the family's welfare loss from the money illusion. We observe that welfare loss increases with the degree of money illusion and the risk-aversion coefficient. This escalation is attributed to the risk-averse behavior of individuals, leading them to allocate more to riskless assets. While non-illusioned investors perceive inflation-linked bonds as riskless, illusioned investors favor nominal bonds, creating a divergence in welfare loss stemming from the preference discrepancy.

Our paper contributes to the existing literature in three aspects: (i) We study the money illusion effect on the trading strategy over short-term nominal bonds, long-term nominal bonds, inflation-linked bonds, and stock indexes. We find that when ignoring inflation risk, the family reduces the short position in short-term nominal bonds and long more long-term nominal bonds. Additionally, they decrease holdings in inflation-linked bonds while maintaining a steady stance on stocks. (ii) We explore the influence of money illusion on consumption strategies. Our results suggest a shift in consumption patterns from increasing to decreasing when ignoring the inflation risk. (iii) In examining life insurance and annuity decisions, we discover that neglecting inflation risk leads to an increased demand for life insurance but a reduced demand for annuities among breadwinners. This phenomenon contributes to the annuity puzzle, a well-documented issue where individuals exhibit reluctance to purchase annuities close to retirement age. Existing explanations for the annuity puzzle include annuity mispricing (\cite{davidoff2005annuities}), bequest motives (\cite{lockwood2012bequest}), mortality misperceive (\cite{han2021annuity}), risk aversion (\cite{milevsky2007annuitization}), and time inconsistency (\cite{zhang2021optimal}), etc. Our analysis introduces a novel perspective by highlighting money illusion as a potential driver of the annuity puzzle, particularly in an inflationary economic environment.

The rest of the paper is organized as follows: Section \ref{model_settings} introduces the model settings, including the financial market, mortality, wealth process, and breadwinner's objective. Section \ref{optimization_problem} solves the optimization problem via the HJB equations. Global existence and verification theorem are proved for explicit solutions. Section \ref{numerical_research} calibrates the model with U.S. data, conducts sensitivity analysis of the optimal strategies, and analyzes the welfare loss. Section \ref{conclusion} concludes.

\section{Model settings}\label{model_settings}
\subsection{Financial market}
We consider a financial market similar to that proposed by \cite{koijen2011optimal}, which incorporates time variations in real interest rates, inflation rates, and risk premia. Let $(\Omega, \mathcal{F}, \mathbb{P})$ be a filtered complete probability space. We use a four-dimensional vector of independent Brownian motions $Z_t$ to characterize the financial risk, which generates the filtration $\mathbb{F}:=\{\mathcal{F}_t\}_{t\in[0,T]}$.
 
Within the financial market, the real short rate is modelled as affine in a single factor, $X_1$,
\begin{equation*}
	r_t = \delta_r + X_{1,t}, ~~\delta_r > 0.
\end{equation*}
Meanwhile, the expected inflation rate is influenced by a second factor, $X_2$,
\begin{equation*}
	\pi^e_t = \delta_{\pi^e} + X_{2,t}, ~~\delta_{\pi^e}>0.
\end{equation*}
These two factors are governed by an Ornstein-Uhlenbeck process
\begin{equation}\label{Xt_process}
	dX_t = -K_X X_t dt + \Sigma_X dZ_t,
\end{equation}
where $X_t = (X_{1,t}, X_{2,t})^{\top}$, $K_X = diag(\kappa_1, \kappa_2),~\kappa_i>0,~i=1,2$, $\Sigma_X = (\sigma_1,\sigma_2)^{\top},~\sigma_i \in \mathbb{R}^4,~i=1,2$. The evolution of realized inflation is captured by the stochastic differential equation (SDE)
\begin{equation}\label{realized_inflation}
	\frac{d\Pi_t}{\Pi_t} = \pi^e_t dt + \sigma^{\top}_{\Pi} dZ_t,~\Pi_0 = 1,
\end{equation}
where $\Pi_t$ represents the level of the (consumer) price index at time $t$ and $\sigma_{\Pi} \in \mathbb{R}^4$. The equity index $S_t$ is given by
\begin{equation*}
	\frac{dS_t}{S_t} = \mu_t dt + \sigma^{\top}_{S} dZ_t,
\end{equation*}
where the drift term is defined as $\mu_t = R_t + \mu_0 + \mu^{\top}_1 X_t$, with $R_t$ representing the instantaneous nominal short rate as specified in \eqref{Rt} below. To aid in model identification, we impose the condition that the volatility matrix $(\sigma_1, \sigma_2, \sigma_{\Pi}, \sigma_S)^{\top}$ is lower triangular.

We postulate the nominal state price density $\phi$ evolve as
\begin{equation*}\label{phit}
	\frac{d \phi_t}{\phi_t} = - R_t dt - \Lambda^{\top}_t dZ_t, \quad \phi_0 = 1,
\end{equation*}
where the market prices of risk, denoted by $\Lambda_t$, are affine in the term-structure variables, i.e.,
\begin{equation*}\label{Lambda_t}
	\Lambda_t = \Lambda_0 + \Lambda_1 X_t.
\end{equation*}
We follow \cite{koijen2011optimal} to impose restrictions on $\Lambda_0$ and $\Lambda_1$
\begin{equation*}
	\Lambda_0 =
	\left(
	\begin{array}{c}
		\Lambda_{0(1)}\\
		\Lambda_{0(2)}\\
		 0\\
		\Lambda_{0(4)}\\
	\end{array}
	\right),~
	\Lambda_1 =
	\left(
	\begin{array}{ccc}
		\Lambda_{1(1,1)} & 0\\
		0 & \Lambda_{1(2,2)}\\
		0 & 0               \\
		\Lambda_{1(4,1)} & \Lambda_{1(4,2)} \\
	\end{array}
	\right),
\end{equation*}
with $\sigma^{\top}_S\Lambda_0 = \mu^{\top}_0$ and $\sigma^{\top}_S\Lambda_1 = \mu^{\top}_1$. The real state price density $\phi^R_t = \phi_t\Pi_t$ is governed by
\begin{eqnarray*}
	\frac{d\phi^R_t}{\phi^R_t} = - (R_t - \pi^e_t + \sigma^{\top}_{\Pi} \Lambda_t) dt - (\Lambda^{\top}_t - \sigma^{\top}_{\Pi}) dZ_t 
	= - r_t dt - (\Lambda^{\top}_t - \sigma^{\top}_{\Pi}) dZ_t, \quad \phi^R_0 = 1,
\end{eqnarray*}
which implies the instantaneous nominal short rate can be expressed as
\begin{equation}\label{Rt}
	R_t = \delta_R + (\iota^{\top}_2 - \sigma^{\top}_{\Pi}\Lambda_1)X_t,
\end{equation}
where $\delta_R = \delta_r + \delta_{\pi^e} - \sigma^{\top}_{\Pi}\Lambda_0$ and $\iota_2  = (1,1)^{\top}$.

Finally, we present the prices of nominal and inflation-linked bonds, deriving from the standard approach in the literature \citep[e.g.][]{duffie1996yield}. The time-$t$ price of a nominal bond with maturity $s$ has an exponential affine form
\begin{equation*}
	P(X_t,t,s) = \exp\{A_0(s-t) + [A_1(s-t)]^{\top} X_t\},
\end{equation*}
where $A_0$ and $A_1$ are determined by the following ordinary differential equation (ODE) system
\begin{eqnarray}
	&& \frac{\partial A_0(\tau)}{\partial \tau} = \frac{1}{2}[A_1(\tau)]^{\top}\Sigma_X \Sigma^{\top}_X A_1(\tau) - [A_1(\tau)]^{\top} \Sigma_X \Lambda_0 - \delta_R, ~A_0(0)=0,\label{A0}\\
	&& \frac{\partial A_1(\tau)}{\partial \tau} = -[K^{\top}_X + \Lambda^{\top}_1\Sigma^{\top}_X] A_1(\tau) - \iota_2 + \Lambda^{\top}_1 \sigma_{\Pi},~A_1(0)=0. \label{A1}
\end{eqnarray}
Additionally, the dynamics of $P(X_t,t,s)$ are governed by 
\begin{equation*}
	\frac{dP(X_t,t,s)}{P(X_t,t,s)} = \{ R_t + [A_1(s-t)]^{\top}\Sigma_X \Lambda_t \} dt + [A_1(s-t)]^{\top} \Sigma_X dZ_t.
\end{equation*}
Similarly, the time-$t$ real price of an inflation-linked bond with maturity $s$ is
\begin{equation}
	P^R(X_t,t,s) = \exp\{A^R_0(s-t) + [A^R_1(s-t)]^{\top} X_t\},\label{inflation_linked_bond}
\end{equation}
where $A^R_0$ and $A^R_1$ are subject to the ODE system
\begin{eqnarray*}
	&&\frac{\partial A^R_0(\tau)}{\partial \tau} = \frac{1}{2}[A^R_1(\tau)]^{\top}\Sigma_X\Sigma_X^{\top} A^R_1(\tau) - [A^R_1(\tau)]^{\top}\Sigma_X(\Lambda_0 - \sigma_{\Pi}) - \delta_r, \label{AR0}\\
	&&\frac{\partial A^R_1(\tau)}{\partial \tau} = -(K^{\top}_X + \Lambda^{\top}_1 \Sigma^{\top}_X)A^R_1(\tau) - e_1,\label{AR1}\\
	&&A^R_0(0)=A^R_1(0)=0, \notag
\end{eqnarray*}
where $e_i$ represents the $i$-th unit vector in $\mathbb{R}^2$. Then, the nominal price of the inflation-linked bond, $\Pi_t P^R(X_t,t,s)$, follows the SDE 
\begin{equation*}
	\frac{d(\Pi_t P^R(X_t,t,s))}{\Pi_t P^R(X_t,t,s)} = \{R_t + [A^R_1(s-t)]^{\top}\Sigma_X\Lambda_t + \sigma^{\top}_{\Pi}\Lambda_t\}dt + \{[A^R_1(s-t)]^{\top}\Sigma_X + \sigma^{\top}_{\Pi}\} dZ_t.
\end{equation*}

\subsection{Mortality}
This subsection introduces the breadwinner's mortality risk. We use $T_x$ to denote the future lifetime of a breadwinner aged $x$, which is a non-negative random variable independent of the financial market (i.e., $T_x$ is independent of the filtration $\mathbb{F}$ associated with the financial market). We then define the following probabilities
\begin{eqnarray*}
	_{t}p_x = \mathbb{P}[T_x>t], \quad 
	_{t}q_x = \mathbb{P}[T_x\leq t] = 1- {_{t}p_x}, \quad 
	\lim \limits_{t\rightarrow \infty} {_{t}p_x} =0, \quad 
 \lim \limits_{t\rightarrow \infty} {_{t}q_x} =1,
\end{eqnarray*}
where $_{t}p_x$ represents the probability that the breadwinner aged $x$ survives to at least age $x+t$, and $_{t}q_x$ is the probability that the breadwinner dies before age $x+t$. In actuarial science, it is conventional to define the instantaneous force of mortality (or hazard rate) as
\begin{equation*}
	\lambda_{x+t} = \frac{1}{_{t}p_x}\frac{d}{dt} {_{t}q_x} = - \frac{1}{_{t}p_x}\frac{d}{dt} {_{t}p_x},
\end{equation*}
leading to the relationships
\begin{eqnarray*}
_{t}p_x = \exp\left\{ - \int_0^t \lambda_{x+s} ds \right\}, \quad _{t}q_x = \int_0^t {_{s}p_x}\lambda_{x+s} ds.
\end{eqnarray*}
The probability density function of $T_x$ is then expressed as $f_{T_x}(t) = {_{t}p_x}\lambda_{x+t}$ for $t>0$.

\subsection{Wealth Process}
We consider two dates of interest: the breadwinner's retirement time denoted by $T_R$ and the terminal time denoted by $T$. The breadwinner can purchase life insurance before the first time of death time $T_x$ and retirement time $T_R$. The nominal wealth follows
\begin{equation*}
    \left\{
    \begin{array}{ll}
    dW_t &= W_t (\alpha^{\top}_t\Sigma\Lambda_t + R_t) dt + W_t \alpha^{\top}_t \Sigma dZ_t + Y^{\$}_tdt - I^{\$}_t dt - c^{\$}_{1,t}dt - c^{\$}_{2,t}dt, ~0\leq t< T_R \wedge T_x,\\
    dW_t &= W_t (\alpha^{\top}_t\Sigma\Lambda_t + R_t) dt + W_t \alpha^{\top}_t \Sigma dZ_t - I^{\$}_t dt - c^{\$}_{1,t}dt - c^{\$}_{2,t}dt, ~T_R\leq t< T_x \wedge T,\\
    dW_t &= W_t(\alpha^{\top}_t\Sigma\Lambda_t + R_t) dt + W_t \alpha^{\top}_t \Sigma dZ_t - c^{\$}_{2,t}dt, ~T_x\wedge T \leq t\leq T,
    \end{array}
    \right.
\end{equation*}
where the initial condition is $W_0=Y^{\$}_0$, $Y^{\$}$ represents the nominal income, $I^{\$}$ is the nominal insurance premium, $c^{\$}_1$ denotes the nominal consumption for the breadwinner, $c^{\$}_2$ represents the nominal consumption for the rest of the family, and $\Sigma$ is the volatility matrix given by
\begin{equation*}
\Sigma = \left( \begin{matrix}
		[A_1(T_1)]^{\top}\Sigma_X \\
		[A_1(T_2)]^{\top}\Sigma_X \\
		[A^R_1(T_3)]^{\top}\Sigma_X + \sigma^{\top}_{\Pi} \\
		\sigma^{\top}_S
	\end{matrix}
	\right).
\end{equation*}
We introduce the real wealth $W^R_t = W_t/\Pi_t$, leading to the SDEs
\begin{equation}
    \left\{
    \begin{array}{ll}
    dW^R_t &= W^R_t [r_t + (\alpha^{\top}_t\Sigma-\sigma^{\top}_{\Pi})(\Lambda_t - \sigma_{\Pi}) ] dt + W^R_t (\alpha^{\top}_t \Sigma-\sigma^{\top}_{\Pi}) dZ_t + Y_tdt - I_t dt - c_{1,t}dt\\
    &-c_{2,t}dt,~~0\leq t< T_R \wedge T_x, \label{real_wealth_1}\\
    dW^R_t &= W^R_t [r_t + (\alpha^{\top}_t\Sigma-\sigma^{\top}_{\Pi})(\Lambda_t - \sigma_{\Pi}) ] dt + W^R_t (\alpha^{\top}_t \Sigma-\sigma^{\top}_{\Pi}) dZ_t - I_t dt - c_{1,t}dt-c_{2,t}dt,\\
    &~~~~T_R\leq t< T_x \wedge T,\\
    dW^R_t &= W^R_t[r_t + (\alpha^{\top}_t\Sigma-\sigma^{\top}_{\Pi})(\Lambda_t - \sigma_{\Pi})] dt + W^R_t (\alpha^{\top}_t \Sigma - \sigma^{\top}_{\Pi}) dZ_t - c_{2,t}dt, ~T_x \wedge T \leq t \leq T.
    \end{array}
    \right.
\end{equation}

In equation \eqref{real_wealth_1}, the initial condition is $W^R_0=Y_0$.  $Y$ is the breadwinner's real income satisfying 
\begin{equation}
    dY_t = Y_t g^R_t dt, ~0\leq t < T_R \wedge T_x,\label{income_process}
\end{equation}
where $g^R_t$ is a deterministic function. The real insurance premium is denoted by $I$, $c_1$ represents the real consumption for the breadwinner, and $c_2$ is the real consumption for the rest of the family. At the death time $T_x$
\begin{equation*}
    W_{T_x} = W_{T_x-} + \frac{I_{T_x}}{\lambda_{x+T_x}}.
\end{equation*}

\subsection{Preference}
Let $U_1$ and $U_2$ denote the consumption utility of the breadwinner and the rest of the family, respectively. The rest of the family typically includes the spouse and/or children, with the spouse being the more commonly considered member \citep[see][]{bernheim1991strong, inkmann2011deep}. These consumption utilities depend not only on real consumption $c$ but also on inflation $\Pi$, which captures the family's preference for nominal monetary value over real value. Inspired by \cite{huang2008portfolio}, \cite{huang2008portfolio2}, and \cite{kwak2011optimal}, we make the assumption that the rest of the family has a certain life expectancy. The family's objective is to determine values for $c_1$, $c_2$, $I$, and $\alpha$ that maximize the expectation

\begin{align}
   & \sup \limits_{\alpha,c_1,c_2,I} E\left[ \kappa_1 \int_0^{T_x \wedge T} e^{-\delta t} U_1(c_{1,t},\Pi_t)dt+ \kappa_2 \int_0^T e^{-\delta t}U_2(c_{2,t},\Pi_t)dt\right]\label{objective}\\
   &=  \sup \limits_{\alpha,c_1,c_2,I} E\left\{ \int_0^T {_tp_x} e^{-\delta t} \left[\kappa_1 U_1(c_{1,t},\Pi_t)\right.\right.\notag\\
   &\left.\left.+\kappa_2 U_2(c_{2,t},\Pi_t)+\kappa_2\lambda_{x+t}E_t\left[\int_t^T e^{-\delta(u-t)}U_2(c_{2,u},\Pi_u)du\right]\right]dt
   \right\}.\label{fix_horizon_objective}
\end{align}
Here, $E_t[\cdot] = E[\cdot|\mathcal{F}_t]$, $\kappa_1$ and $\kappa_2$ are non-negative weight parameters for consumption utilities that satisfy $\kappa_1+\kappa_2 = 1$, and $\delta \geq 0$ represents the utility discount factor. Equation \eqref{fix_horizon_objective} transforms a random horizon problem into a fixed horizon problem.

\section{Optimization problem}\label{optimization_problem}
\subsection{HJB equations and optimal strategies}
Based on three divisions of time intervals, we introduce the first value function as
\begin{equation*}
    \Phi_1(t,w^R,\pi,X) = \sup \limits_{\alpha,c_2} E_{t,w^R,\pi,X}\left[\int_t^T e^{-\delta (s-t)} U_2(c_{2,s},\Pi_s) ds\right],
\end{equation*}
for $ T_x \wedge T \leq t \leq T$, where $E_{t,w^R,\pi,X}[\cdot]$ is short for $E[\cdot|W^R_t = w^R, \Pi_t = \pi, X_t=X]$. The secondary value function is defined as
\begin{eqnarray*}
    &&\Phi_2(t,w^R,\pi,X) = \sup \limits_{\alpha,c_1,c_2,I} E_{t,w^R,\pi,X}\left[\int_t^T 
    {_{s-t}p_{x+t}} e^{-\delta(s-t)}\left[\kappa_1 U_1(c_{1,s},\Pi_s) + \kappa_2 U_2(c_{2,s},\Pi_s)\right.\right.\\
    &&\left.\left.+\kappa_2\lambda_{x+s}\Phi_1\left(s,W^R_s+\frac{I_s}{\lambda_{x+s}},\Pi_s,X_s\right)\right]ds\right], 
\end{eqnarray*}
for $ T_R \leq t \leq  T_x \wedge T $. Lastly, for $0\leq t \leq T_R$, we define a primary value function
\begin{small}
\begin{eqnarray}
    &&V(t,w^R,\pi,X) =\sup \limits_{\alpha,c_1,c_2,I}   E_{t,w^R,\pi,X}\left[ \int_t^{T_R} 
    {_{s-t}p_{x+t}} e^{-\delta(s-t)}\left[\kappa_1 U_1(c_{1,s},\Pi_s) + \kappa_2 U_2(c_{2,s},\Pi_s)\right.\right.\notag\\
    &&\left.\left.+\kappa_2\lambda_{x+s}\Phi_1\left(s,W^R_s+\frac{I_s}{\lambda_{x+s}},\Pi_s,X_s\right)\right]ds + e^{-\delta(T_R-t)}{_{T_R-t}p_{x+t}} \Phi_2\left(T_R,W^R_{T_R},\Pi_{T_R},X_{T_R}\right)\right].\label{primal_value_function}
\end{eqnarray}
\end{small}

By the dynamic programming principle, the first value function satisfies the HJB equation
\begin{eqnarray}
	\sup \limits_{\alpha,c_2} \left\{  U_2(c_2,\pi) - \delta \Phi_1(t,w^{R},\pi,X) + \mathcal{D}^{\alpha,c_2}\Phi_1(t,w^{R},\pi,X)\right\}=0,\label{first_HJB_1}
\end{eqnarray}
where $\mathcal{D}^{\alpha,c_2}$ is the infinitesimal generator given by
\begin{align*}
&\mathcal{D}^{\alpha,c_2}\Phi_1(t,w^{R},\pi,X)=\frac{\partial \Phi_1}{\partial t}  - \frac{\partial \Phi_1}{\partial X} K_X X + \frac{\partial \Phi_1}{\partial \pi}\pi\pi^{e}_t\\
&+ \frac{\partial \Phi_1}{\partial w^R} \{ w^R[r_t + (\alpha^{\top} \Sigma -\sigma^{\top}_{\Pi})(\Lambda_t-\sigma_{\Pi})] - c_2\} + \frac{1}{2}\frac{\partial^2 \Phi_1}{(\partial w^R)^2}(w^R)^2(\alpha^{\top} \Sigma -\sigma^{\top}_{\Pi})(\Sigma^{\top}\alpha  -\sigma_{\Pi}) \notag\\
&+ \frac{1}{2}\frac{\partial^2 \Phi_1}{\partial \pi^2}\pi^2\sigma^{\top}_{\Pi}\sigma_{\Pi} +\frac{1}{2} \text{Tr}\bigg( \Sigma^{\top}_X \frac{\partial^2 \Phi_1}{\partial X^{\top} \partial X } \Sigma \bigg) + w^{R}(\alpha^{\top}\Sigma - \sigma^{\top}_{\Pi})\Sigma^{\top}_X \frac{\partial^2 \Phi_1}{\partial w^{R} \partial X^{\top}} \\
&+ \frac{\partial^2 \Phi_1}{\partial w^{R} \partial \pi}w^{R} \pi(\alpha^{\top}\Sigma-\sigma^{\top}_{\Pi})\sigma_{\Pi} + \frac{\partial^2 \Phi_1}{\partial \pi \partial X} \Sigma_X\sigma_{\Pi} \pi.
\end{align*}
After solving it, the optimal consumption and trading strategy are given by
\begin{eqnarray}
    c^*_{2,t} &=& (U_2')^{-1}\left( \frac{\partial \Phi_1}{\partial w^R},\Pi_t\right),\label{first_opt_c2}\\
    \alpha^*_t &=& \frac{-(\Sigma^{\top})^{-1}}{W^R_t\frac{\partial^2\Phi_1}{(\partial w^R)^2}}\left[ \frac{\partial \Phi_1}{\partial w^R}(\Lambda_t-\sigma_{\Pi}) + \frac{\partial^2 \Phi_1}{\partial w^R\partial \pi}\Pi_t \sigma_{\Pi} + \Sigma^{\top}_X \frac{\partial^2 \Phi_1}{\partial w^R \partial X^{\top}} \right]+(\Sigma^{\top})^{-1}\sigma_{\Pi}.\label{first_opt_alpha}
\end{eqnarray}
The HJB equation for the secondary value function is subject to 
\begin{eqnarray}
 &&	\sup \limits_{\alpha,I,c_1,c_2} \left\{\kappa_1 U_1(c_1,\pi) + \kappa_2 U_2(c_2,\pi) +\kappa_2\lambda_{x+t}\Phi_1\left(t,w^R+\frac{I}{\lambda_{x+t}},\pi,X\right)\right.\notag\\
 &&\left. - (\lambda_{x+t}+\delta) \Phi_2(t,w^{R},\pi,X) + \mathcal{D}^{\alpha,I,c_1,c_2}\Phi_2(t,w^{R},\pi,X)\right\}=0,\label{secondary_HJB_1}
\end{eqnarray}
where $\mathcal{D}^{\alpha,I,c_1,c_2}$ is the infinitesimal generator given by
\begin{align*}
&\mathcal{D}^{\alpha,I,c_1,c_2}\Phi_2(t,w^{R},\pi,X)=\frac{\partial \Phi_2}{\partial t}  - \frac{\partial \Phi_2}{\partial X} K_X X + \frac{\partial \Phi_2}{\partial \pi}\pi\pi^{e}_t\\
&+ \frac{\partial \Phi_2}{\partial w^{R}} \{ w^{R}[r_t + (\alpha^{\top} \Sigma -\sigma^{\top}_{\Pi})(\Lambda_t-\sigma_{\Pi})] - I - c_1 - c_2\} \notag\\
& + \frac{1}{2}\frac{\partial^2 \Phi_2}{(\partial w^{R})^2}(w^{R})^2(\alpha^{\top} \Sigma -\sigma^{\top}_{\Pi})(\Sigma^{\top}\alpha  -\sigma_{\Pi}) + \frac{1}{2}\frac{\partial^2 \Phi_2}{\partial \pi^2}\pi^2\sigma^{\top}_{\Pi}\sigma_{\Pi}+\frac{1}{2} \text{Tr}\bigg( \Sigma^{\top}_X \frac{\partial^2 \Phi_2}{\partial X^{\top} \partial X } \Sigma \bigg) \notag\\
&  + w^{R}(\alpha^{\top}\Sigma - \sigma^{\top}_{\Pi})\Sigma^{\top}_X \frac{\partial^2 \Phi_2}{\partial w^{R} \partial X^{\top}} + \frac{\partial^2 \Phi_2}{\partial w^{R} \partial \pi}w^{R} \pi(\alpha^{\top}\Sigma-\sigma^{\top}_{\Pi})\sigma_{\Pi} + \frac{\partial^2 \Phi_2}{\partial \pi \partial X} \Sigma_X\sigma_{\Pi} \pi,
\end{align*}
and the optimal strategies are
\begin{eqnarray}
    c^*_{1,t} &=& (U_1')^{-1} \left(\frac{1}{\kappa_1} \frac{\partial \Phi_2}{\partial w^R},\Pi_t\right),\label{secondary_opt_c1}\\
    c^*_{2,t} &=& (U_2')^{-1} \left( \frac{1}{\kappa_2} \frac{\partial \Phi_2}{\partial w^R},\Pi_t\right),\label{secondary_opt_c2}\\
    \alpha^*_t &=& \frac{-(\Sigma^{\top})^{-1}}{w^R \frac{\partial^2 \Phi_2}{(\partial w^R)^2}}\left[\frac{\partial \Phi_2}{\partial w^R}(\Lambda_t - \sigma_{\Pi}) + \frac{\partial^2 \Phi_2}{\partial w^R \partial \pi}\Pi_t\sigma_{\Pi}+\Sigma^{\top}_X\frac{\partial^2 \Phi_2}{\partial w^R \partial X^{\top}}\right]+(\Sigma^{\top})^{-1}\sigma_{\Pi},\label{secondary_opt_alpha}\\
    I^*_t &=& \lambda_{x+t} \left[ \left(\frac{\partial \Phi_1}{\partial w^R}\right)^{-1}\left(t,\frac{1}{\kappa_2}\frac{\partial \Phi_2}{\partial w^R},\Pi_t,X_t\right)-W^R_t \right].\label{secondary_opt_I}
\end{eqnarray}
where $\frac{\partial \Phi_1}{\partial w}$ is the partial derivative with the second variable of $\Phi_1(t,w,\pi)$. 

For the primary value function, we need a special treatment for the income process $Y_t$. Following the approach in \cite{deelstra2003optimal}, we introduce the surplus process $W^{\widetilde{Y}}_t$ as
\begin{equation}\label{surplus_process}
	W^{\widetilde{Y}}_t = W^R_t + \widetilde{Y}(t,X_t),
\end{equation}
where $\widetilde{Y}(t,X_t)$ represents the time-$t$ value of (discounted) future income 
\begin{equation*}\label{adc}
	\widetilde{Y}(t,X_t) = \int_t^T {_{s-t}p_{x+t}}P^R(X_t,t,s)Y_s ds.
\end{equation*}
Applying Ito's formula, we can express the differential of $\widetilde{Y}(t,X_t)$ as
\begin{eqnarray}\label{tildeC_SDE1}
	d\widetilde{Y}(t,X_t) &=& -Y_t dt + (r_t+\lambda_{x+t})\widetilde{Y}(t,X_t)dt + \frac{\partial \widetilde{Y}(t,X_t)}{\partial X}\Sigma_X (\Lambda_t-\sigma_{\Pi})dt \notag\\ 
	&&+ \frac{\partial \widetilde{Y}(t,X_t)}{\partial X} \Sigma_X dZ_t.
\end{eqnarray}
Assuming the existence of a process $\xi_t$, we can rewrite the above SDE as
\begin{eqnarray}\label{tildeC_SDE2}
	d\widetilde{Y}(t,X_t) &=& -Y_t dt + \widetilde{Y}(t,X_t)[r_t + (\xi^{\top}_t \Sigma -\sigma^{\top}_{\Pi})(\Lambda_t - \sigma_{\Pi})] dt  + \mu_{x+t} \widetilde{Y}(t,X_t) dt \notag\\
	&&+ \widetilde{Y}(t,X_t)(\xi^{\top}_t\Sigma - \sigma^{\top}_{\Pi})dZ_t,
\end{eqnarray}
which allows us to derive the $\xi$ by comparing the relevant terms in \eqref{tildeC_SDE1} and \eqref{tildeC_SDE2}
\begin{equation*}
	\xi_t = \frac{1}{\widetilde{Y}(t,X_t)} (\Sigma^{\top})^{-1}\Sigma^{\top}_X \frac{\partial \widetilde{Y}(t,X_t)}{\partial X^{\top}} + (\Sigma^{\top})^{-1}\sigma_{\Pi}.\label{xi_t}
\end{equation*}
Furthermore, we derive the SDE for the surplus process by combining \eqref{real_wealth_1} and \eqref{tildeC_SDE2}
\begin{eqnarray}
	dW^{\widetilde{Y}}_t &=& dW^R_t + d\widetilde{Y}(t,X_t) \notag\\
	&=& W^{\widetilde{Y}}_t \{r_t + (\beta^{\top}_t \Sigma - \sigma^{\top}_{\Pi})(\Lambda_t - \sigma_{\Pi})\}dt + W^{\widetilde{Y}}_t(\beta^{\top}_t\Sigma - \sigma^{\top}_{\Pi})dZ_t\notag\\
	&& + \mu_{x+t} \widetilde{Y}(t,X_t) dt - I_t dt  - c_{1,t}dt - c_{2,t}dt,\label{WY_1}
\end{eqnarray}
where $0\leq t < T \wedge T_x$ and $\beta^{\top}_t = [W^R_t \alpha^{\top}_t + \widetilde{Y}(t,X_t)\xi^{\top}_t]/W^{\widetilde{Y}}_t$. The SDE \eqref{WY_1} represents the investment in the financial market and the purchase of life insurance with premium $ - \mu_{x+t} \widetilde{Y}(t, X_t) + I_t$. When the breadwinner dies before retirement, the surplus process has the following jump
\begin{equation*}\label{death_benefit}
	W^{\widetilde{Y}}_t = W^{\widetilde{Y}}_{t-}-\widetilde{Y}(t,X_t) + \frac{I_t}{\mu_{x+t}},~\text{if}~T_x=t<T. 
\end{equation*} 
Then, by definition \eqref{surplus_process}, and given that $W^{\widetilde{Y}}_{T_R}=W^R_{T_R}$ at $T_R$, the objective function \eqref{primal_value_function} can be transformed to
\begin{small}
\begin{eqnarray*}
 &&V(t,w^{\widetilde{Y}},\pi,X) = \sup \limits_{\beta,c_1,c_2,I}  E_{t,w^{\widetilde{Y}},\pi,X}\left[ \int_t^{T_R} 
    {_{s-t}p_{x+t}} e^{-\delta(s-t)}\left[\kappa_1 U_1(c_{1,s},\Pi_s) + \kappa_2 U_2(c_{2,s},\Pi_s)+\kappa_2\right.\right.\notag\\
&&\left.\left. \lambda_{x+s}\Phi_1\left(s,W^{\widetilde{Y}}_t - \widetilde{Y}(t,X_t) + \frac{I_t}{\mu_{x+t}},\Pi_s,X_s\right)\right]ds + e^{-\delta(T_R-t)}{_{T_R-t}p_{x+t}} \Phi_2\left(T_R,W^{\widetilde{Y}}_{T_R},\Pi_{T_R},X_{T_R}\right)\right],
\end{eqnarray*}
\end{small}
where $E_{t,w^{\widetilde{Y}},\pi,X}[\cdot]$ is short for $E[\cdot|W^{\widetilde{Y}}_t = w^{\widetilde{Y}}, \Pi_t = \pi, X_t=X]$.
Moreover, the corresponding HJB equation is
\begin{eqnarray}
	\sup \limits_{\beta,I,c_1,c_2} \left\{ \kappa_1 U_1(c_1,\pi) + \kappa_2 U_2(c_2,\pi) +\kappa_2 \lambda_{x+t} \Phi_1\left(t,w^{\widetilde{Y}}-\widetilde{Y}(t,X)+\frac{I}{\lambda_{x+t}},\pi,X\right)\right.\notag\\
    \left.- (\lambda_{x+t}+\delta) V(t,w^{\widetilde{Y}},\pi,X) + \mathcal{D}^{\beta,I,c_1,c_2}V(t,w^{\widetilde{Y}},\pi,X)\right\}=0,\label{primary_adj_HJB_1}
\end{eqnarray}
where $\mathcal{D}^{\beta,I,c_1,c_2}$ is the infinitesimal generator given by
\begin{align}
&\mathcal{D}^{\beta,I,c_1,c_2}V(t,w^{\widetilde{Y}},\pi,X)=\frac{\partial V}{\partial t}  - \frac{\partial V}{\partial X} K_X X + \frac{\partial V}{\partial \pi}\pi\pi^{e}_t \notag\\
&+ \frac{\partial V}{\partial w^{\widetilde{Y}}} \{ w^{\widetilde{Y}}[r_t + (\beta^{\top} \Sigma -\sigma^{\top}_{\Pi})(\Lambda_t-\sigma_{\Pi})]+ \lambda_{x+t}\widetilde{Y}(t,X) - I - c_1 - c_2\} \notag\\
& + \frac{1}{2}\frac{\partial^2 V}{(\partial w^{\widetilde{C}})^2}(w^{\widetilde{C}})^2(\beta^{\top} \Sigma -\sigma^{\top}_{\Pi})(\Sigma^{\top}\beta  -\sigma_{\Pi}) + \frac{1}{2}\frac{\partial^2 V}{\partial \pi^2}\pi^2\sigma^{\top}_{\Pi}\sigma_{\Pi} +\frac{1}{2} \text{Tr}\bigg( \Sigma^{\top}_X \frac{\partial^2 V}{\partial X^{\top} \partial X } \Sigma \bigg)\notag\\
& + w^{\widetilde{Y}}(\beta^{\top}_t\Sigma - \sigma^{\top}_{\Pi})\Sigma^{\top}_X \frac{\partial^2 V}{\partial w^{\widetilde{Y}} \partial X^{\top}}  + \frac{\partial^2 V}{\partial w^{\widetilde{Y}} \partial \pi}w^{\widetilde{Y}} \pi(\beta^{\top}\Sigma-\sigma^{\top}_{\Pi})\sigma_{\Pi} + \frac{\partial^2 V}{\partial \pi \partial X} \Sigma_X\sigma_{\Pi} \pi, \label{infi_generator_primal}
\end{align}
and the optimal strategies are
\begin{eqnarray}
    c^*_{1,t} &=& (U_1')^{-1} \left( \frac{1}{\kappa_1}\frac{\partial V}{\partial w^{\widetilde{Y}}},\Pi_t\right),\label{primary_opt_c1_1}\\
    c^*_{2,t} &=& (U_2')^{-1} \left( \frac{1}{\kappa_2}\frac{\partial V}{\partial w^{\widetilde{Y}}},\Pi_t\right),\label{primary_opt_c2_1}\\
    \beta^*_t &=& \frac{-(\Sigma^{\top})^{-1}}{w^{\widetilde{Y}} \frac{\partial^2 V}{(\partial w^{\widetilde{Y}})^2}}\left[\frac{\partial V}{\partial w^{\widetilde{Y}}}(\Lambda_t - \sigma_{\Pi}) + \frac{\partial^2 V}{\partial w^{\widetilde{Y}} \partial \pi}\Pi_t\sigma_{\Pi}+\Sigma^{\top}_X\frac{\partial^2 V}{\partial w^{\widetilde{Y}} \partial X^{\top}}\right]+(\Sigma^{\top})^{-1}\sigma_{\Pi},\label{primary_opt_beta_1}\\
    I^*_t &=& \lambda_{x+t} \left(\frac{\partial \Phi_1}{\partial w}\right)^{-1}\left(t,\frac{1}{\kappa_2}\frac{\partial V}{\partial w^{\widetilde{Y}}},\Pi_t,X_t\right)-\lambda_{x+t}W^{\widetilde{Y}}_t+\lambda_{x+t}\widetilde{Y}(t,X_t) .\label{primary_opt_I_1}
\end{eqnarray}
where $\frac{\partial \Phi_1}{\partial w}$ is the partial derivative with the second variable of $\Phi_1(t,w,\pi)$. 

\subsection{Explicit solutions under CRRA utility}
Inspired by \cite{miao2013economic}, \cite{wei2023optimal}, and \cite{donnelly2024money}, we restrict consumption utilities to the following form
\begin{equation*}
    U_i(c_t,\Pi_t) = \frac{1}{1-\gamma_i} [c^{1-\theta_i}_t(\Pi_tc_t)^{\theta_i}]^{1-\gamma_i}, ~i = 1,2,
\end{equation*}
where $c_t$ and $\Pi_t c_t$ represent the real and nominal consumptions, respectively, and $\theta_i \in [0,1]$ represents the degree of money illusion. For tractability, we assume \( \gamma_1 = \gamma_2 = \gamma \) and \( \theta_1 = \theta_2 = \theta \) to obtain explicit solutions. When $\theta = 0$, the breadwinner is non-illusioned and only considers real consumption. When $\theta = 1$, the breadwinner is fully-illusioned and focuses solely on nominal consumption. As $\theta$ increases from 0 to 1, the breadwinner increasingly values the nominal value and ignores inflation risk.

Under this utility specification, we derive explicit solutions for the HJB equations \eqref{first_HJB_1}-\eqref{primary_adj_HJB_1}. These explicit solutions provide valuable insights into optimal strategies considering varying degrees of money illusion.

\begin{proposition}\label{prop1}
The candidate solution to HJB \eqref{first_HJB_1} is given by
\begin{eqnarray*}
    G_1(t,W^R_t,\Pi_t,X_t) &=& \frac{1}{1-\gamma} (W^R_t)^{1-\gamma} \Pi^{\theta(1-\gamma)}_t f_1(t,X_t)^{\gamma},
\end{eqnarray*}
where 
\begin{eqnarray*}
    f(X_t,\tau) &=& \exp\left( \Gamma_0(\tau)+\Gamma^{\top}_1(\tau)X_t + \frac{1}{2}X^{\top}_t\Gamma_2(\tau)X_t\right),\\
    f_1(t,X_t)  &=& \int_t^T e^{-\frac{\delta}{\gamma}(s-t)} f(X_t,s-t)ds,
\end{eqnarray*}
Functions $\Gamma_0(\tau) \in \mathbb{R}$, $\Gamma_1(\tau) \in \mathbb{R}^2$, and $\Gamma_2(\tau) \in \mathbb{R}^2 \times \mathbb{R}^2$ follow the ODE system
\begin{eqnarray}
   &&\frac{\partial \Gamma_2(\tau)}{\partial \tau}  - \Gamma_2(\tau) Z_2 \Gamma_2(\tau)  - Z^{\top}_1\Gamma_2(\tau)- \Gamma_2(\tau) Z_1 - Z_0=0,  ~\Gamma_2(0)=0,\label{Gamma2_tau}\\
   &&\frac{\partial \Gamma_1(\tau)}{\partial \tau}  - \Gamma_2(\tau)B_2\Gamma_1(\tau) - \Gamma_2(\tau)B_{11} - B_{12}\Gamma_1(\tau) - B_0=0, ~\Gamma_1(0)=0,\label{Gamma1_tau}\\
   &&\frac{\partial \Gamma_0(\tau)}{\partial \tau} - \Gamma^{\top}_1(\tau) D_2 \Gamma_1(\tau) - \Gamma^{\top}_1(\tau)D_1 - \frac{1}{2} \text{Tr} \{ \Sigma^{\top}_X \Gamma_2(\tau) \Sigma_X \} - D_0=0, ~\Gamma_0(0)=0,\label{Gamma0_tau}
\end{eqnarray}
in which
\begin{eqnarray*}
  &&Z_2 = \Sigma_X \Sigma^{\top}_X, Z_1 = \frac{1-\gamma}{\gamma}  \Sigma_X \Lambda_1 - K_X, Z_0 = \frac{1-\gamma}{\gamma^2}\Lambda^{\top}_1\Lambda_1,\\
  &&B_2 = Z_2, B_{11} =  \Sigma_X \left[\frac{1-\gamma}{\gamma}(\Lambda_0 - \sigma_{\Pi})+\theta \frac{1-\gamma}{\gamma}\sigma_{\Pi}\right], B_{12} = Z^{\top}_1, \\
  &&B_0 = \frac{1-\gamma}{\gamma}e_1 + \frac{1-\gamma}{\gamma}\theta e_2 + \frac{1-\gamma}{\gamma^2} \Lambda^{\top}_1(\Lambda_0 - \sigma_{\Pi}) + \theta \left( \frac{1-\gamma}{\gamma} \right)^2 \Lambda^{\top}_1 \sigma_{\Pi} ,\\
  &&D_2 = \frac{1}{2}Z_2, D_1 = B_{11}, D_0 = \frac{1-\gamma}{\gamma} (\delta_r+\theta \delta_{\pi}) + \frac{1-\gamma}{2\gamma^2}(\Lambda^{\top}_0-\sigma^{\top}_{\Pi})(\Lambda_0-\sigma_{\Pi})\\
  &&+\frac{1}{2}\theta \frac{1-\gamma}{\gamma} \left( \theta \frac{1-\gamma}{\gamma} -1 \right)\sigma^{\top}_{\Pi}\sigma_{\Pi} + \theta \left(\frac{1-\gamma}{\gamma}\right)^2(\Lambda^{\top}_0-\sigma^{\top}_{\Pi})\sigma_{\Pi}.
\end{eqnarray*}
The candidate strategies are given by
\begin{eqnarray*}
   &&c^*_{2,t} = W^R_t f_1(t,X_t)^{-1}, \\
   &&\alpha^*_t = \frac{1}{\gamma} (\Sigma^{\top})^{-1} \left[(\Lambda_t - \sigma_{\Pi}) + \theta(1-\gamma)\sigma_{\Pi} + \Sigma^{\top}_X\frac{\gamma}{f_1(t,X_t)}\frac{\partial f_1(t,X_t)}{\partial X^{\top}}\right]+(\Sigma^{\top})^{-1}\sigma_{\Pi}.
\end{eqnarray*}
\end{proposition}

\begin{proposition}\label{prop2}
The candidate solution to HJB \eqref{secondary_HJB_1} is given by
\begin{eqnarray*}
    &&G_2(t,W^R_t,\Pi_t,X_t) = \frac{1}{1-\gamma} (W^R_t)^{1-\gamma} \Pi^{\theta(1-\gamma)}_t f_2(t,X_t)^{\gamma},
\end{eqnarray*}
where 
\begin{eqnarray*}
    f(X_t,\tau) &=& \exp\left( \Gamma_0(\tau)+\Gamma^{\top}_1(\tau)X_t + \frac{1}{2}X^{\top}_t\Gamma_2(\tau)X_t\right),\\
    f_1(t,X_t)  &=& \int_t^T e^{-\frac{\delta}{\gamma}(s-t)} f(X_t,s-t)ds,\\
    f_2(t,X_t) &=& \kappa_1^{\frac{1}{\gamma}}\int_t^{T}{_{s-t}p_{x+t}} e^{-\frac{\delta}{\gamma}(s-t)}f(X_t,s-t)ds + \kappa_2^{\frac{1}{\gamma}}f_1(t,X_t).
\end{eqnarray*}
%Functions $g^R_t$ comes from \eqref{income_process} and $P^R(X_t,t,s)$ is the inflation-linked bond in \eqref{inflation_linked_bond}. 
Functions $\Gamma_0(\tau) \in \mathbb{R}$, $\Gamma_1(\tau) \in \mathbb{R}^2$, and $\Gamma_2(\tau) \in \mathbb{R}^2 \times \mathbb{R}^2$ follow the ODE system \eqref{Gamma2_tau}-\eqref{Gamma0_tau}. The optimal strategies are given by
\begin{eqnarray*}
   c^*_{1,t} &=& \kappa^{\frac{1}{\gamma}}_1 W^R_t f_2(t,X_t)^{-1}, \\
   c^*_{2,t} &=& \kappa^{\frac{1}{\gamma}}_2 W^R_t f_2(t,X_t)^{-1}, \\
   \alpha^*_t &=& \frac{1}{\gamma}(\Sigma^{\top})^{-1} \left[(\Lambda_t-\sigma_{\Pi}) + \theta(1-\gamma)\sigma_{\Pi}+\Sigma^{\top}_X\frac{\gamma}{f_2(t,X_t)}\frac{\partial f_2(t,X_t)}{\partial X^{\top}}\right]+(\Sigma^{\top})^{-1}\sigma_{\Pi},\\
   I^*_t &=& \lambda_{x+t} W^R_t\left[\kappa^{\frac{1}{\gamma}}_2\frac{f_1(t,X_t)}{f_2(t,X_t)}-1\right].
\end{eqnarray*}
\end{proposition}

\begin{proposition}\label{prop3}
The candidate solution to HJB \eqref{primary_adj_HJB_1} is given by
\begin{equation}
    G(t,W^{\widetilde{Y}}_t,\Pi_t,X_t) = \frac{1}{1-\gamma} (W^{\widetilde{Y}}_t)^{1-\gamma} \Pi^{\theta(1-\gamma)}_t f_2(t,X_t)^{\gamma},\label{adj_G}
\end{equation}
where 
\begin{eqnarray*}
    f(X_t,\tau) &=& \exp\left( \Gamma_0(\tau)+\Gamma^{\top}_1(\tau)X_t + \frac{1}{2}X^{\top}_t\Gamma_2(\tau)X_t\right),\\
    f_1(t,X_t)  &=& \int_t^T e^{-\frac{\delta}{\gamma}(s-t)} f(X_t,s-t)ds,\\
    f_2(t,X_t) &=& \kappa_1^{\frac{1}{\gamma}}\int_t^{T}{_{s-t}p_{x+t}} e^{-\frac{\delta}{\gamma}(s-t)}f(X_t,s-t)ds + \kappa_2^{\frac{1}{\gamma}}f_1(t,X_t),
    %f_3(t,X_t)&=&\int_t^{T_R} {_{s-t}p_{x+t}}e^{\int_t^s g^R_u du}P^R(X_t,t,s)ds.
\end{eqnarray*}
%Functions $g^R_t$ comes from \eqref{income_process} and $P^R(X_t,t,s)$ is the inflation-linked bond in \eqref{inflation_linked_bond}. 
Functions $\Gamma_0(\tau) \in \mathbb{R}$, $\Gamma_1(\tau) \in \mathbb{R}^2$, and $\Gamma_2(\tau) \in \mathbb{R}^2 \times \mathbb{R}^2$ follow the ODE system \eqref{Gamma2_tau}-\eqref{Gamma0_tau}. The optimal strategies are given by
\begin{eqnarray}
   &&c^*_{1,t} = \kappa^{\frac{1}{\gamma}}_1 W^{\widetilde{Y}}_t f_2(t,X_t)^{-1}, \label{primary_opt_c1_2}\\
   &&c^*_{2,t} =\kappa^{\frac{1}{\gamma}}_2  W^{\widetilde{Y}}_t  f_2(t,X_t)^{-1}, \label{primary_opt_c2_2}\\
   &&\beta^*_t = \frac{(\Sigma^{\top})^{-1}}{\gamma}\left[\Lambda_t -\sigma_{\Pi} + \theta(1-\gamma)\sigma_{\Pi}+ \Sigma^{\top}_X \frac{\gamma}{f_2(t,X_t)}\frac{\partial f_2(t,X_t)}{\partial X^{\top}}\right] + (\Sigma^{\top})^{-1}\sigma_{\Pi}\label{primary_opt_beta_2}\\
   &&I^*_t = \lambda_{x+t} \left\{\kappa^{\frac{1}{\gamma}}_2W^{\widetilde{Y}}_t\frac{f_1(t,X_t)}{f_2(t,X_t)}-W^R_t\right\}.\label{primary_opt_I_2}
\end{eqnarray}
Since $\beta^{\top}_t = [W^R_t \alpha^{\top}_t + \widetilde{Y}(t,X_t)\xi^{\top}_t]/W^{\widetilde{Y}}_t$, the corresponding trading strategy is 
\begin{align*}
   &\alpha^*_t = \frac{W^R_t + \widetilde{Y}(t,X_t)}{\gamma W^R_t}(\Sigma^{\top})^{-1} \left\{(\Lambda_t-\sigma_{\Pi}) + \theta(1-\gamma)\sigma_{\Pi}\right.\\
   &\left.+\Sigma^{\top}_X \left[-\gamma[W^R_t+\widetilde{Y}(t,X_t)]^{-1} \frac{\partial \widetilde{Y}(t,X_t)}{\partial X^{\top}}+\frac{\gamma}{f_2(t,X_t)}\frac{\partial f_2(t,X_t)}{\partial X^{\top}}\right]\right\}+(\Sigma^{\top})^{-1}\sigma_{\Pi}.
\end{align*}
\end{proposition}

\subsection{The global existence and verification theorem}
Several linear ODEs, including \eqref{Gamma1_tau} and \eqref{Gamma0_tau}, which are linear ODEs, determine the candidate solutions. These equations have unique solutions that exist globally \citep[see Theorem 1.1.1. in][]{abou2012matrix}. However, the ODE \eqref{Gamma2_tau} is a Hermitian matrix Riccati differential equation (HRDE), which requires special treatment for its existence. The HRDE can be represented as a matrix in the following way
\begin{eqnarray}\label{HRDE}
	\frac{\partial \Gamma_2(\tau)}{\partial \tau}  &=& (\widetilde{I}_2, \Gamma_2(\tau))JH(\tau)\begin{pmatrix}
		\widetilde{I}_2 \\
		\Gamma_2(\tau)
	\end{pmatrix}	:= \mathcal{H}(\Gamma_2;H), \tau\in[0,T],
\end{eqnarray}
where $\widetilde{I}_2$ is the 2nd-order identity matrix, and
\begin{equation*}
	J := \begin{pmatrix}
		0_{2\times 2} & \widetilde{I}_2 \\
		-\widetilde{I}_2 & 0_{2 \times 2}
	\end{pmatrix}
	\in \mathbb{R}^2 \times \mathbb{R}^2,
\text{ and }
	H := \begin{pmatrix}
		-Z_1 & -Z_2 \\
		Z_0 & Z^{\top}_1
	\end{pmatrix}
	\in \mathbb{R}^2 \times \mathbb{R}^2, \label{H_matrix}
\end{equation*}
which is called the Hamiltonian matrix. The global existence of the HRDE \eqref{HRDE} is heavily influenced by the relative risk aversion coefficient $\gamma$, which is also a factor for the verification theorem. Inspired by \cite{honda2011verification}, we divide the proofs in this subsection into two cases $\gamma>1$ and $0<\gamma<1$.

\begin{proposition}\label{global_exist_gamma_great_1}
	For $\gamma>1$, define the admissible set as
	\begin{equation*}\label{adset1_gamma_great_1}
		\mathcal{A}_{\gamma}(0,T_R):=
		\left\{
		\begin{array}{c|c}
			&  \text{$\beta(t,X_t):[0,T_R]\times \mathbb{R}^2\rightarrow \mathbb{R}^4$ }\\
			(\beta, I, c_1, c_2)  &  \text{grows linearly with respect to $X_t$, }\\
			&  \text{and SDE \eqref{WY_1} has a unique strong solution.}
		\end{array}
		\right\}.
	\end{equation*} 
    If $\Sigma_X\Sigma^{\top}_X>0$ and $\Lambda^{\top}_1\Lambda_1>0$, then the candidate solution $G(t,W^{\widetilde{Y}}_t,\Pi_t,X_t)$ in Proposition \ref{prop3} exists in $[0,T]$ and equals the primal value function $V(t,W^{\widetilde{Y}}_t,\Pi_t,X_t)$. The strategies $(\beta^*,I^*,c^*_1, c^*_2)$ given by \eqref{primary_opt_c1_2} to \eqref{primary_opt_I_2} are the optimal strategies. For matrices, ``$>$'' (``$<$'') indicates positive (negative) definite.
\end{proposition}

For candidate solution $G_1$, define its admissible set as
\begin{equation*}\label{adset2_gamma_great_1}
		\mathcal{A}^{(1)}_{\gamma}(0,T):=
		\left\{
		\begin{array}{c|c}
			&  \text{$\alpha(t,X_t):[0,T]\times \mathbb{R}^2\rightarrow \mathbb{R}^4$ }\\
			(\alpha, c_2)  &  \text{grows linearly with respect to $X_t$, }\\
			&  \text{and SDE \eqref{real_wealth_1} has a unique strong solution.}
		\end{array}
		\right\}.
\end{equation*} 

For candidate solution $G_2$, define its admissible set as
\begin{equation*}\label{adset3_gamma_great_1}
		\mathcal{A}^{(2)}_{\gamma}(T_R,T):=
		\left\{
		\begin{array}{c|c}
			&  \text{$\alpha(t,X_t):[T_R,T]\times \mathbb{R}^2\rightarrow \mathbb{R}^4$ }\\
			(\alpha, I, c_1, c_2)  &  \text{grows linearly with respect to $X_t$, }\\
			&  \text{and SDE \eqref{real_wealth_1} has a unique strong solution.}
		\end{array}
		\right\}.
\end{equation*} 

Notice $G_1$ is the special case of $G$ restricted to $\mathcal{A}^{(1)}_{\gamma}(0,T)$ when $\widetilde{Y}(t,X_t)\equiv 0$, $\kappa_1=0$, and $\kappa_2=1$. Similarly, $G_2$ is the special case of $G$ restricted to $\mathcal{A}^{(2)}_{\gamma}(T_R,T)$ when $\widetilde{Y}(t,X_t)\equiv 0$. The same approach used in Proposition \ref{global_exist_gamma_great_1} can be applied to prove the global existence and verification theorems for $G_1$ and $G_2$.

For $0 < \gamma < 1$, the existence of \eqref{Gamma2_tau} can be established using Radon's Lemma with additional conditions. Let $(Q,P)^{\top}$ represent a solution to the linear system of differential equations
\begin{eqnarray}\label{ode_PQ}
	\frac{d}{d\tau}
	\begin{pmatrix}
		Q(\tau)\\
		P(\tau)
	\end{pmatrix}
	= H
	\begin{pmatrix}
		Q(\tau)\\
		P(\tau)
	\end{pmatrix}, Q(0) = \widetilde{I}_{2}, P(0) = \Gamma_2(0)Q(0) = 0.
\end{eqnarray}
According to Radon's Lemma \citep[see Theorem 3.1.1 in][]{abou2012matrix}, the solution to \eqref{Gamma2_tau} can be expressed as $\Gamma_2(\tau) = P(\tau)/Q(\tau)$. Next, we only need $\Gamma_2(\tau)<0$ to guarantee the candidate solution's global existence. For tractability, we adopt the assumption by \cite{abou2012matrix} that $H$ is diagonalizable. This means that there exists a 4-dimensional basis of eigenvectors
\begin{equation*}
	v_1, ..., v_4 \in \mathbb{C}^4,
\end{equation*}
where $\mathbb{C}^4$ denotes the complex vector space of $4 \times 1$ complex vectors. The corresponding eigenvalues are $\lambda_1, ..., \lambda_4$ sorted by their real parts
\begin{equation*}
	\mathcal{R}(\lambda_1) \leq \mathcal{R}(\lambda_2) \leq \mathcal{R}(\lambda_3) \leq \mathcal{R}(\lambda_{4}).
\end{equation*}
Let $V = (v_1,..., v_4) \in \mathbb{C}^{4\times 4}$, where $\mathbb{C}^{4\times 4}$ denotes the complex vector space of $4 \times 4$ complex matrices, then the solution to \eqref{ode_PQ} can be expressed as
\begin{equation*}\label{sol_ode_PQ}
	\begin{pmatrix}
		Q(\tau)\\
		P(\tau)
	\end{pmatrix}
	=V e^{\Delta \tau} V^{-1} \begin{pmatrix}
		Q(0)\\
		P(0)
	\end{pmatrix}
	=V e^{\Delta \tau} V^{-1} \begin{pmatrix}
		\widetilde{I}_{2}\\
		0
	\end{pmatrix},
\end{equation*}
where $\Delta := V^{-1}HV = \text{diag}(\lambda_1, ..., \lambda_4)$. Furthermore, define
\begin{equation}\label{f_lambda}
	f_{\lambda}(\lambda) = |\lambda \widetilde{I}_4 - H| = \lambda^4 + b\lambda^3 + c\lambda^2 + d\lambda + j,
\end{equation}
we can finally prove the following proposition of global existence and verification.
\begin{proposition}\label{global_exist_gamma_01}
	For $0<\gamma<1$, define the admissible set as
	\begin{equation*}\label{adset1_gamma_01}
		\mathcal{A}_{\gamma}(0,T_R):=
		\left\{
		\begin{array}{c|c}
			(\beta, I, c_1, c_2)  &  \text{$(\beta,I,c_1,c_2)$ ~\text{such that $W^{\widetilde{Y}}_t>0$, }}\\
			&  \text{and SDE \eqref{WY_1} has a unique strong solution.}
		\end{array}
		\right\}.
	\end{equation*}
	If
	\begin{equation}\label{root_condition}
		\widetilde{\Delta}>0,~ q<0, ~ s<\frac{q^2}{4},
	\end{equation}
	\begin{equation}\label{det_condition}
		\text{det}|Q(\tau)|\neq 0 ~\text{and}~ P(\tau)/Q(\tau)<0 ~\text{for}~ \forall \tau \in (0,T],
	\end{equation}
	then the candidate solution $G(t,W^{\widetilde{Y}}_t,\Pi_t,X_t)$ in Proposition \ref{prop3} exists in $[0,T]$ and equals the primal value function $V(t,W^{\widetilde{Y}}_t,\Pi_t,X_t)$. Moreover, the strategy $(\beta^*,I^*,c^*_1,c^*_2)$ given by \eqref{primary_opt_c1_2}-\eqref{primary_opt_I_2} is the optimal portfolio and insurance strategy. The expressions of $\widetilde{\Delta}$, $q$, and $s$ are given in Appendix \ref{appendix3}. 
\end{proposition}
For candidate solution $G_1$, define its admissible set as
\begin{equation*}\label{adset2_gamma_01}
		\mathcal{A}^{(1)}_{\gamma}(0,T):=
		\left\{
		\begin{array}{c|c}
			(\alpha,c_2) &  \text{$(\alpha,c_2)$ ~\text{such that $W^R_t>0$, }}\\
			&  \text{and SDE \eqref{real_wealth_1} has a unique strong solution.}
		\end{array}
		\right\}.
\end{equation*} 

For candidate solution $G_2$, define its admissible set as
\begin{equation*}\label{adset3_gamma_01}
		\mathcal{A}^{(2)}_{\gamma}(T_R,T):=
		\left\{
		\begin{array}{c|c}
			(\alpha, I, c_1, c_2)  & \text{$(\alpha, I, c_1, c_2)$ ~\text{such that $W^R_t>0$, }}\\
			&  \text{and SDE \eqref{real_wealth_1} has a unique strong solution.}
		\end{array}
		\right\}.
\end{equation*} 

Following the discussions after Proposition \ref{global_exist_gamma_great_1}, one can apply the approach in Proposition \ref{global_exist_gamma_01} to prove $G_1$ and $G_2$'s global existences and verification theorems.

\section{Numerical results}\label{numerical_research}
\subsection{Model calibration}
According to the parameter settings in \cite{huang2008portfolio}, we consider a breadwinner who is 35 years old at the initial time and retires at the age of 65. The family ceases making investment decisions at the breadwinner's age of 95, so $T_R = 30$ and $T = 60$. The breadwinner allocates wealth among 3-year nominal bonds, 10-year nominal bonds, 10-year inflation-linked bonds, the equity index, and cash ($T_1=3, T_2=T_3=10$) and also purchases life insurance. Following \cite{koijen2011optimal}, we suppose that the growth rate $g^R_t$ in the real income \eqref{income_process} is given by
\begin{equation*}\label{g^R_t}
	g^R_t = 0.1682 - 0.00646(45+t) + 0.00006(45+t)^2, 
\end{equation*}
which corresponds to an individual with a high school education in the estimates of \cite{cocco2005consumption} and \cite{munk2010dynamic}. We assume the breadwinner's force of mortality is subject to the Gompertz law
\begin{equation*}
	\lambda_{x+t} = \frac{1}{9.5} e^{\frac{x+t-86.3}{9.5}}, ~x=35,
\end{equation*}
and set other base model parameters as 
\begin{eqnarray*}
	&&\delta = 0.10,~~ W_0 = 35.00,~~ Y_0 = 25.00.
\end{eqnarray*} 

We utilize monthly data from the U.S. financial market, spanning from June 1961 to December 2023. To estimate the parameters, we employ zero-coupon nominal yields from Gurkaynak et al. (2007), comprising eight different maturities: three months, six months, one year, two years, three years, five years, seven years, and ten years. The realized inflation index is obtained from the CRSP's Consumer Price Index for All Urban Consumers (CPI-U NSA index). Additionally, we utilize the equity index based on the CRSP's value-weighted NYSE/Amex/Nasdaq index, which includes dividend payments.

We implement a Kalman filter algorithm to estimate the two factors and model parameters, detailed in Appendix \ref{appendix8}. The results are presented in Table \ref{financial_market_estimate_table} and Figure \ref{fig_estimated_r_pi}. Aligned with \cite{koijen2011optimal}, we observe that $\kappa_1 > \kappa_2$, which indicates that expected inflation is more persistent than the real short rate.  In terms of innovations, we detect a negative correlation between the real short rate and expected inflation ($\sigma_{2(1)}<0$). Regarding the equity index process, we find that the risk premium diminishes with the real short rate and expected inflation ($\mu_{1(1)}, \mu_{1(2)} < 0$). Moreover, the unconditional price of risk, $\Lambda_0$, is negative for the real short rate and expected inflation but positive for the equity index. Notably, all the parameters in the conditional price of risk, $\Lambda_1$, are negative, implying that the price of risk decreases with two factors $X_t$. Figure \ref{fig_estimated_r_pi} depicts the estimated short rates and expected inflation.

\begin{table}[htbp]
	\caption{\textbf{Estimation results for the financial market}} % title name of the table
	\smallskip
	\centering % centering table
	\begin{tabular}{c r c r c r}\label{financial_market_estimate_table}
		Parameter     &Estimate &Parameter     &Estimate
        &Parameter     &Estimate\\
		\hline
		\multicolumn{6}{l}{Average short rate \& average expected inflation}\\
		 $\delta_r$           &0.01254       &$\delta_R$           &0.05120  &$\delta_{\pi^e}$     &0.03831  \\
		\multicolumn{6}{l}{Two-factor process}\\
		 $\kappa_1$           &0.61921       &$\kappa_2$           &0.18894  &$\sigma_{1(1)}$      &0.02209  \\  
		 $\sigma_{2(1)}$      &-0.00673      &$\sigma_{2(2)}$      &0.01408  &                     &         \\
		\multicolumn{6}{l}{Realized inflation process}\\
		 $\sigma_{\Pi(1)}$    &0.00042       &$\sigma_{\Pi(2)}$    &0.00207  &$\sigma_{\Pi(3)}$    &0.01363  \\
	    \multicolumn{6}{l}{Equity index process}\\    
	     $\mu_0$              &0.04600       &$\mu_{1(1)}$         &-1.97000 &$\mu_{1(2)}$         &-1.41000  \\    $\sigma_{S(1)}$      &-0.01974     & $\sigma_{S(2)}$     &-0.01785 &$\sigma_{S(3)}$      &-0.00793 \\ 
		 $\sigma_{S(4)}$      &0.15410       &                     &         & &\\
	    \multicolumn{6}{l}{Prices of risk of real short rate, inflation, and equity}\\
	     $\Lambda_{0(1)}$     &0.00487      &$\Lambda_{0(2)}$     &-0.17007 & $\Lambda_{0(4)}$     &0.27943 \\      $\Lambda_{1(1,1)}$   &-9.92002    &$\Lambda_{1(2,2)}$   &-9.98001     &$\Lambda_{1(4,1)}$   & -14.05465\\
		 $\Lambda_{1(4,2)}$   & -10.30593   &                     &          &&\\
		 \hline
	\end{tabular}
	\begin{flushleft}
		{\small
			The parameters in the table are annualized. $\Lambda_0(1)$, $\Lambda_0(2)$, $\Lambda_1(4,1)$, and $\Lambda_1(4,2)$ can be obtained by solving three equations: $\delta_R = \delta_r + \delta_{\pi^e} - \sigma^{\top}_{\Pi}\Lambda_0$, $\sigma^{\top}_S\Lambda_0=\mu_0$, $\sigma^{\top}_S\Lambda_1=\mu_1$. So, there are 21 parameters in total to be estimated. 
		}
	\end{flushleft}
\end{table}
\begin{figure}[htbp]	
	{\flushleft  
	  \includegraphics[width=7in,height=2.5in]{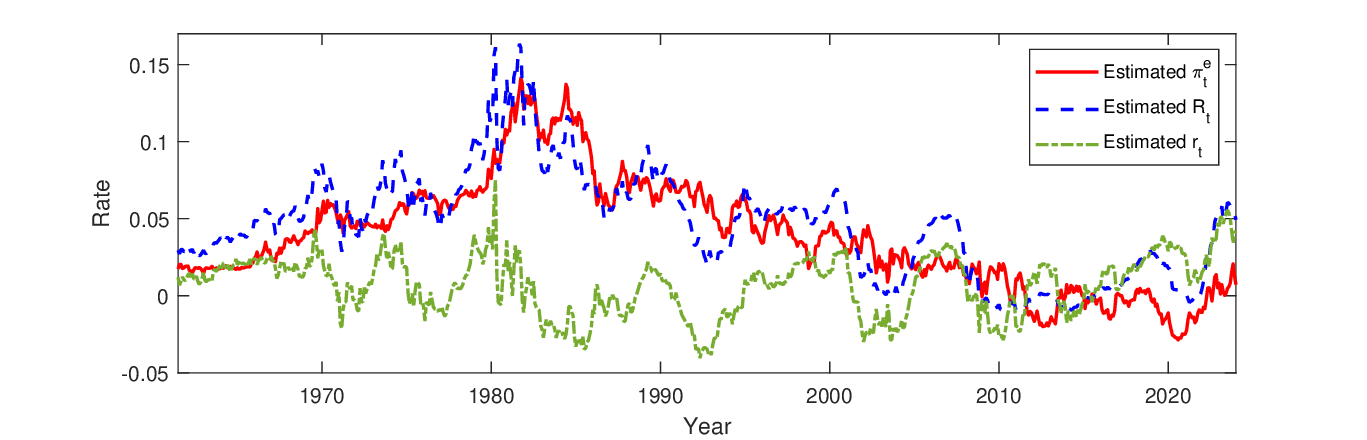}
	\caption{Estimated short rate and expected inflation process. The solid line is the estimated expected inflation $\pi^e_t$. The dashed line is the estimated nominal short rate $R_t$. The dash-dotted line is the estimated real short rate $r_t$.}\label{fig_estimated_r_pi}
    }
\end{figure}

\subsection{Sensitivity of optimal strategies to the age}\label{sensitivity_to_age}

In this subsection, we analyze how the optimal strategies change with age. We conduct Monte Carlo simulations involving 10,000,000 paths, using a time step of one year. The expected trading, consumption, and insurance strategies are shown in Figures  \ref{fig_beta_2D_gamma10:figures} - \ref{fig_insurance_2D_gamma5:figures}. Specifically, Figures \ref{fig_beta_2D_gamma10:figures} and \ref{fig_insurance_2D_gamma10:figures} correspond to a risk aversion coefficient of $\gamma=10$, while Figures \ref{fig_beta_2D_gamma5:figures} and \ref{fig_insurance_2D_gamma5:figures} correspond to $\gamma=5$.

To obtain more insightful observations, we decompose the optimal trading strategy $\beta^*$ in \eqref{primary_opt_beta_2} into the following three components
\begin{eqnarray}
	\beta^*_t =\underbrace{\frac{(\Sigma^{\top})^{-1}}{\gamma}\Lambda_t}_\text{standard myopic demand} +\underbrace{\frac{\gamma-1}{\gamma}(1-\theta)(\Sigma^{\top})^{-1}\sigma_{\Pi}}_\text{inflation hedging demand} +\underbrace{(\Sigma^{\top})^{-1}\Sigma^{\top}_X\frac{1}{f_2(t,X_t)}\frac{\partial f_2(t,X_t)}{\partial X^{\top}}}_\text{intertemporal hedging demand}. \label{adj_beta_decomposition}
\end{eqnarray}
Among three components, the standard myopic demand (SMD) characterizes the risk-return trade-off of the assets. The inflation hedging demand (IFHD) represents the family's aspiration to hedge against realized inflation $\Pi_t$ ($\sigma_{\Pi}$ is the volatility term of $\Pi_t$, as specified in \eqref{realized_inflation}). Lastly, the intertemporal hedging demand (ITHD), determined by the investment horizon, reflects the family's intent to hedge against potential changes in future investment opportunity sets.

We present the unconditional expectations of SMD and IFHD in Table \ref{Expected_smd_ifhd}. The expected SMD is a constant vector because $X_t$ in \eqref{Xt_process} follows a normal distribution $N(0,\Sigma_t)$, where $\Sigma_t = \int_0^t e^{-K_x(t-s)}\Sigma_X \Sigma^{\top}_X e^{-K^{\top}_X(t-s)} ds$. Moreover, $\text{IFHD}_4$ is zero as the fourth entry of $\sigma_{\Pi}$ is zero, which is implied by the assumption that  $(\sigma_1, \sigma_2, \sigma_{\Pi}, \sigma_S)^{\top}$ is lower triangular. 
\begin{table}[htbp]
	\caption{\textbf{Expected standard myopic demand and inflation hedging demand}} 
	\smallskip
	\centering % centering table
	\begin{tabular}{c c c c c}\label{Expected_smd_ifhd}
              & $\text{SMD}_1$ & $\text{SMD}_2$ & $\text{SMD}_3$ & $\text{SMD}_4$\\
     Values   &-0.293846  & 0.226313 &0.105499 &0.181331\\
     \hline
              & $\text{IFHD}_1$ & $\text{IFHD}_2$ & $\text{IFHD}_3$ & $\text{IFHD}_4$\\
     $\theta=0.0$ & -2.089113 & 0.735923 & 0.900000 & 0.000000\\
     $\theta=0.2$ & -1.671290 & 0.588739 & 0.720000 & 0.000000\\
     $\theta=0.4$ & -1.253468 & 0.441554 & 0.540000 & 0.000000\\
     $\theta=0.6$ & -0.835645 & 0.294369 & 0.360000 & 0.000000\\
     $\theta=0.8$ & -0.417823 & 0.147185 & 0.180000 & 0.000000\\
     $\theta=1.0$ &  0.000000 & 0.000000 & 0.000000 & 0.000000\\
     \hline
	\end{tabular}
 	\begin{flushleft}
		{\small
  $\text{SMD}_i$ is the $i$th entry of the standard myopic demand vector.
  $\text{IFHD}_i$ is the $i$th entry of the inflation hedging demand vector.}
	\end{flushleft}
 \label{Table_SMD_IFHD}
\end{table} 

We plot the expected optimal trading strategies $\beta^*$ in Figure \ref{fig_beta_2D_gamma10:figures}. A comparison between Figures \ref{fig_beta_2D_gamma10:figure1} and \ref{fig_beta_2D_gamma10:figure2} reveals that the breadwinner is likely to short the short-term nominal bond in favor of the high-risk premium from the long-term nominal bond. When $\theta$ goes to 1 (implying a higher valuation of the nominal value by the breadwinner), there is a decrease in shorting the short-term nominal bond and an increase in longing the long-term nominal bond. Figures \ref{fig_beta_2D_gamma10:figure3} and \ref{fig_beta_2D_gamma10:figure4} indicate that the breadwinner tends to long the inflation-linked bond and stocks. When $\theta$ goes to 1, the demand for inflation-linked bonds decreases, while the demand for stocks remains unchanged. The expected ITHD is also depicted in Figure \ref{fig_beta_2D_gamma10:figures}. The $\text{ITHD}_3$ and $\text{ITHD}_4$ are zero, given the third and fourth rows of $\Sigma_X$ are zero. Finally, we observe that ITHD predominantly influences the evolution of expected investment strategies.

Figure \ref{fig_insurance_2D_gamma10:figures} shows the expected optimal consumption and insurance strategies. It is noticeable that both the breadwinner and the family exhibit a higher consumption rate in their early years when they disregard inflation risk (i.e., the consumption pattern transitions from increasing to decreasing as $\theta$ approaches 1). From Figure \ref{fig_insurance_2D_gamma10:figure3}, we obtain three key observations for the insurance premium: (1) $E[I^*_t]$ change from positive to negative through time; (2) the positive range is significantly smaller than the negative range; (3) $E[I^*_t]$ moves upward as $\theta$ approaches 1. Observation (1) can be understood in the context of the opposing roles of life insurance and annuities. Life insurance protects a breadwinner's income before retirement, while an annuity serves as a source of ``income" after retirement. Consequently, to safeguard income, the breadwinner purchases life insurance ($I^*_t>0$) prior to retirement and switches to an annuity ($I^*_t<0$) approaching retirement (as interpreted in \cite{fischer1973life}, \cite{pirvu2012optimal}, and \cite{shen2016optimal}). Observation (2) is reflected in the maximum positive value in Figure \ref{fig_insurance_2D_gamma10:figure3} being 0.277k USD per year and the maximum negative value being 16.04k USD per year. This aligns with the real-world scenario where life insurance is considerably cheaper than an annuity. Despite the significant price gap, Figure \ref{fig_insurance_2D_gamma10:figure4} demonstrates that life insurance and annuities' expected payoffs (insurance face value) are within the same range. Observation (3) reveals that the demand for life insurance increases and the demand for annuity decreases when the family ignores inflation risk. In addition to life insurance demand, we also examined the breadwinner's bequest demand. After rearranging \eqref{primary_opt_I_2}, we define the bequest-wealth ratio as
\begin{equation}
    \frac{\kappa^{\frac{1}{\gamma}}_2f_1(t,X_t)}{f_2(t,X_t)} = \frac{W^R_t + \frac{I^*_t}{\lambda_{x+t}}}{W^{R}_t + \widetilde{Y}(t,X_t)}, \label{bequest_wealth_ratio}
\end{equation}
where the right-hand side is the ratio of the death benefit over the surplus process. We plot the expected curves of surplus process and wealth-bequest ratio in Figures \ref{fig_insurance_2D_gamma10:figure5} and \ref{fig_insurance_2D_gamma10:figure6}. The figures show that they both perform a hump shape and move downward when the inflation risk is ignored.

We also plot the optimal strategies under $\gamma=5$ in Figures \ref{fig_beta_2D_gamma5:figures} and \ref{fig_insurance_2D_gamma5:figures} as the family becomes more risk-seeking. Compared to the case of $\gamma = 10$, we find: (1) the family will short more short-term nominal bonds and long more long-term nominal bonds, inflation-linked bonds, and stocks. (2) Their demand for life insurance decreases while their annuity demand increases. (3) Their expected surplus process sharply increases, but their bequest-wealth ratio slightly declines.

In general, the whole family increases their life insurance demand and reduces their annuity demand when ignoring the inflation risk. This observation leads us to two key conclusions: (1) Money illusion might contribute to the annuity puzzle, which is a phenomenon of individuals not purchasing sufficient annuities to fund their retirement. (2) Promoting inflation education could increase retirees' voluntary purchase of annuities.

\begin{figure}[htbp]
    \centering
    \subfigure[3-year nominal bonds]
    {
	\includegraphics[width=3.3in,height=1.7in]{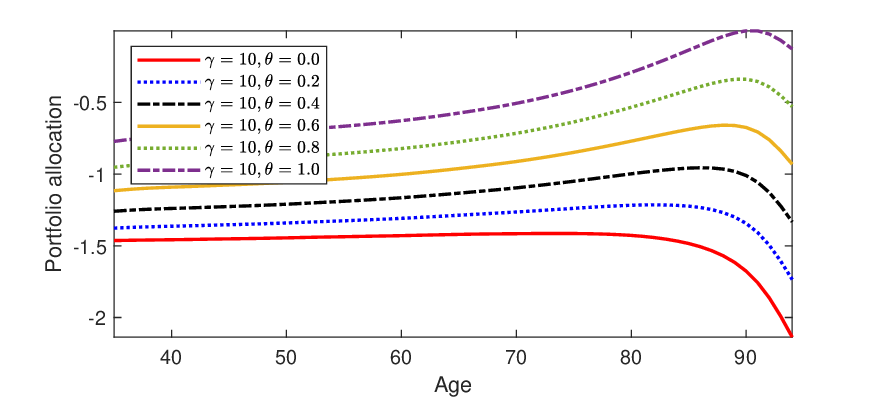}
        \label{fig_beta_2D_gamma10:figure1}
    }
    \hspace{-0.45in}
    \subfigure[10-year nominal bonds]
    {
	\includegraphics[width=3.3in,height=1.7in]{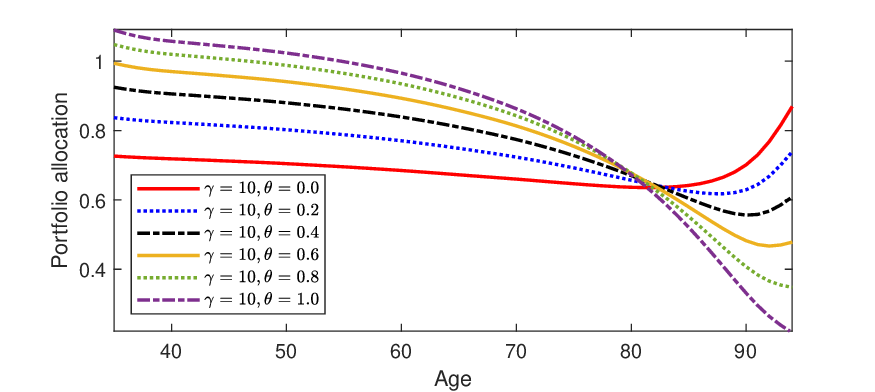}
        \label{fig_beta_2D_gamma10:figure2}
    }
    \\
    \subfigure[10-year inflation-linked bonds]
    {
        \includegraphics[width=3.3in,height=1.7in]{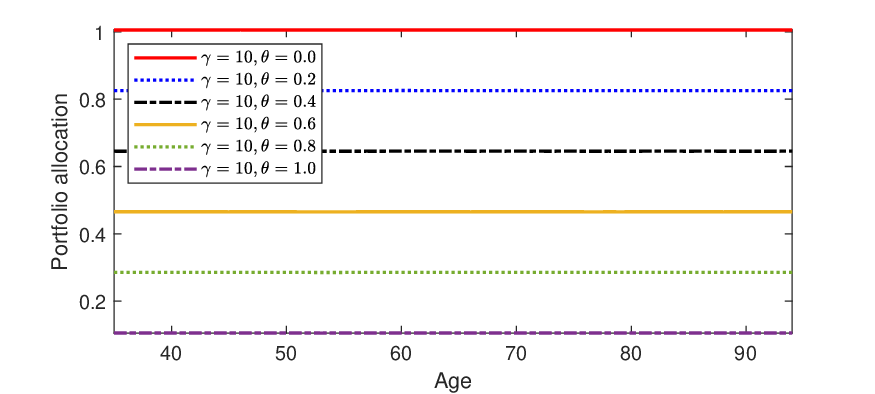}
        \label{fig_beta_2D_gamma10:figure3}
    }
    \hspace{-0.45in}
    \subfigure[Stocks]
    {
        \includegraphics[width=3.3in,height=1.7in]{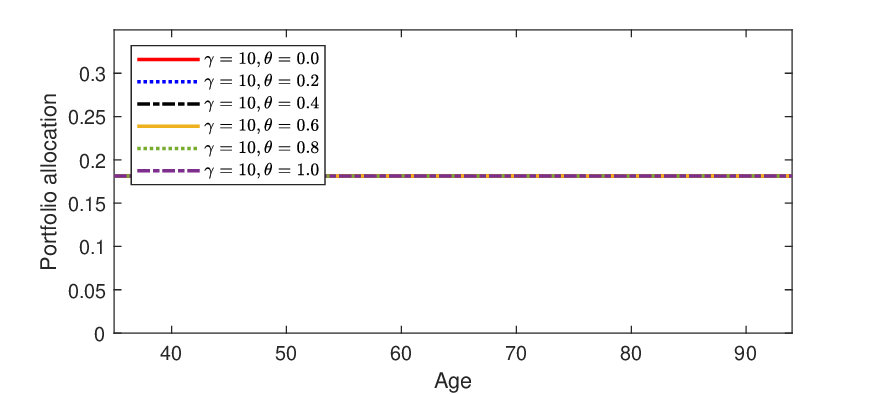}
        \label{fig_beta_2D_gamma10:figure4}
    }
    \subfigure[3-year nominal bonds (ITHD)]
    {
	\includegraphics[width=3.3in,height=1.7in]{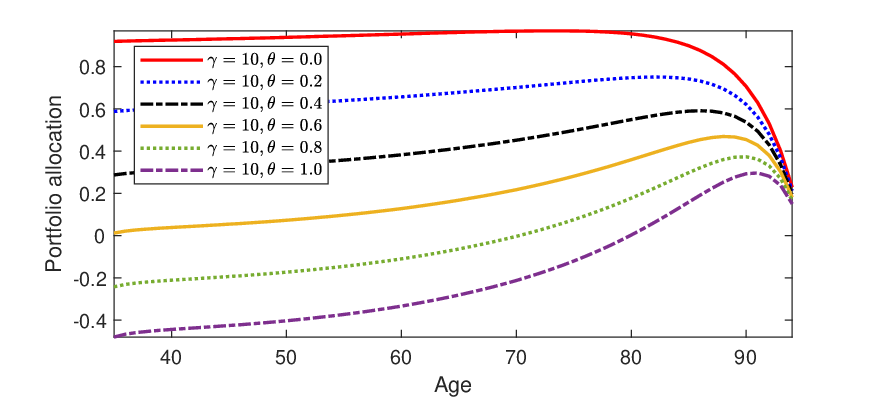}
        \label{fig_beta_2D_gamma10:figure5}
    }
    \hspace{-0.45in}
    \subfigure[10-year nominal bonds (ITHD)]
    {
	\includegraphics[width=3.3in,height=1.7in]{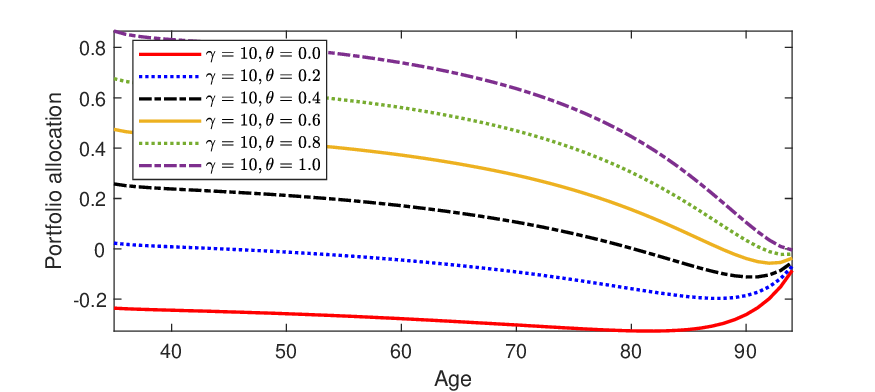}
        \label{fig_beta_2D_gamma10:figure6}
    }
    \\
    \subfigure[10-year inflation-linked bonds (ITHD)]
    {
        \includegraphics[width=3.3in,height=1.7in]{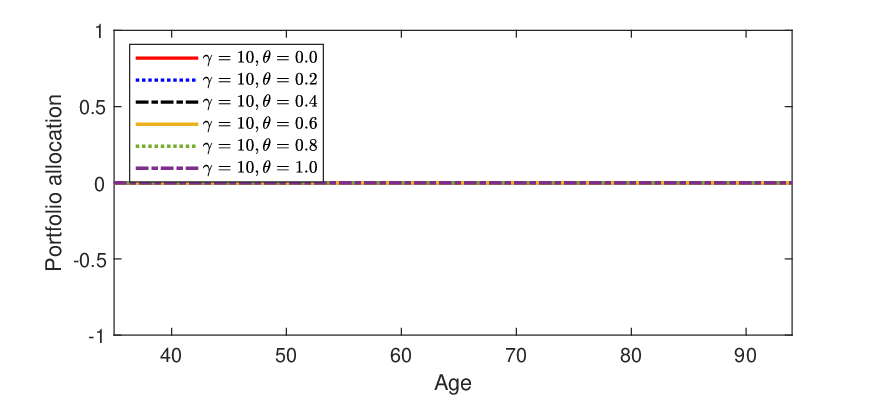}
        \label{fig_beta_2D_gamma10:figure7}
    }
    \hspace{-0.45in}
    \subfigure[Stocks (ITHD)]
    {
        \includegraphics[width=3.3in,height=1.7in]{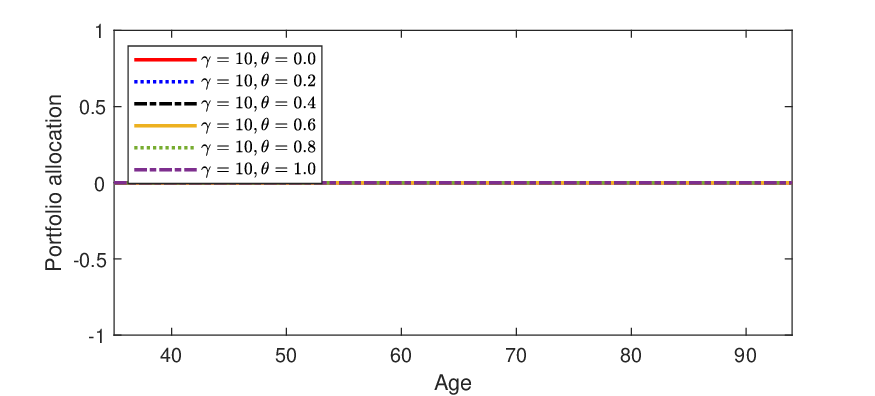}
        \label{fig_beta_2D_gamma10:figure8}
    }
    \caption{Expected optimal trading strategies for surplus process under $\gamma=10$. }
    \label{fig_beta_2D_gamma10:figures}
\end{figure}

\begin{figure}[htbp]
    \centering
    \subfigure[Breadwinner's consumption]
    {
	\includegraphics[width=3.0in,height=2.5in]{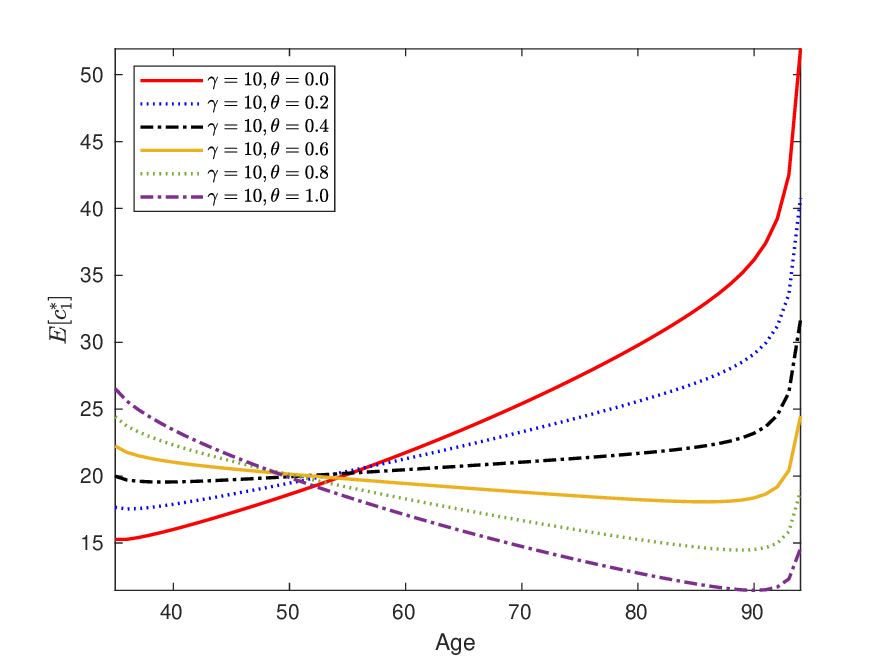}
        \label{fig_insurance_2D_gamma10:figure1}
    }
    \hspace{-0.45in}
    \subfigure[The rest of family's consumption]
    {
	\includegraphics[width=3.0in,height=2.5in]{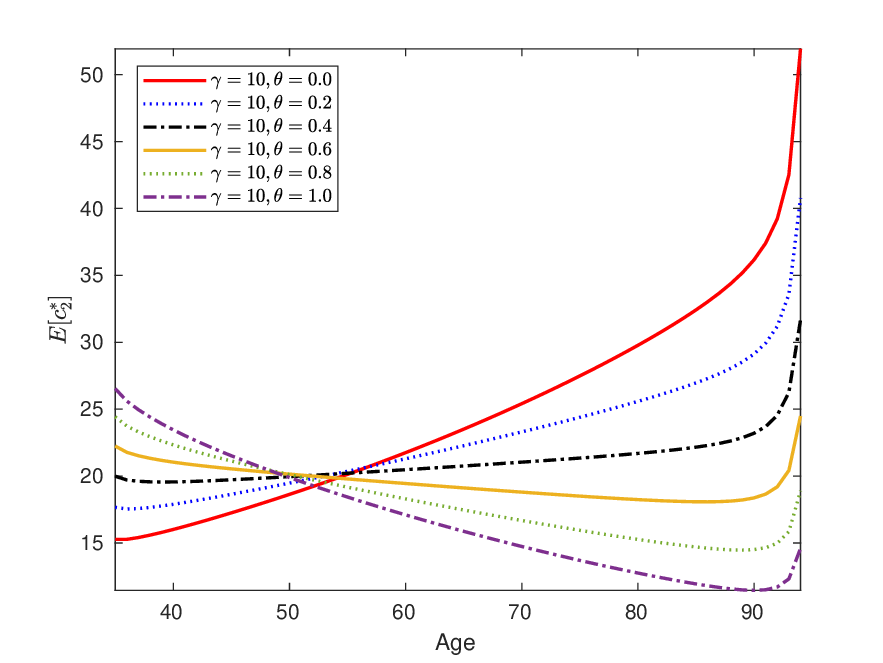}
        \label{fig_insurance_2D_gamma10:figure2}
    }
    \\
    \vspace{-0.10in}
    \subfigure[Insurance premium]
    {
        \includegraphics[width=3.0in,height=2.5in]{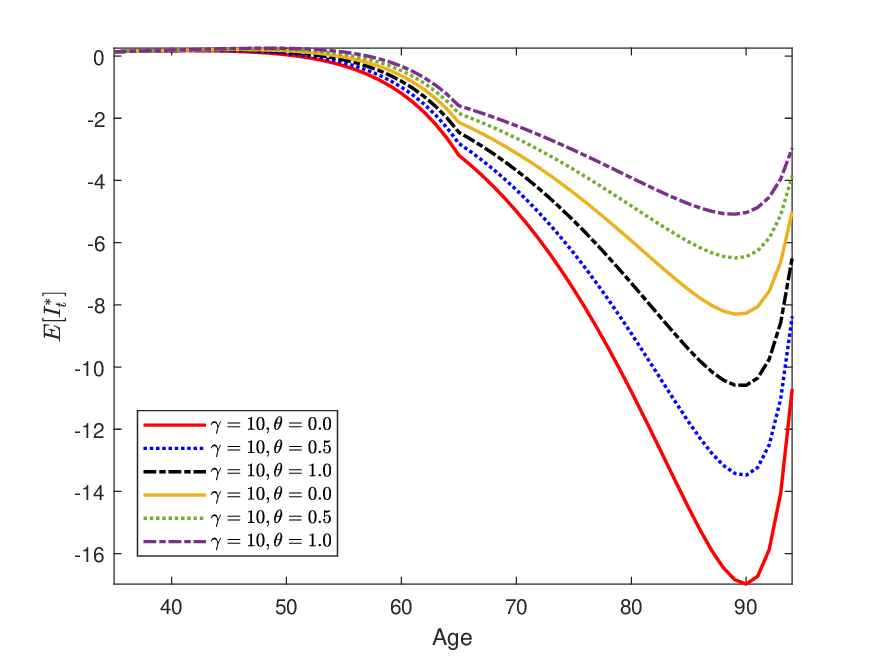}
        \label{fig_insurance_2D_gamma10:figure3}
    }
    \hspace{-0.45in}
    \subfigure[Insurance face value]
    {
        \includegraphics[width=3.0in,height=2.5in]{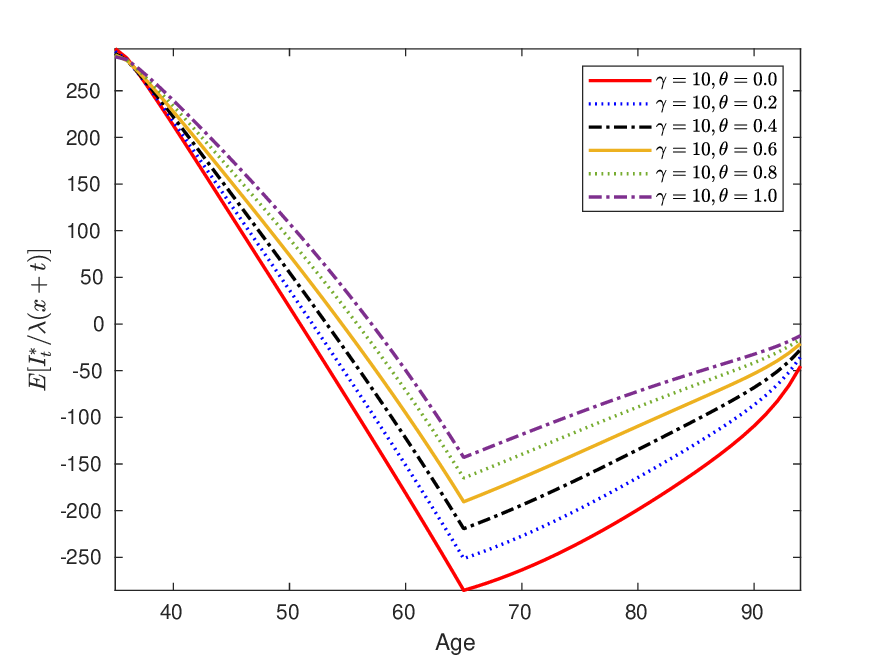}
        \label{fig_insurance_2D_gamma10:figure4}
    }
    \\
    \vspace{-0.10in}
    \subfigure[Surplus process]
    {
        \includegraphics[width=3.0in,height=2.5in]{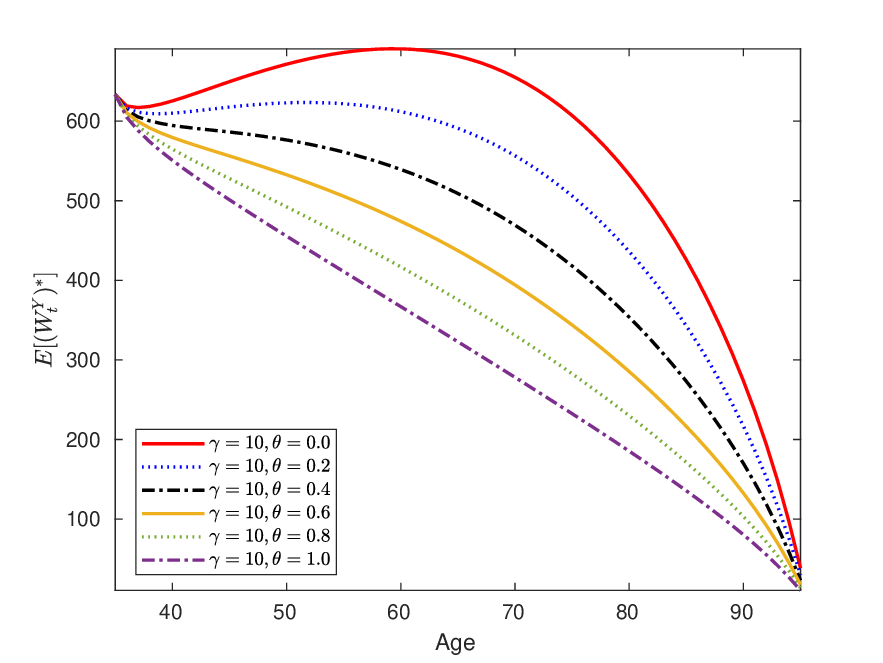}
        \label{fig_insurance_2D_gamma10:figure5}
    }
    \hspace{-0.45in}
    \subfigure[Bequest-wealth ratio]
    {
        \includegraphics[width=3.0in,height=2.5in]{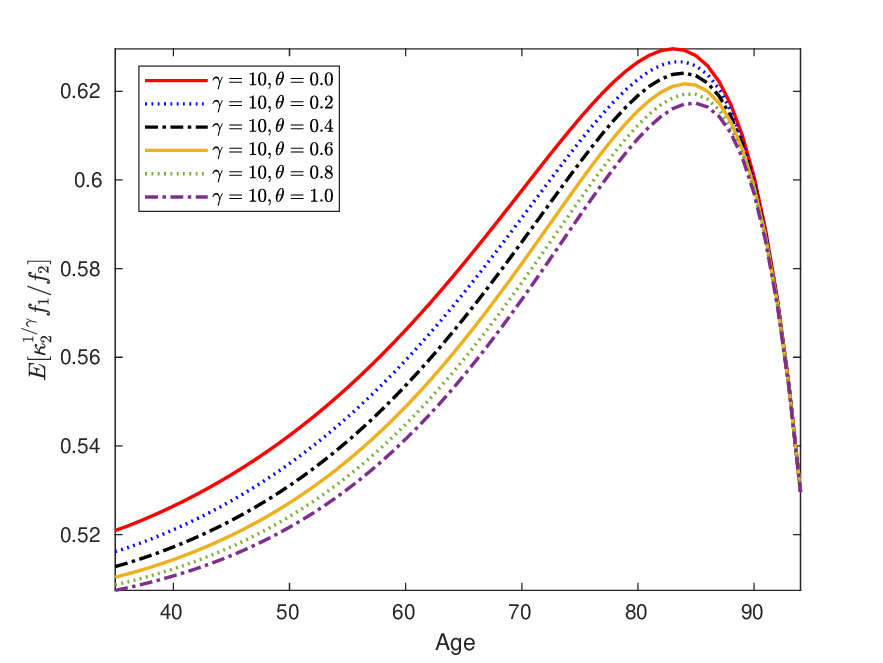}
        \label{fig_insurance_2D_gamma10:figure6}
    }
    \caption{Expected optimal consumption and insurance strategies under $\gamma=10$. }
    \label{fig_insurance_2D_gamma10:figures}
\end{figure}

\begin{figure}[htbp]
    \centering
    \subfigure[3-year nominal bonds]
    {
        \includegraphics[width=3.3in,height=1.7in]{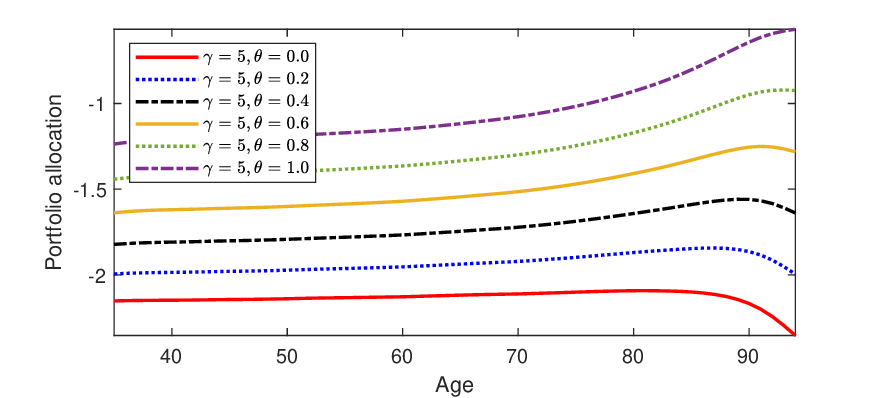}
        \label{fig_beta_2D_gamma5:figure1}
    }
    \hspace{-0.45in}
    \subfigure[10-year nominal bonds]
    {
        \includegraphics[width=3.3in,height=1.7in]{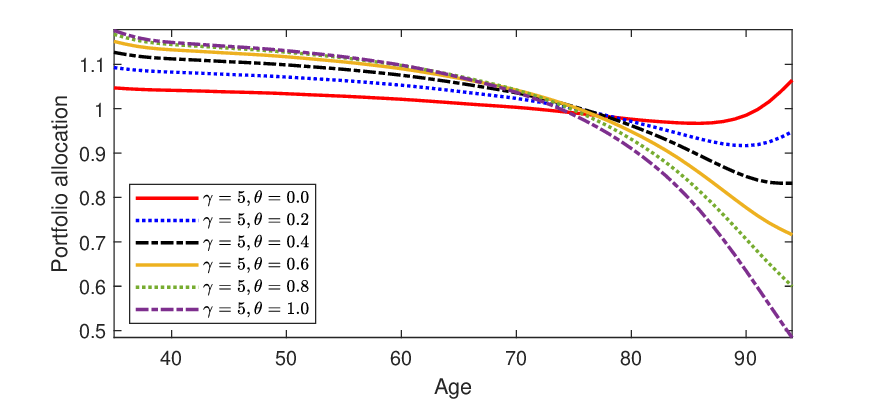}
        \label{fig_beta_2D_gamma5:figure2}
    }
    \\
    \subfigure[10-year inflation-linked bonds]
    {
        \includegraphics[width=3.3in,height=1.7in]{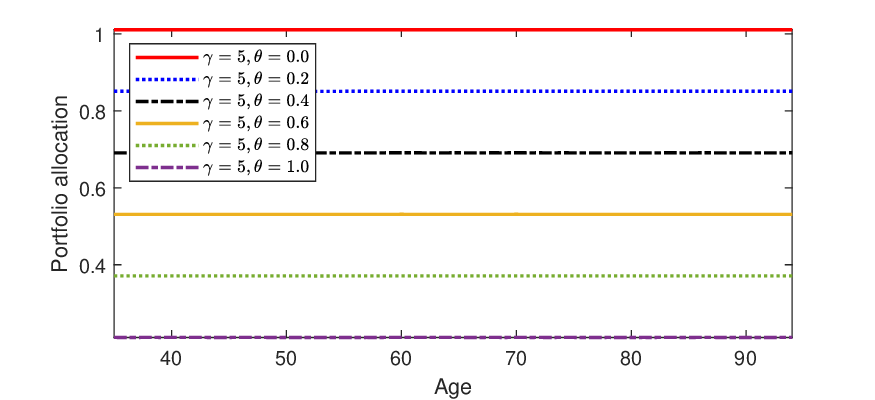}
        \label{fig_beta_2D_gamma5:figure3}
    }
    \hspace{-0.45in}
    \subfigure[Stocks]
    {
        \includegraphics[width=3.3in,height=1.7in]{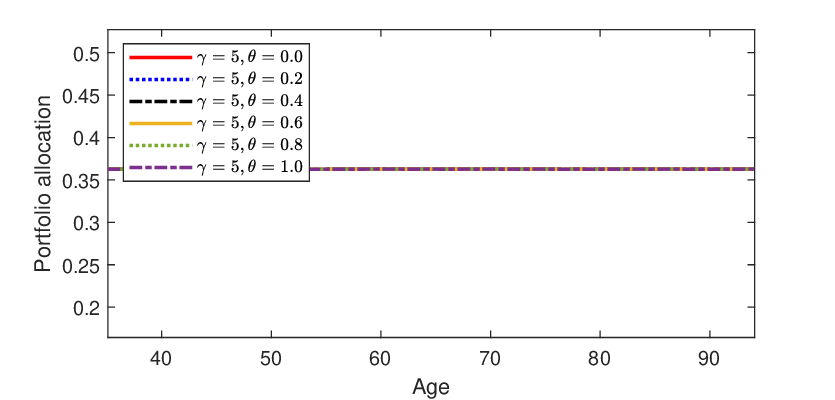}
        \label{fig_beta_2D_gamma5:figure4}
    }
    \subfigure[3-year nominal bonds (ITHD)]
    {
        \includegraphics[width=3.3in,height=1.7in]{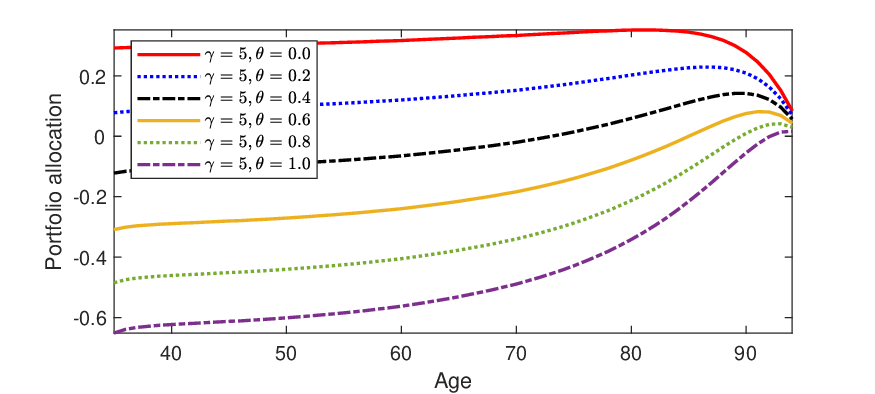}
        \label{fig_beta_2D_gamma5:figure5}
    }
    \hspace{-0.45in}
    \subfigure[10-year nominal bonds (ITHD)]
    {
        \includegraphics[width=3.3in,height=1.7in]{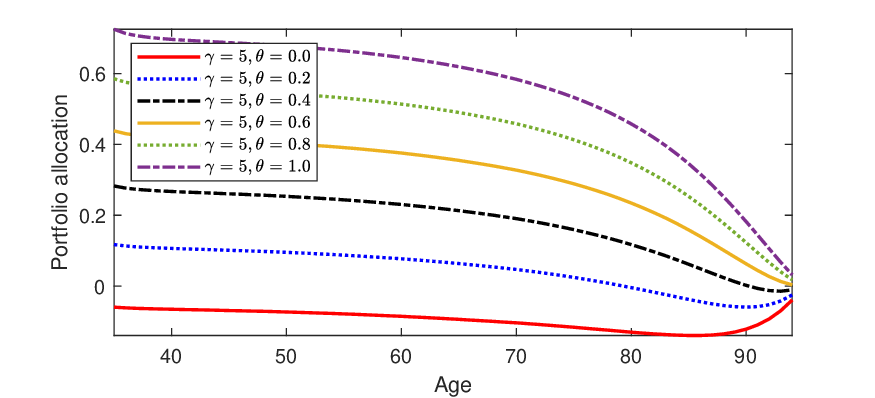}
        \label{fig_beta_2D_gamma5:figure6}
    }
    \\
    \subfigure[10-year inflation-linked bonds (ITHD)]
    {
        \includegraphics[width=3.3in,height=1.7in]{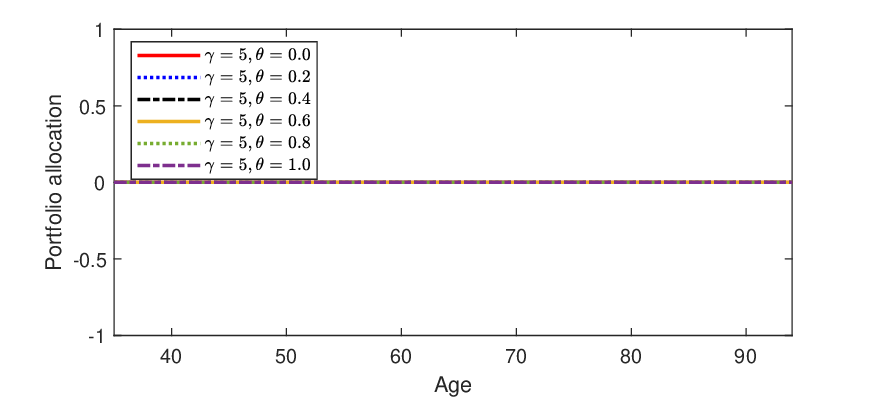}
        \label{fig_beta_2D_gamma5:figure7}
    }
    \hspace{-0.45in}
    \subfigure[Stocks (ITHD)]
    {
        \includegraphics[width=3.3in,height=1.7in]{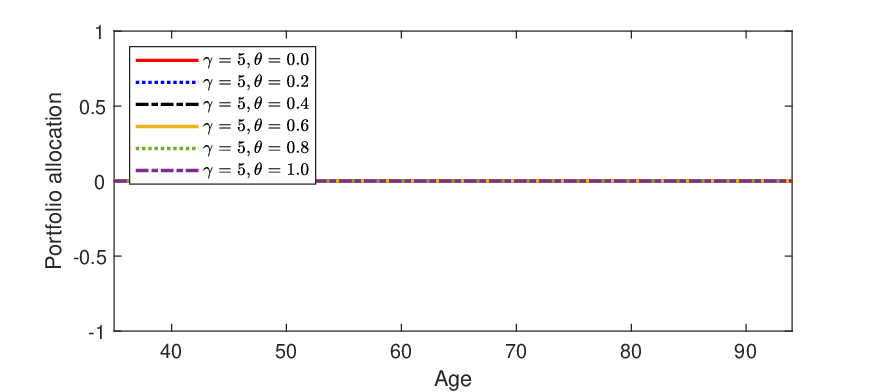}
        \label{fig_beta_2D_gamma5:figure8}
    }
    \caption{Expected optimal investment strategies for surplus process under $\gamma=5$. }
    \label{fig_beta_2D_gamma5:figures}
\end{figure}

\begin{figure}[htbp]
    \centering
    \subfigure[Breadwinner's consumption]
    {
        \includegraphics[width=3.0in,height=2.5in]{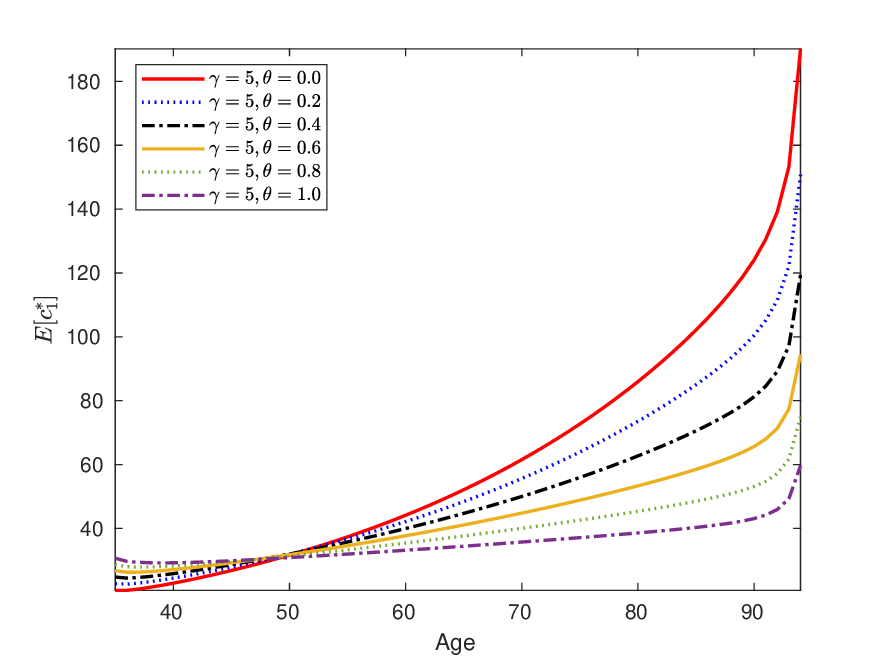}
        \label{fig_insurance_2D_gamma5:figure1}
    }
    \hspace{-0.45in}
    \subfigure[The rest of family's consumption]
    {
	\includegraphics[width=3.0in,height=2.5in]{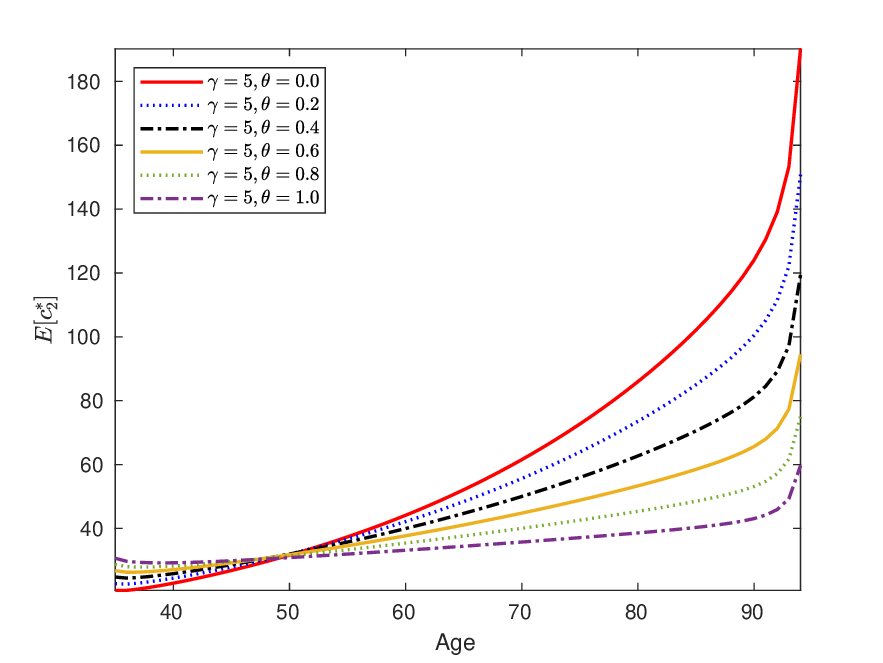}
        \label{fig_insurance_2D_gamma5:figure2}
    }
    \\
    \vspace{-0.10in}
    \subfigure[Insurance premium]
    {
        \includegraphics[width=3.0in,height=2.5in]{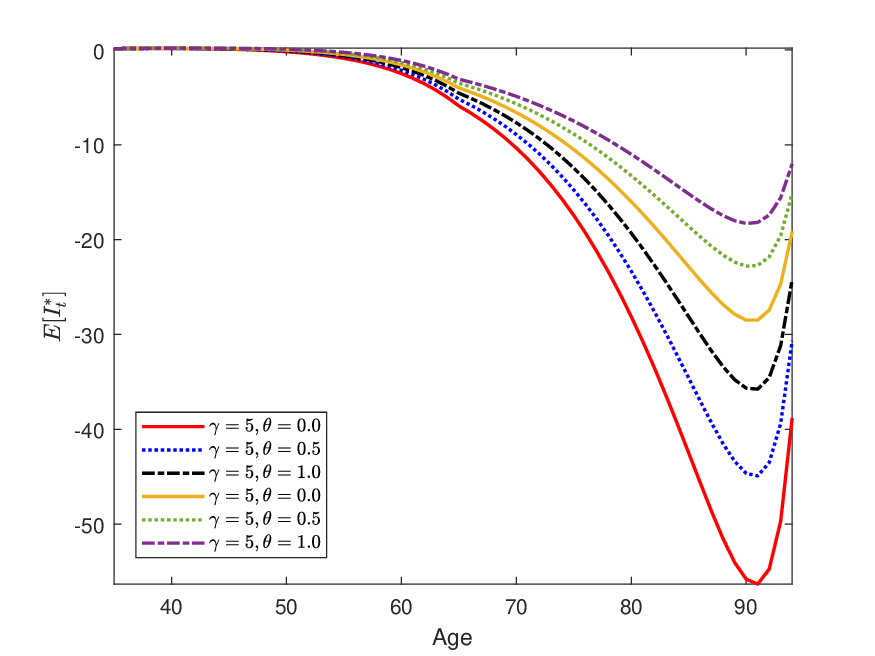}
        \label{fig_insurance_2D_gamma5:figure3}
    }
    \hspace{-0.45in}
    \subfigure[Insurance face value]
    {
        \includegraphics[width=3.0in,height=2.5in]{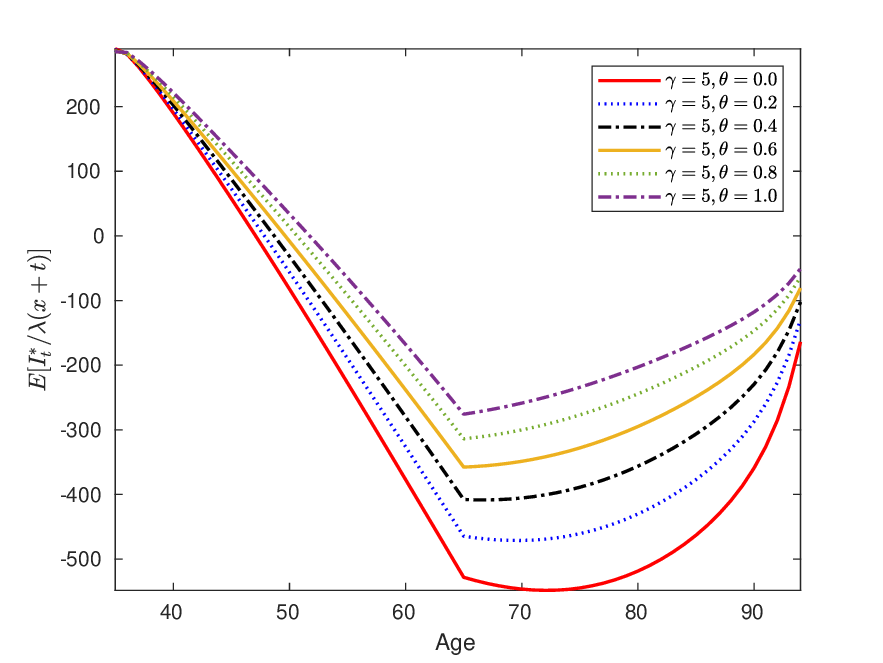}
        \label{fig_insurance_2D_gamma5:figure4}
    }
    \\
    \vspace{-0.10in}
    \subfigure[Surplus process]
    {
        \includegraphics[width=3.0in,height=2.5in]{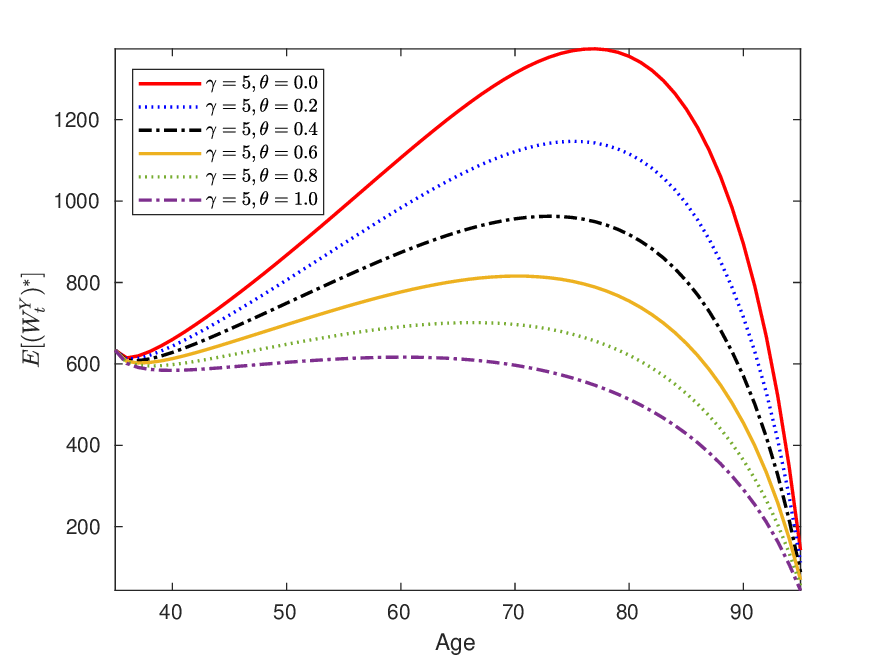}
        \label{fig_insurance_2D_gamma5:figure5}
    }
    \hspace{-0.45in}
    \subfigure[Bequest-wealth ratio]
    {
        \includegraphics[width=3.0in,height=2.5in]{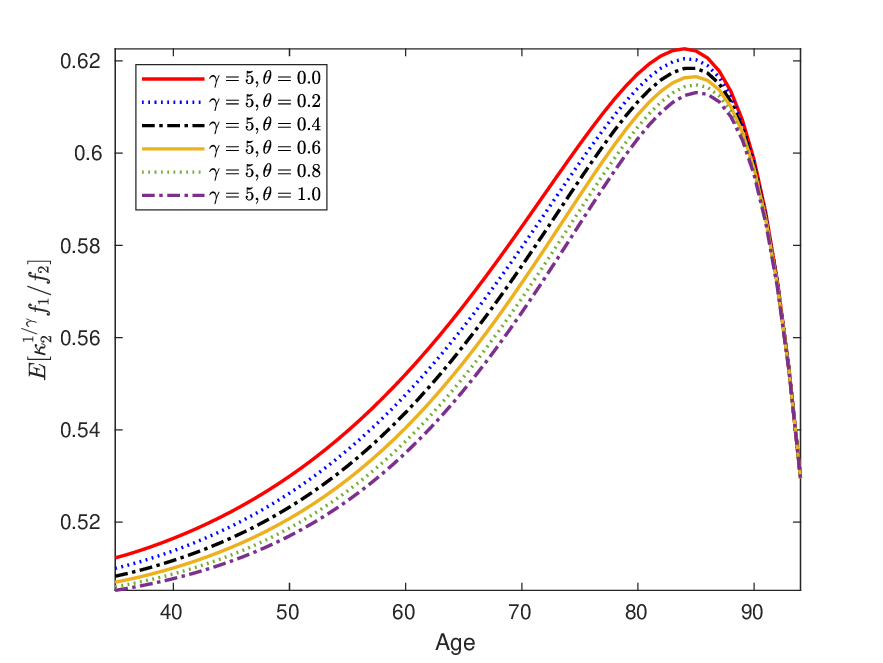}
        \label{fig_insurance_2D_gamma5:figure6}
    }
    \caption{Expected optimal insurance strategy and its components under $\gamma=5$. }
    \label{fig_insurance_2D_gamma5:figures}
\end{figure}

\subsection{Sensitivity of optimal strategies concerning the two factors $X_t$}

This section performs a static analysis of optimal strategies, considering two factors $X_t$. For all figures in this section, the range for $X_1$ is set as $[-0.1454,0.1454]$, and for $X_2$, it is $[-0.1696,0.1696]$. These ranges represent the maximum values for $X_1$ and $X_2$, respectively, in the Monte Carlo simulation. This simulation is conducted with 10,000,000 paths, with a time step of one year.

Figure \ref{fig_beta_3D_theta0} illustrates the optimal trading strategy $\beta^*$ considering two factors without money illusion ($\theta = 0$). The demand for nominal bonds escalates when the expected inflation $\pi^e_t$ increases. On the other hand, the demand for short-term nominal bonds rises with the real short rate $r_t$, while the demand for long-term nominal bonds shows an inverse relation. This outcome indicates that short-term nominal bonds are more favorable when two factors fluctuate over a short period. Figures for inflation-linked bonds and stocks are also presented, showing a decrease when $\pi^e_t$ and $r_t$ increase. In addition to the optimal trading strategy $\beta^*$, we plot the figures for standard myopic demand (SMD) and intertemporal hedging demand (ITHD). According to \eqref{adj_beta_decomposition}, $\text{ITHD}_3$ and $\text{ITHD}_4$ are null as the third and fourth entries of $\Sigma_X$ are zeros. Furthermore, the inflation hedging demand (IFHD) is outlined in Table \ref{Table_SMD_IFHD}, depicted as a constant vector that decreases linearly with $\theta$. Our sensitivity analysis reveals that ITHD primarily shapes the breadwinner's allocation in long-term nominal bonds, whereas SMD governs the allocations to the other three financial instruments.

Figure \ref{fig_beta_3D_theta08} shows the optimal trading strategy under $\theta=0.8$. When compared with Figure \ref{fig_beta_3D_theta0}, it is evident that when the entire family ignores inflation risk, they make several adjustments to their short-period allocations: (1) They purchase more short-term nominal bonds and fewer long-term nominal bonds. (2) They reduce their long positions and increase their short positions in inflation-linked bonds. (3) Their trading strategy for stocks remains relatively unchanged.

Figure \ref{fig_I_3D_theta0} plots the optimal consumption and insurance/annuity strategies when $\theta=0$. The results show that the family's consumption and annuity purchases follow an upward ``U-shape" with respect to the two factors (see Figure \ref{fig_I_3D_theta0:annuity}, the negative downward ``U-shape" of  $I^*$ means a positive upward ``U-shape" annuity demand). In comparison, life insurance purchases exhibit a downward ``U-shape'' with two factors. In other words, when the real short rate and the expected inflation rate are both high or low, the family increases their consumption and annuity purchases while reducing their life insurance purchases. A potential explanation for the above observation is that life insurance protects future income, acting as a substitute for current consumption. Simultaneously, annuities function as a source of current income, supplementing current consumption. As a result, annuity demand shows an upward ``U-shape" aligned with consumption. In contrast, life insurance demand performs a downward ``U-shape". Moreover, we also find the following interesting findings: (1) Life insurance demand is more sensitive to expected inflation than to the real short rate. (2) Annuity demand is sensitive to expected inflation and real short rates. Lastly, according to \eqref{bequest_wealth_ratio}, we plot the figures for the bequest-wealth ratio and future income. The result shows that the bequest-wealth ratio primarily shapes the demand for life insurance and annuities. 

Figure \ref{fig_I_3D_theta08} displays the optimal consumption, life insurance, and annuity strategies when $\theta=0.8$. Compared with Figure \ref{fig_I_3D_theta0}, we find that the family ignoring the inflation risk will: (1) consume more in their early years and consume less in their old years, (2) buy more insurance and fewer annuities, and (3) slightly decrease their bequest-wealth demand. These results coincide with the sensitivity analysis to age, as detailed in Section \ref{sensitivity_to_age}.

\begin{sidewaysfigure}[htbp]
%\begin{figure}[htbp]
	\centering
    \subfigure[$\beta^*_1$]
    {
	    	\includegraphics[width=2.0in,height=2.0in]{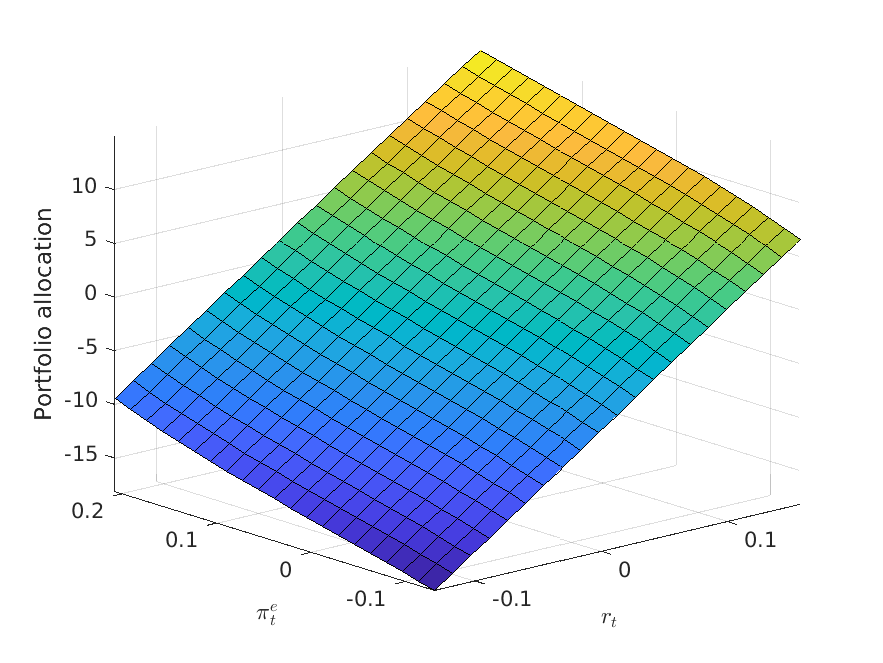}
    }
    %\hspace{-0.2in}
    \subfigure[$\beta^*_2$]
    {
    		\includegraphics[width=2.0in,height=2.0in]{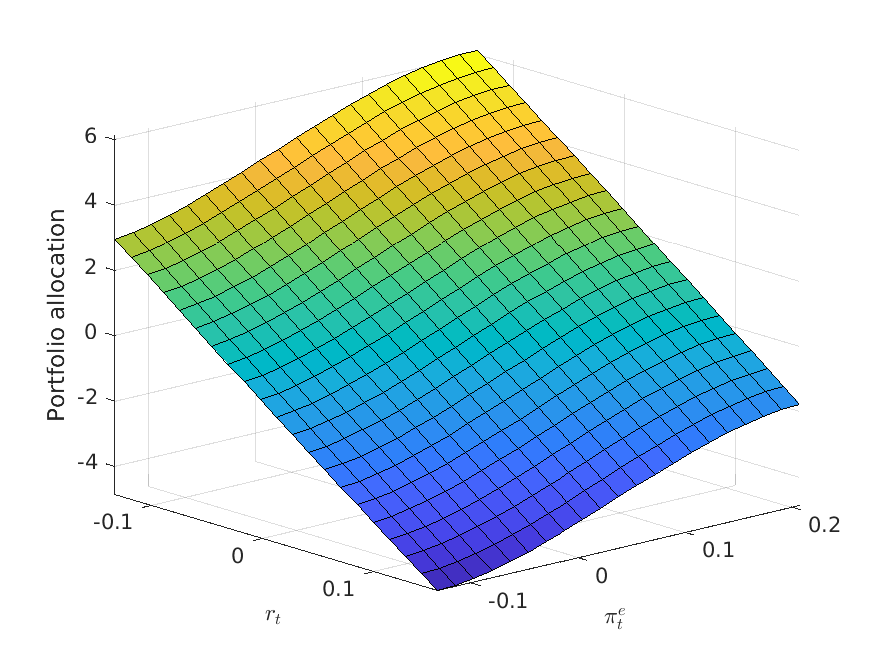}
    }
    %\vspace{-0.1in}
    \subfigure[$\beta^*_3$]
    {
    		\includegraphics[width=2.0in,height=2.0in]{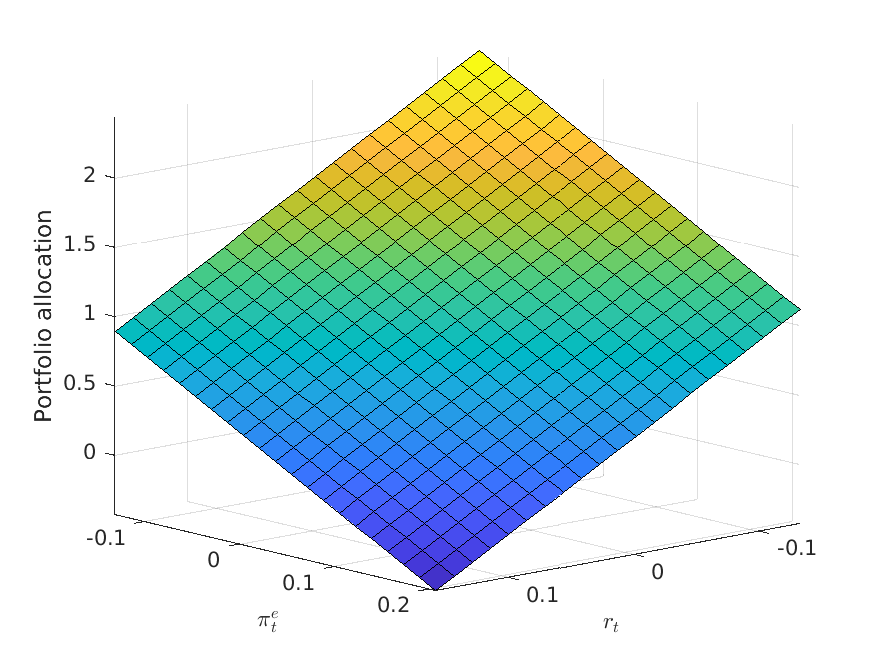}
    }
    %\hspace{-0.2in}
    \subfigure[$\beta^*_4$]
    {
	    	\includegraphics[width=2.0in,height=2.0in]{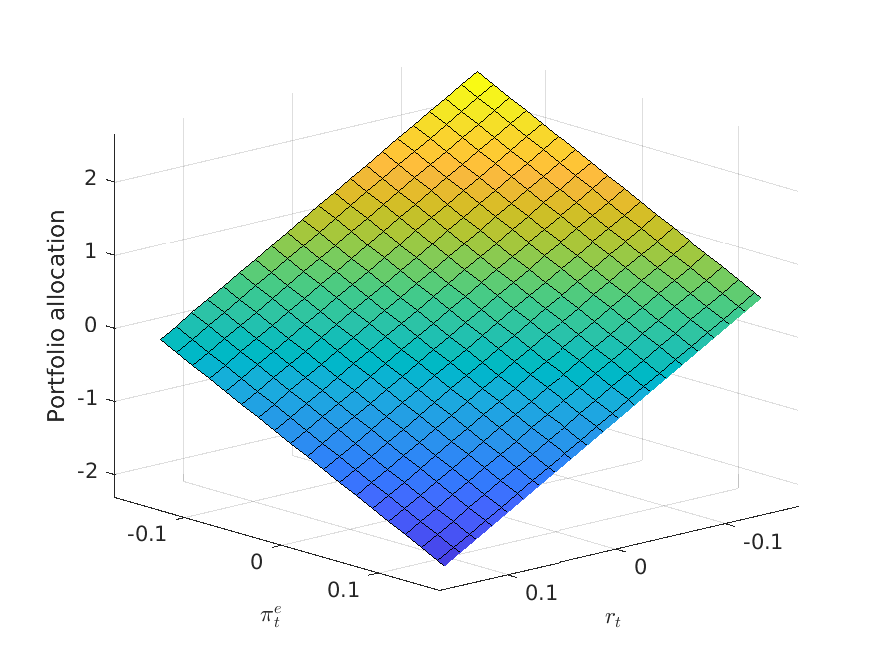}
    }
    \subfigure[$\text{SMD}_1$]
    {
    	\includegraphics[width=2.0in,height=2.0in]{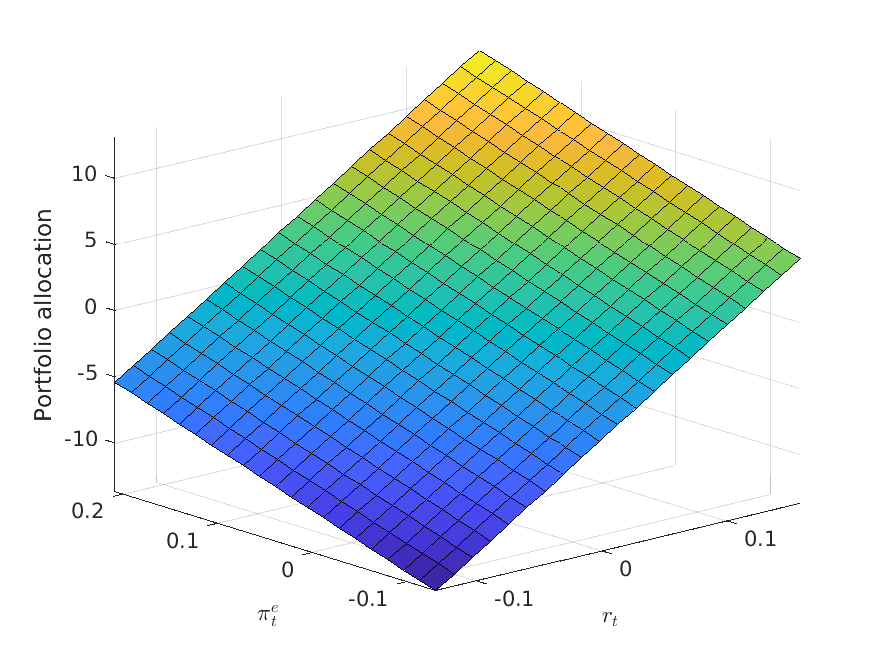}
    }
    %\hspace{-0.2in}
    \subfigure[$\text{SMD}_2$]
    {
    	\includegraphics[width=2.0in,height=2.0in]{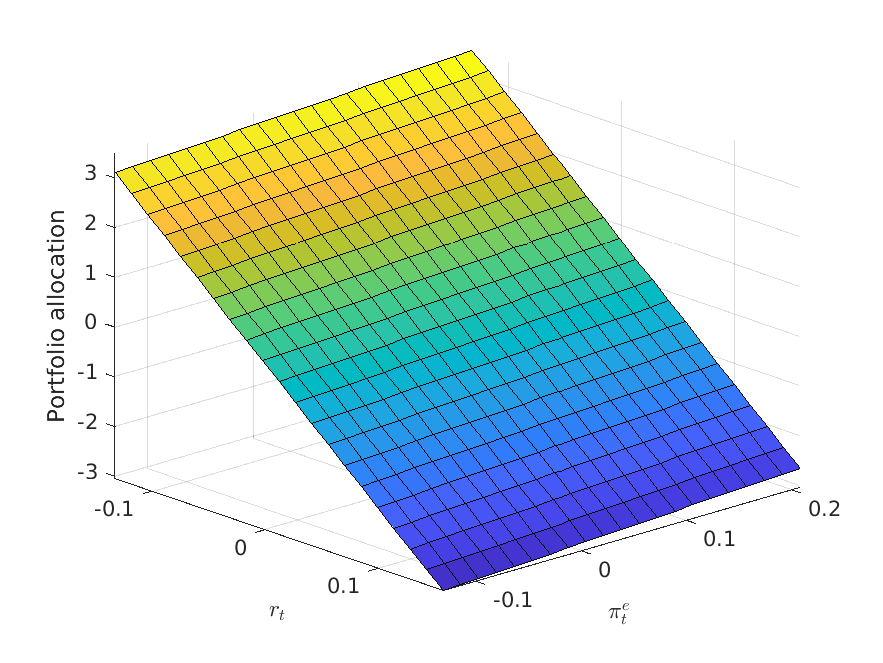}
    }
    %\vspace{-0.1in}
    \subfigure[$\text{SMD}_3$]
    {
    	\includegraphics[width=2.0in,height=2.0in]{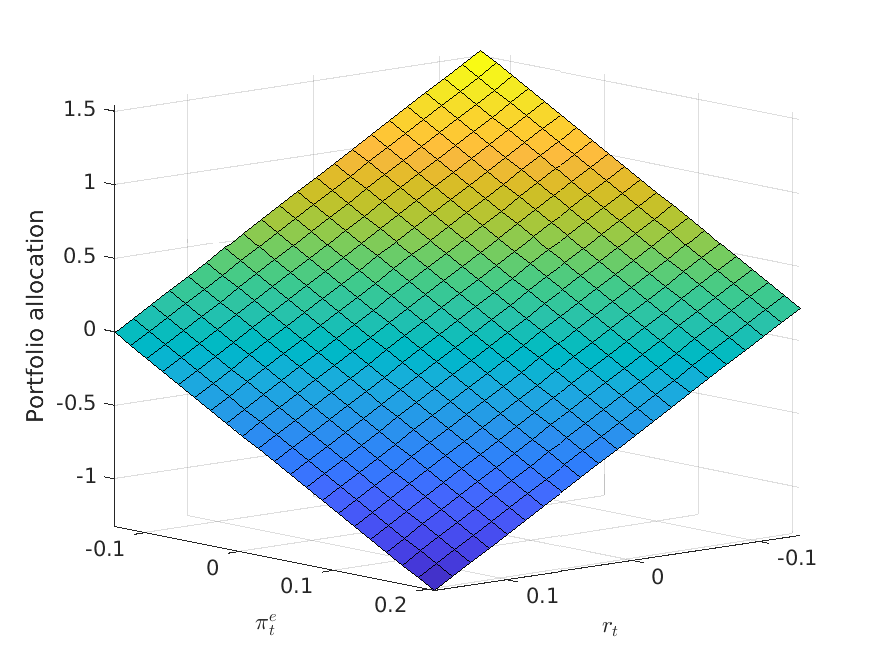}
    }
    %\hspace{-0.2in}
    \subfigure[$\text{SMD}_4$]
    {
    	\includegraphics[width=2.0in,height=2.0in]{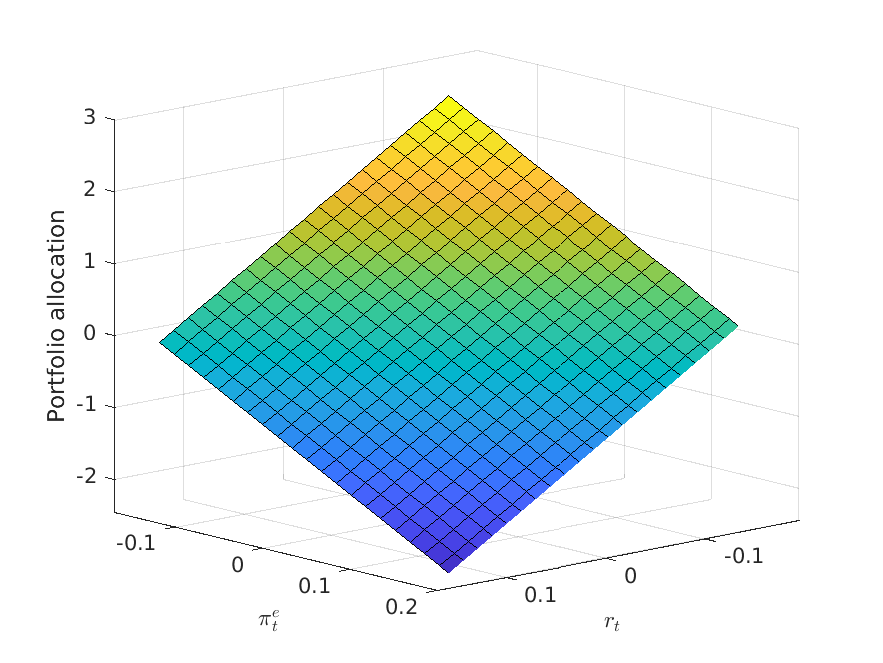}
    }
    \subfigure[$\text{ITHD}_1$]
    {
    	\includegraphics[width=2.0in,height=2.0in]{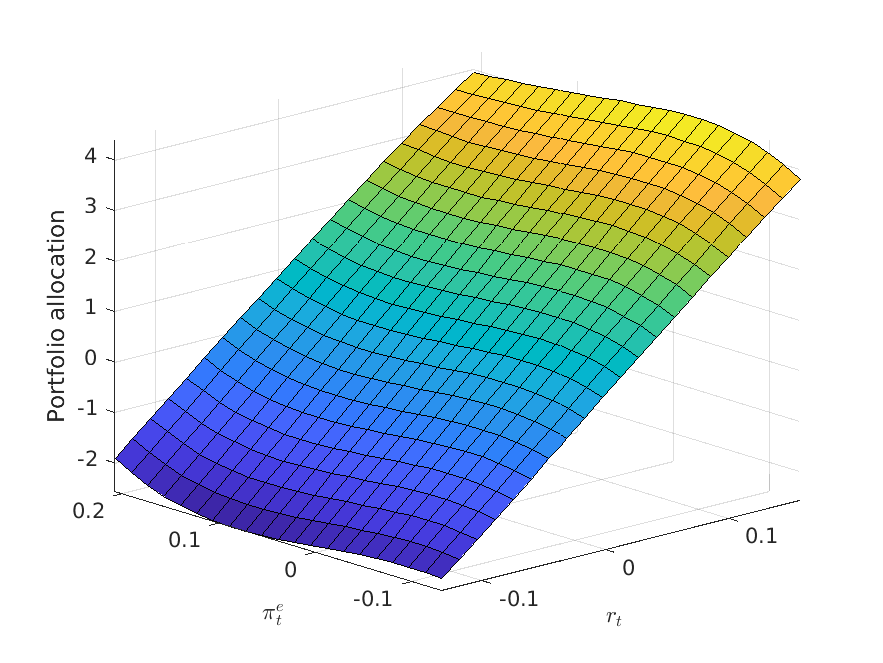}
    }
    %\hspace{-0.2in}
    \subfigure[$\text{ITHD}_2$]
    {
    	\includegraphics[width=2.0in,height=2.0in]{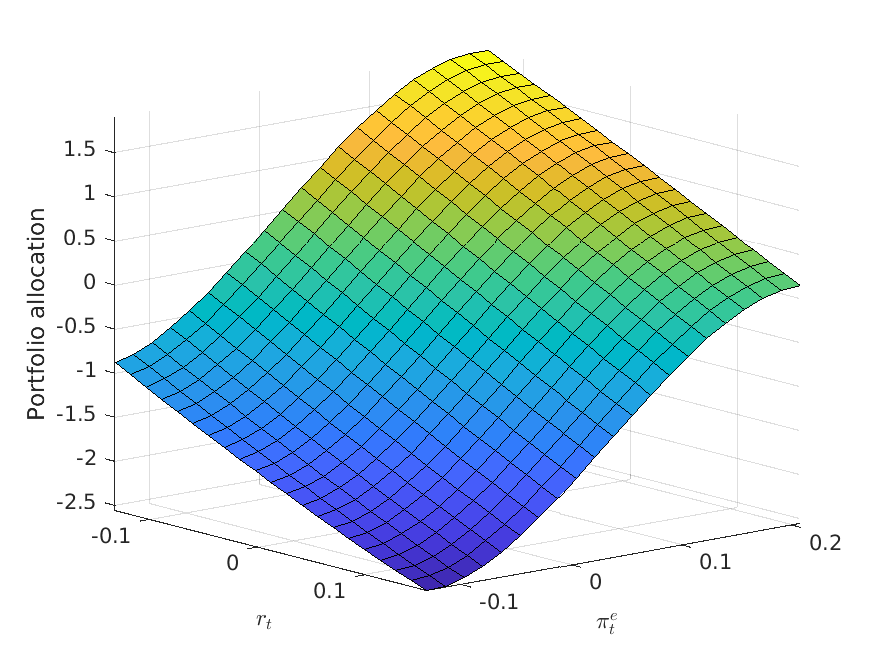}
    }
    %\vspace{-0.1in}
    \subfigure[$\text{ITHD}_3$]
    {
    	\includegraphics[width=2.0in,height=2.0in]{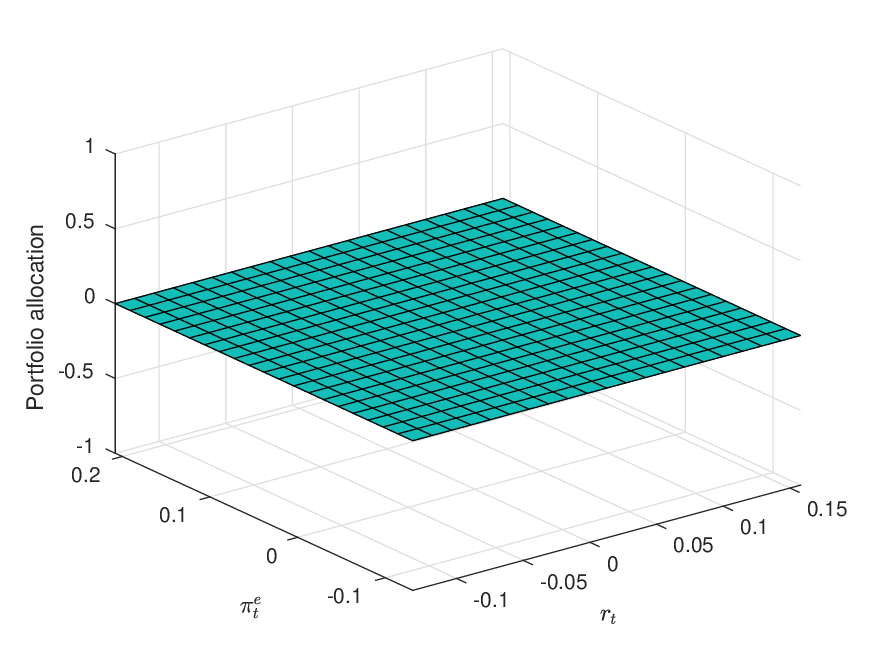}
    }
    %\hspace{-0.2in}
    \subfigure[$\text{ITHD}_4$]
    {
    	\includegraphics[width=2.0in,height=2.0in]{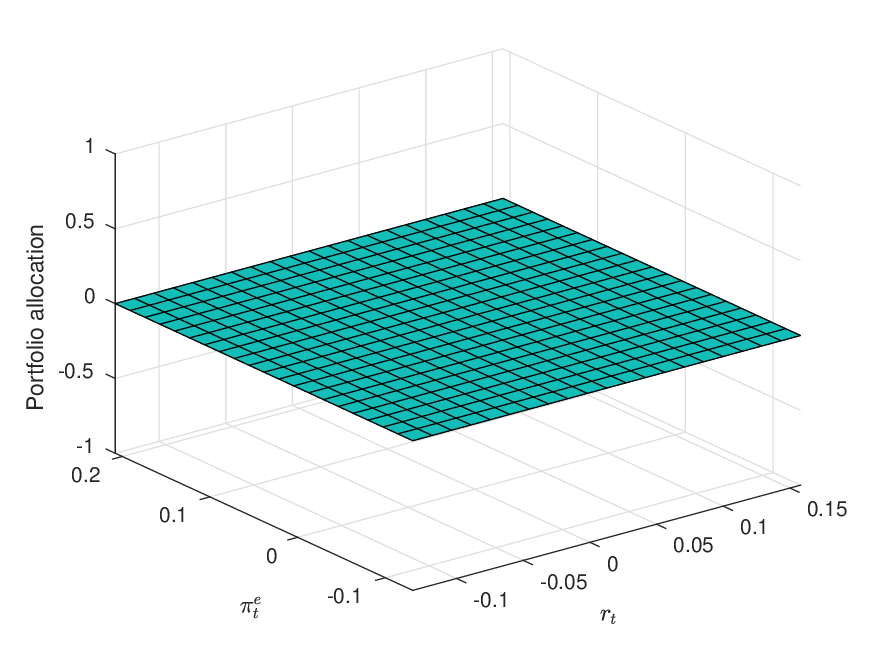}
    }
	\caption{Family's trading strategy when $\theta=0$ and $t=5$ (age 40).}
 \label{fig_beta_3D_theta0}
%\end{figure}
\end{sidewaysfigure}

\begin{sidewaysfigure}[htbp]
%\begin{figure}[htbp]
	\centering
    \subfigure[$\beta^*_1$]
    {
	    	\includegraphics[width=2.0in,height=2.0in]{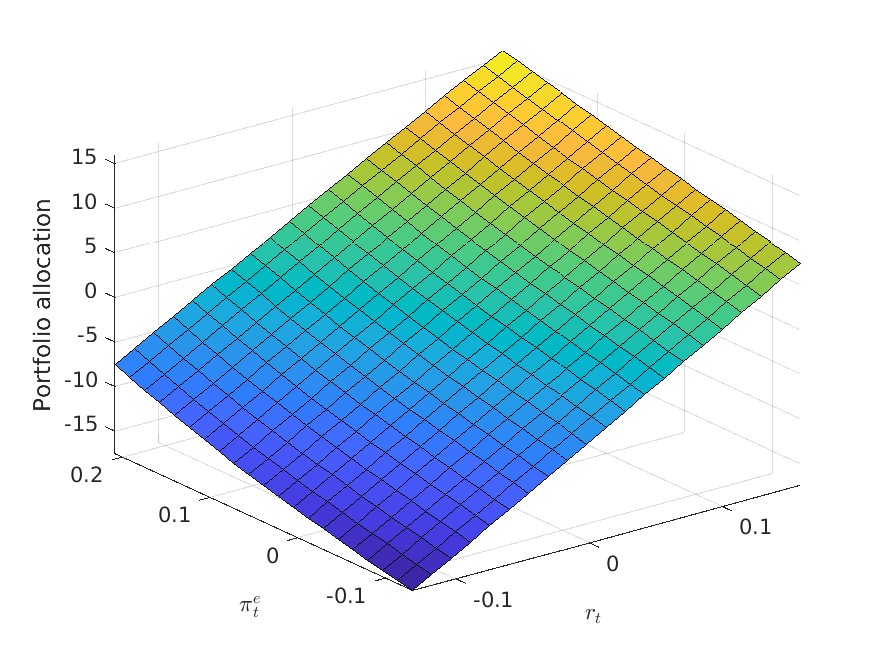}
    }
    %\hspace{-0.2in}
    \subfigure[$\beta^*_2$]
    {
    		\includegraphics[width=2.0in,height=2.0in]{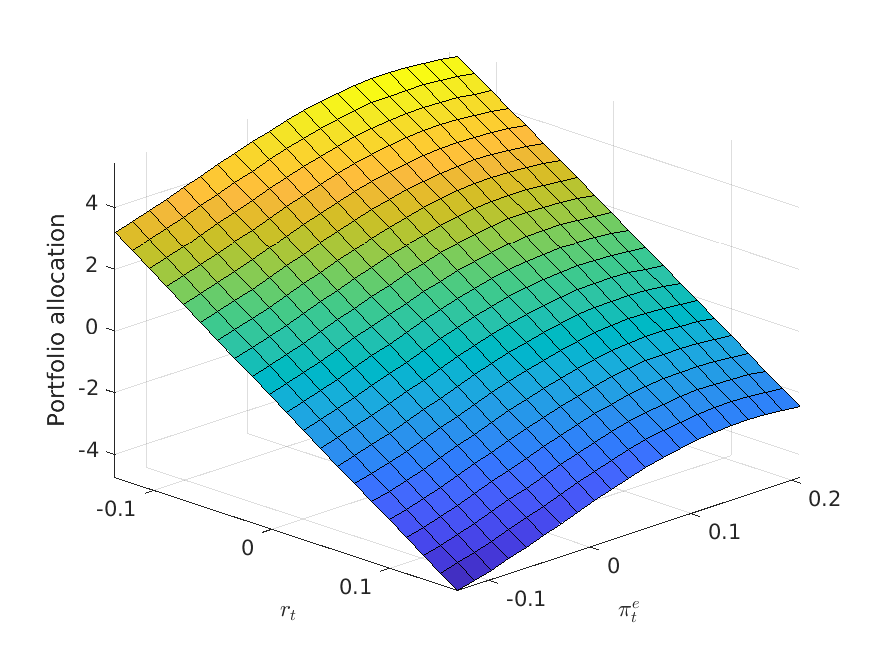}
    }
    %\vspace{-0.1in}
    \subfigure[$\beta^*_3$]
    {
    		\includegraphics[width=2.0in,height=2.0in]{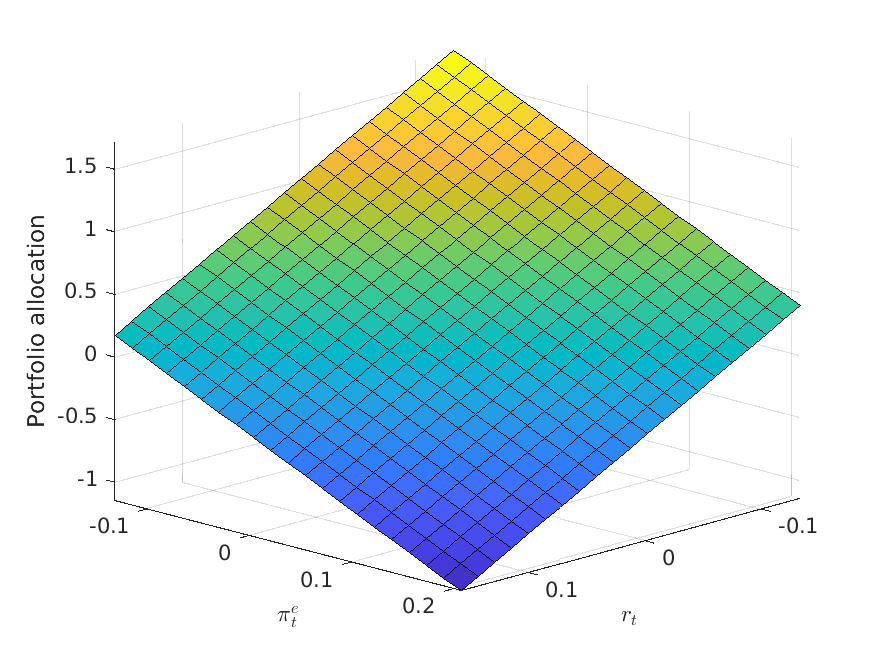}
    }
    %\hspace{-0.2in}
    \subfigure[$\beta^*_4$]
    {
	    	\includegraphics[width=2.0in,height=2.0in]{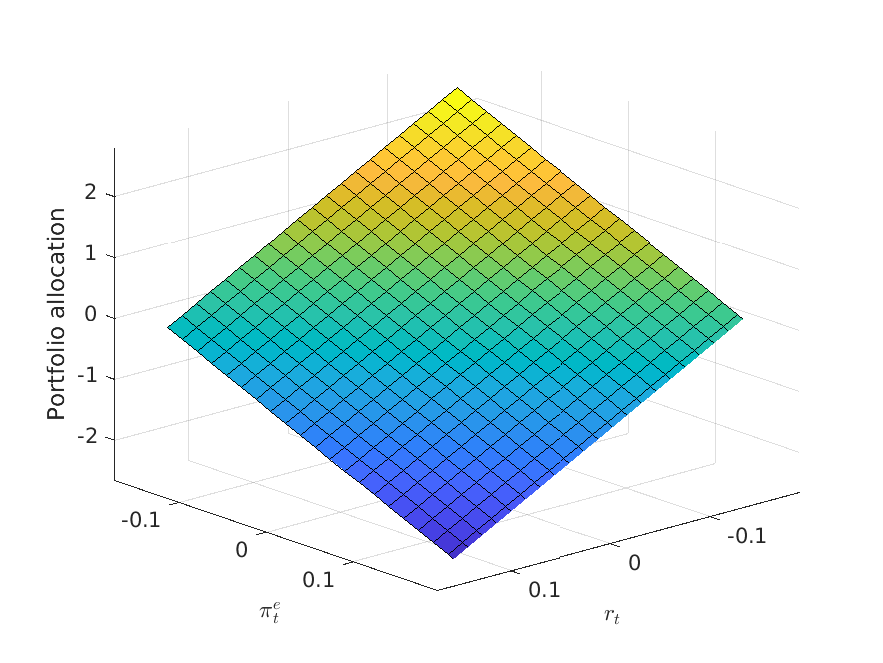}
    }
    \subfigure[$\text{SMD}_1$]
    {
    	\includegraphics[width=2.0in,height=2.0in]{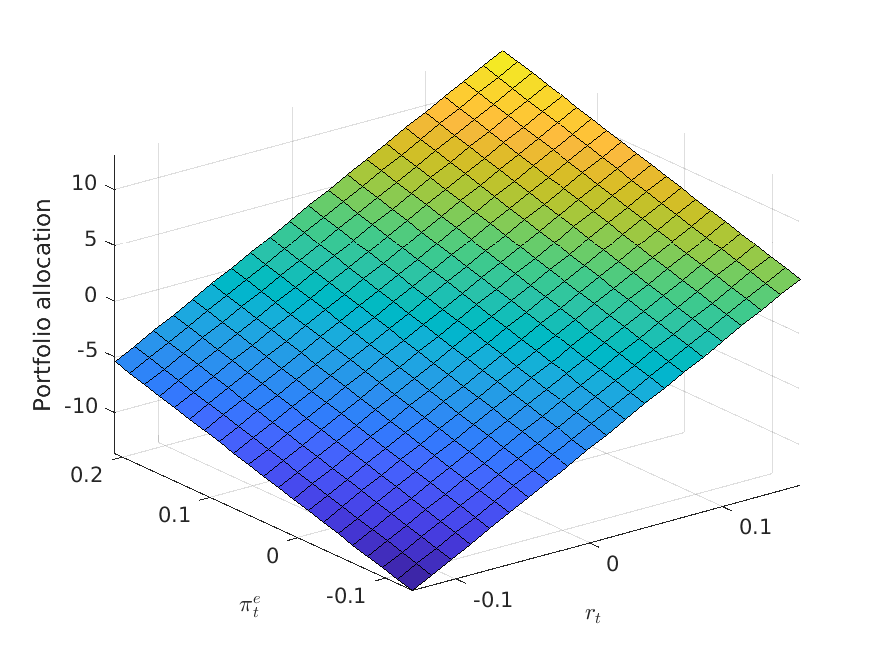}
    }
    %\hspace{-0.2in}
    \subfigure[$\text{SMD}_2$]
    {
    	\includegraphics[width=2.0in,height=2.0in]{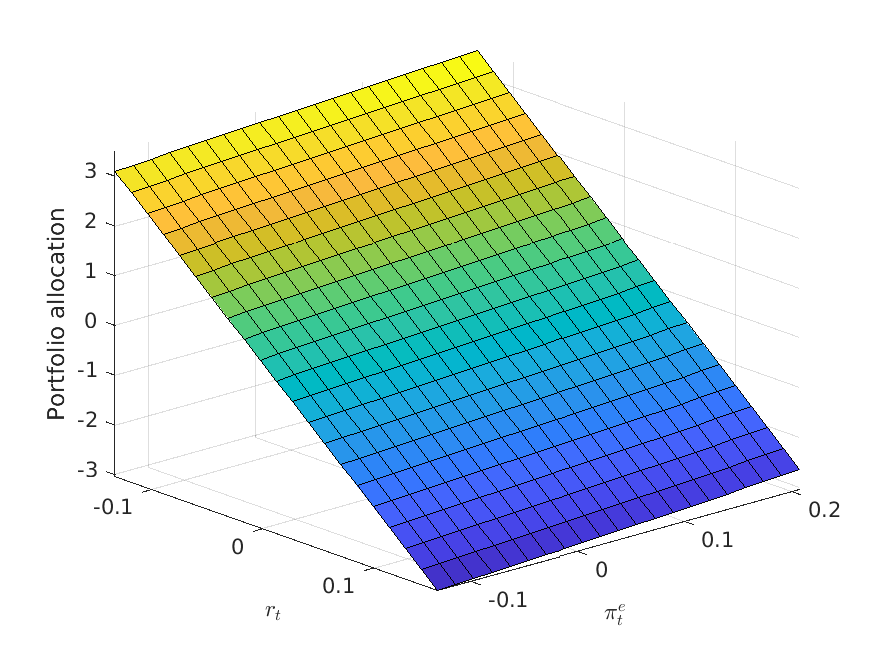}
    }
    %\vspace{-0.1in}
    \subfigure[$\text{SMD}_3$]
    {
    	\includegraphics[width=2.0in,height=2.0in]{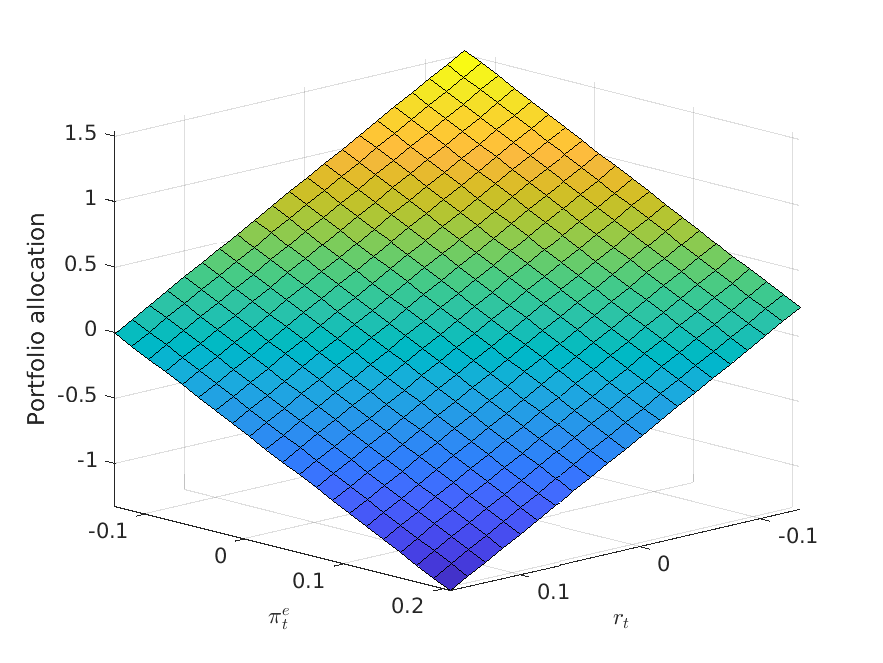}
    }
    %\hspace{-0.2in}
    \subfigure[$\text{SMD}_4$]
    {
    	\includegraphics[width=2.0in,height=2.0in]{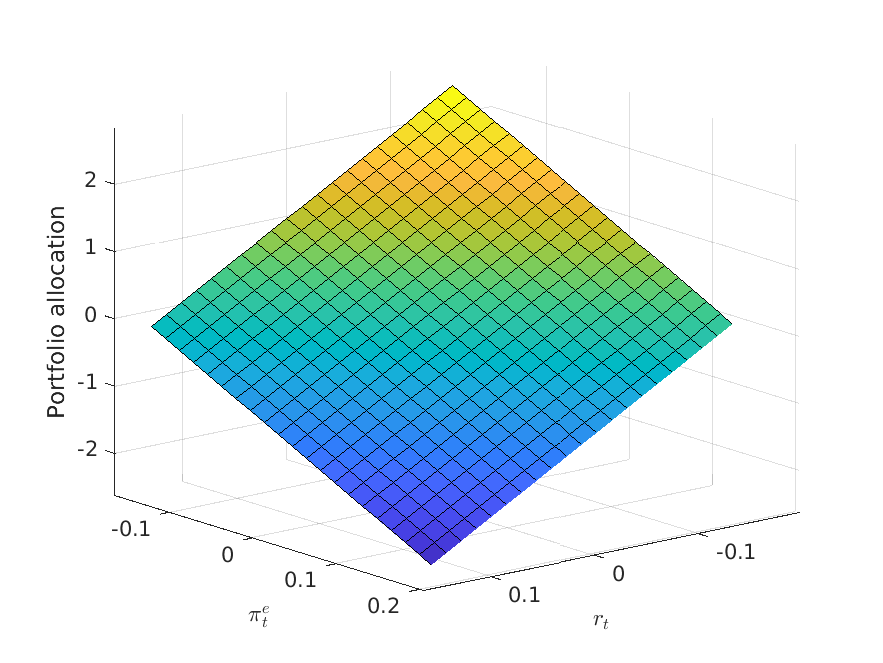}
    }
    \subfigure[$\text{ITHD}_1$]
    {
    	\includegraphics[width=2.0in,height=2.0in]{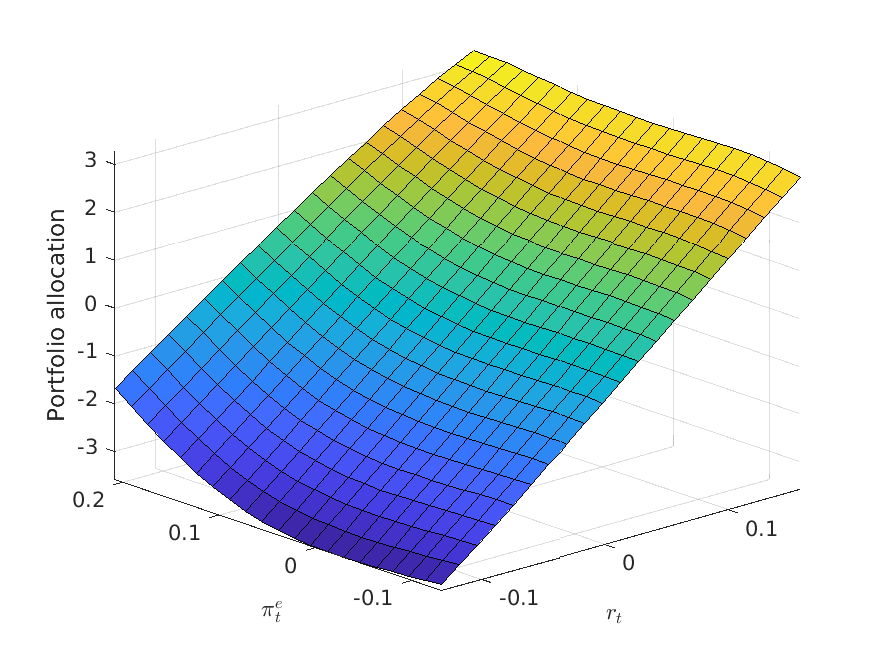}
    }
    %\hspace{-0.2in}
    \subfigure[$\text{ITHD}_2$]
    {
    	\includegraphics[width=2.0in,height=2.0in]{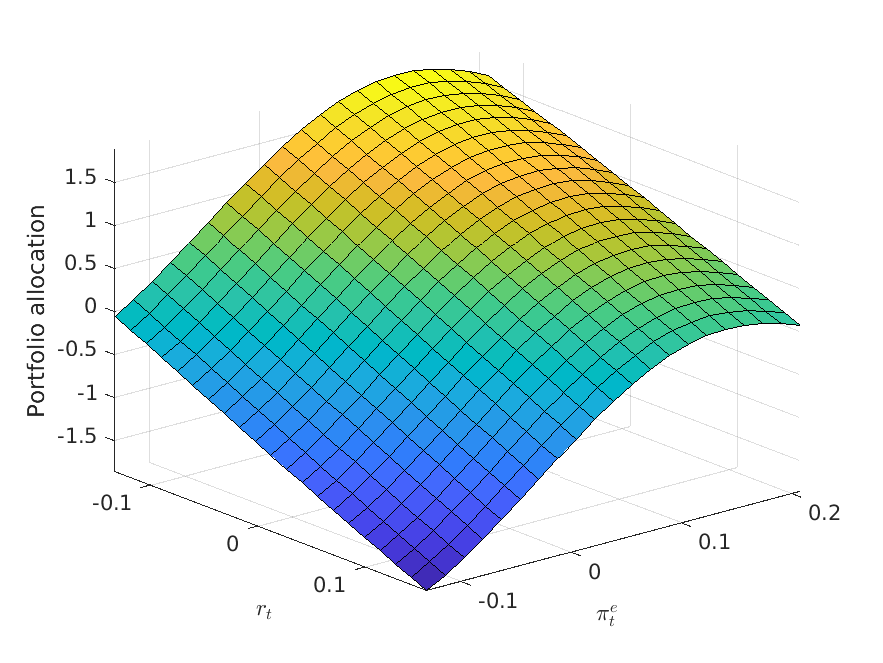}
    }
    %\vspace{-0.1in}
    \subfigure[$\text{ITHD}_3$]
    {
    	\includegraphics[width=2.0in,height=2.0in]{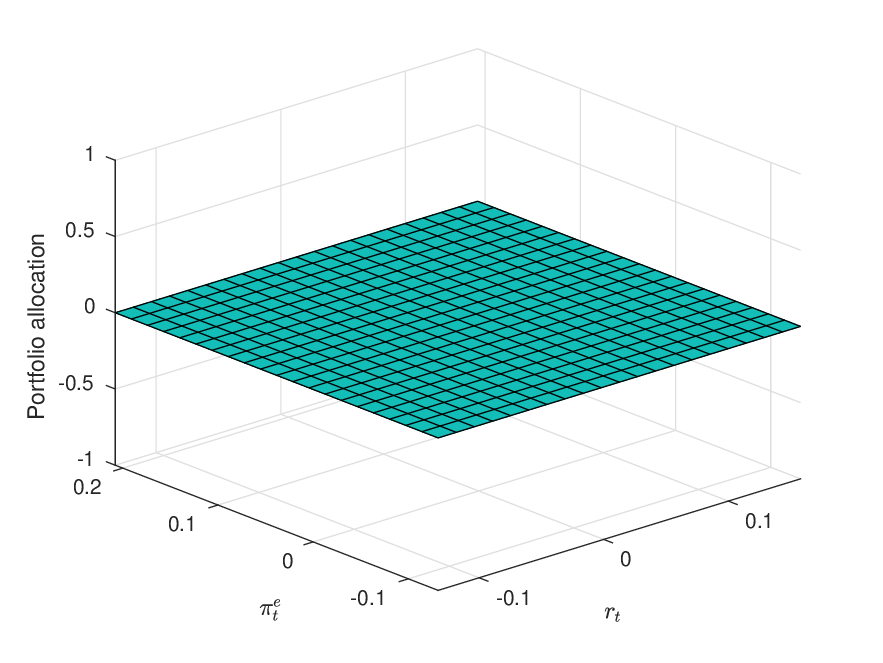}
    }
    %\hspace{-0.2in}
    \subfigure[$\text{ITHD}_4$]
    {
    	\includegraphics[width=2.0in,height=2.0in]{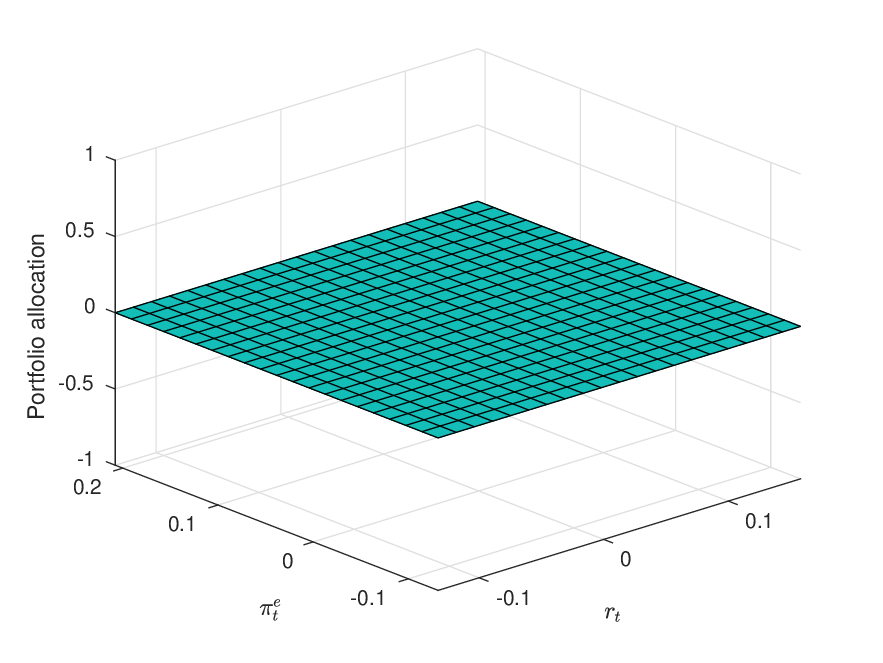}
    }
	\caption{Family's trading strategy when $\theta=0.8$ and $t=5$ (age 40).}
	\label{fig_beta_3D_theta08}
%\end{figure}
\end{sidewaysfigure}

\begin{sidewaysfigure}[htbp]
    \centering
    \subfigure[$c^*_1$ (age 40)]
    {
        \includegraphics[width=2.0in,height=2.0in]{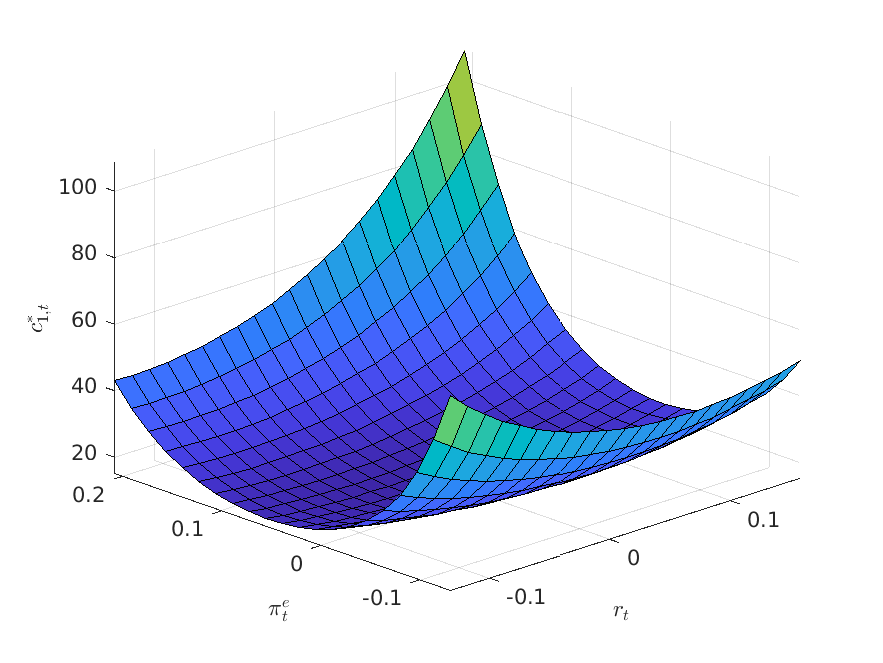}
    }
    %\hspace{-0.35in}	
    \subfigure[$c^*_1$ (age 90)]
    {
	\includegraphics[width=2.0in,height=2.0in]{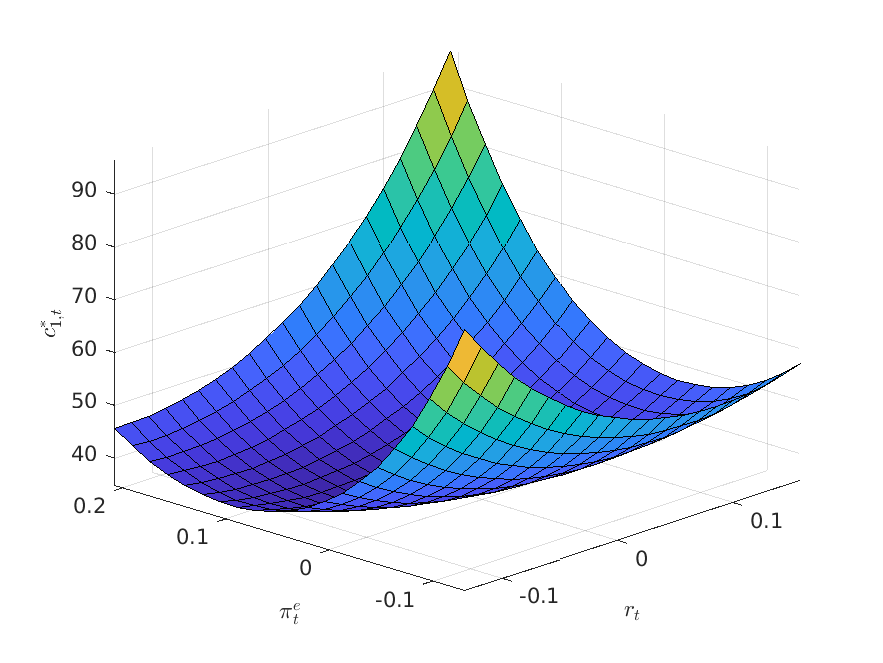}
    }
    %\hspace{-0.35in}
    \subfigure[$c^*_2$ (age 40)]
    {
	\includegraphics[width=2.0in,height=2.0in]{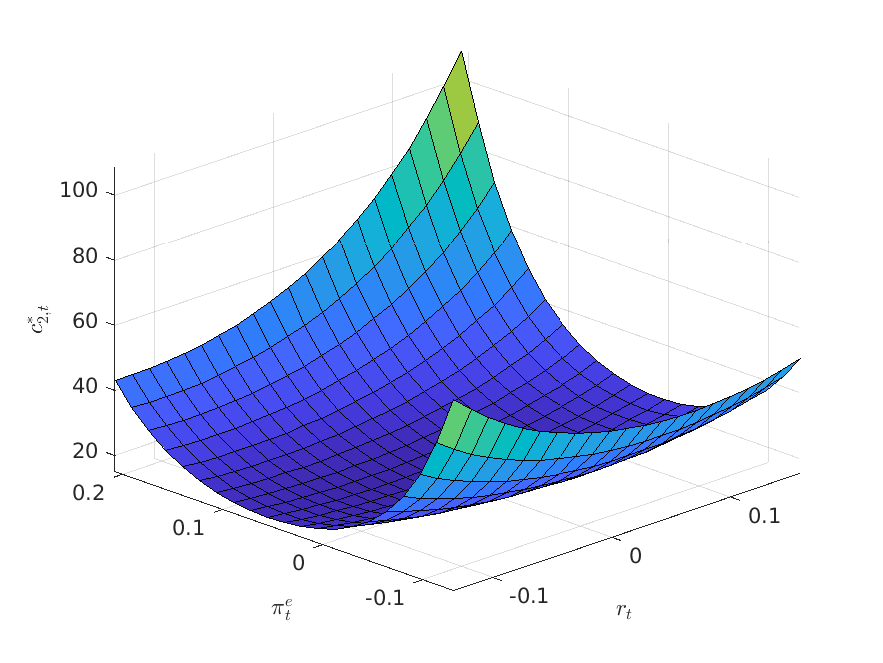}
    }
    %\hspace{-0.35in}
    \subfigure[$c^*_2$ (age 90)]
    {
	\includegraphics[width=2.0in,height=2.0in]{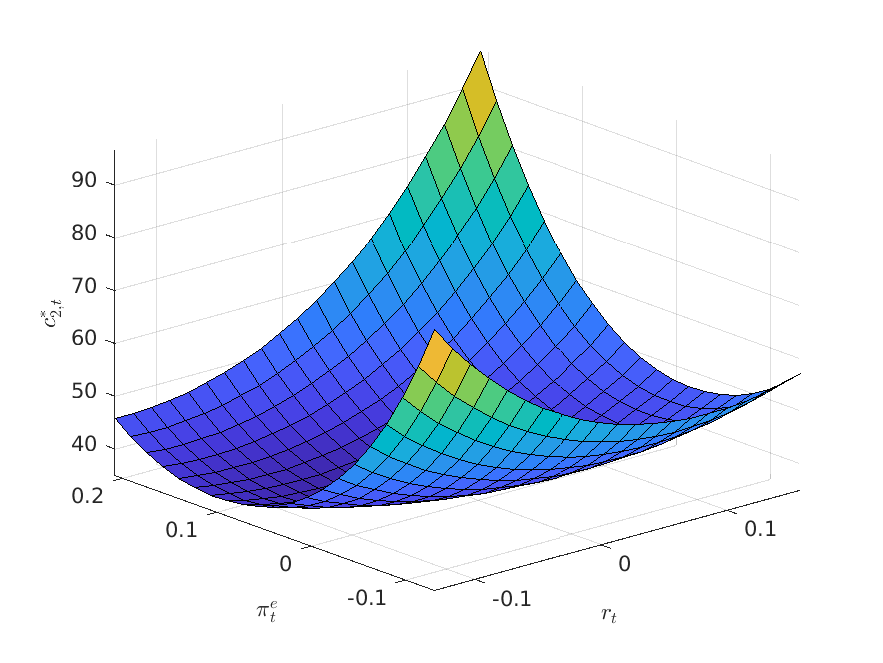}
    }
    \vspace{-0.15in}
    \subfigure[$I^*$ (age 40)]
    {
	\includegraphics[width=2.0in,height=2.0in]{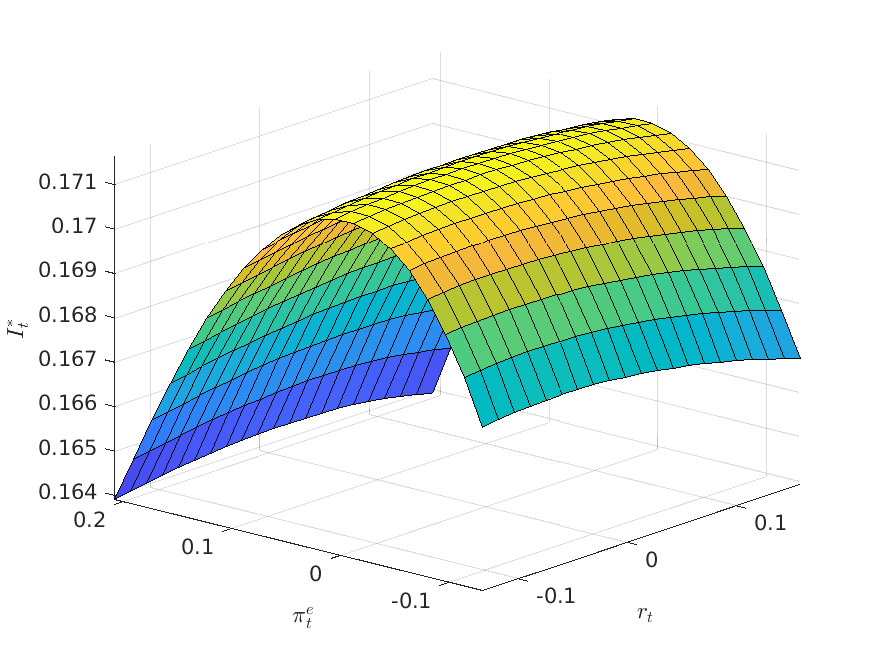}
	\label{fig_I_3D_theta0:life_insurance}
    }
    %\hspace{-0.35in}
    \subfigure[$I^*$ (age 90)]
    {
	\includegraphics[width=2.0in,height=2.0in]{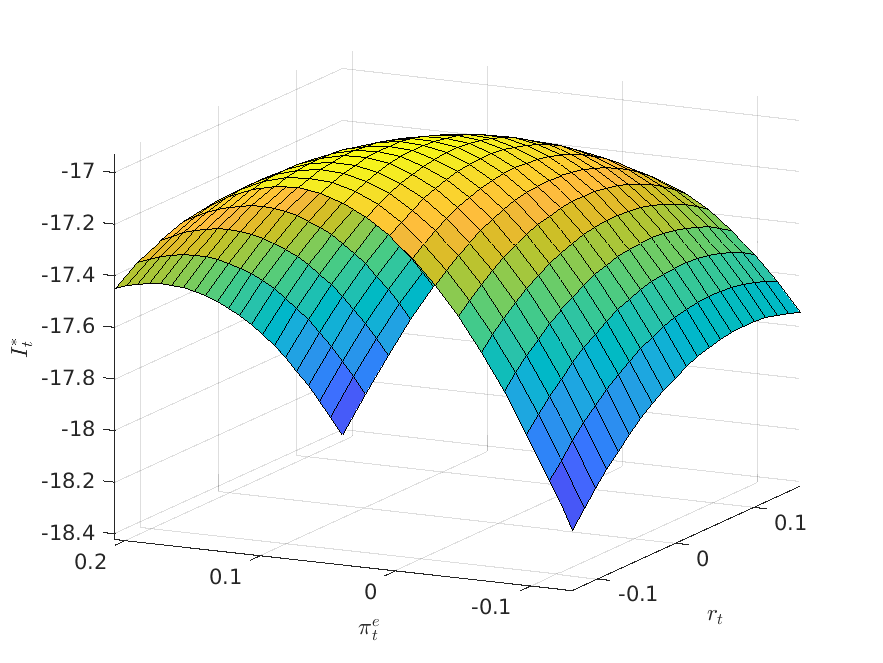}
	\label{fig_I_3D_theta0:annuity}
    }
    %\hspace{-0.35in}
    \subfigure[Insurance face value (age 40)]
    {
    		\includegraphics[width=2.0in,height=2.0in]{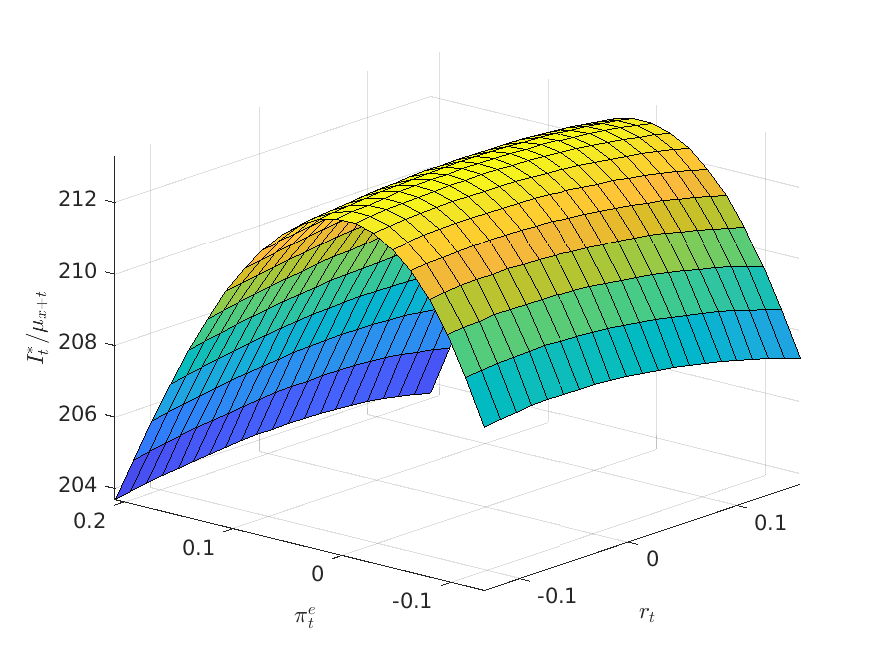}
    }
    %\hspace{-0.35in}    
    \subfigure[Insurance face value (age 90)]
    {
    		\includegraphics[width=2.0in,height=2.0in]{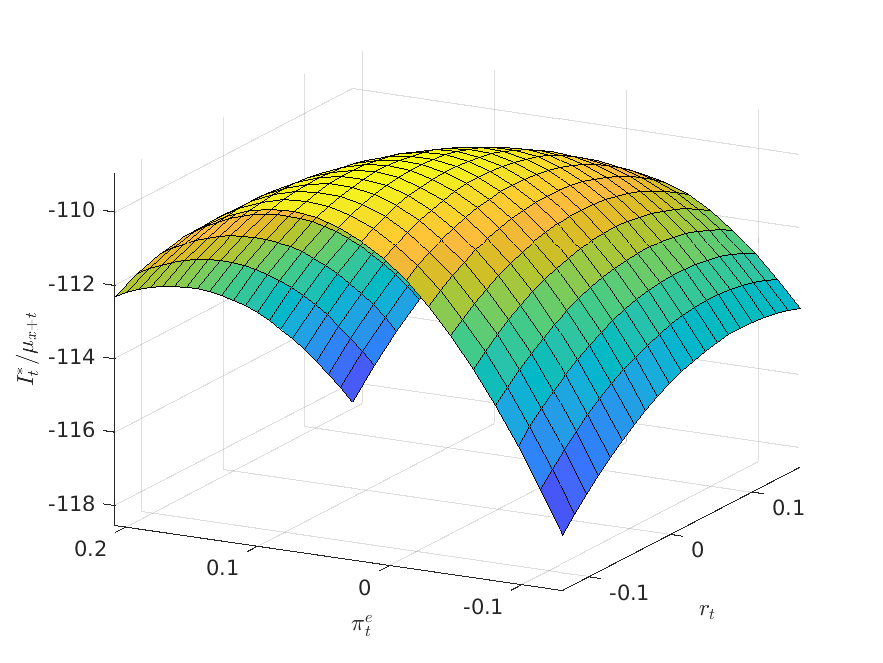}
    }
    \vspace{-0.15in}
    \subfigure[Bequest-wealth ratio (age 40)]
    {
    		\includegraphics[width=2.0in,height=2.0in]{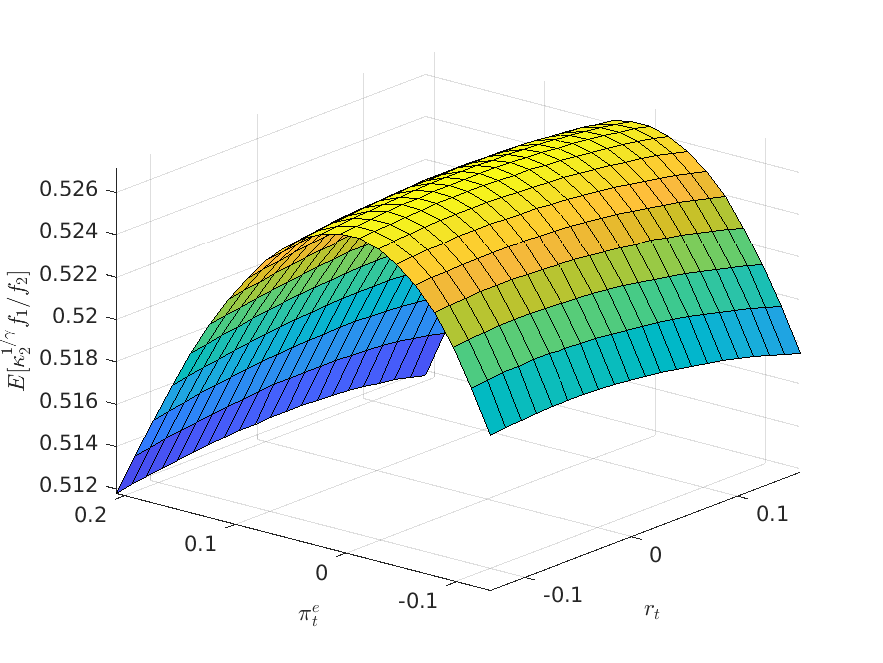}
    }
    %\hspace{-0.35in}
    \subfigure[Bequest-wealth ratio (age 90)]
    {
    		\includegraphics[width=2.0in,height=2.0in]{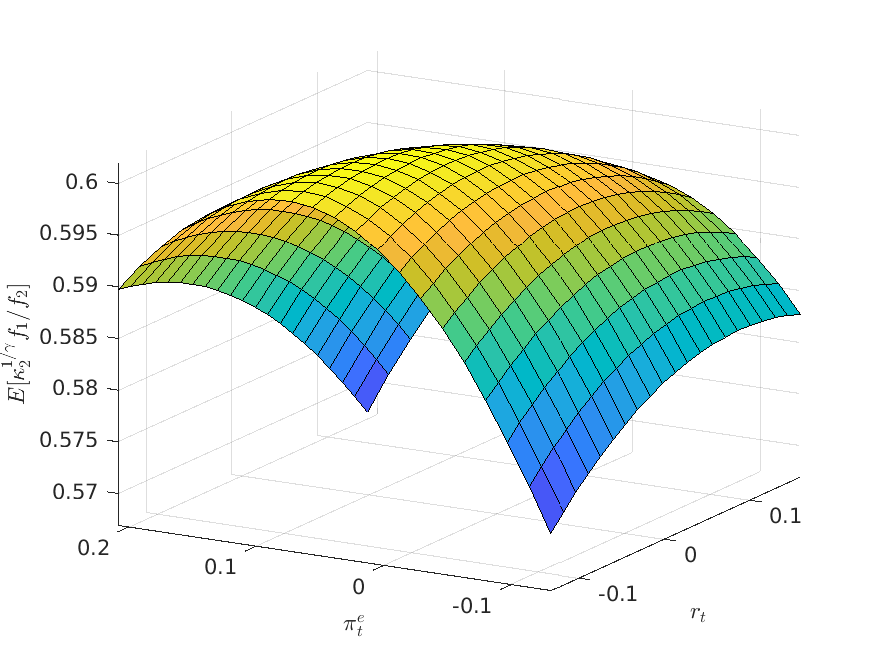}
    }
    %\hspace{-0.35in}
    \subfigure[future income (age 40)]
    {
    		\includegraphics[width=2.0in,height=2.0in]{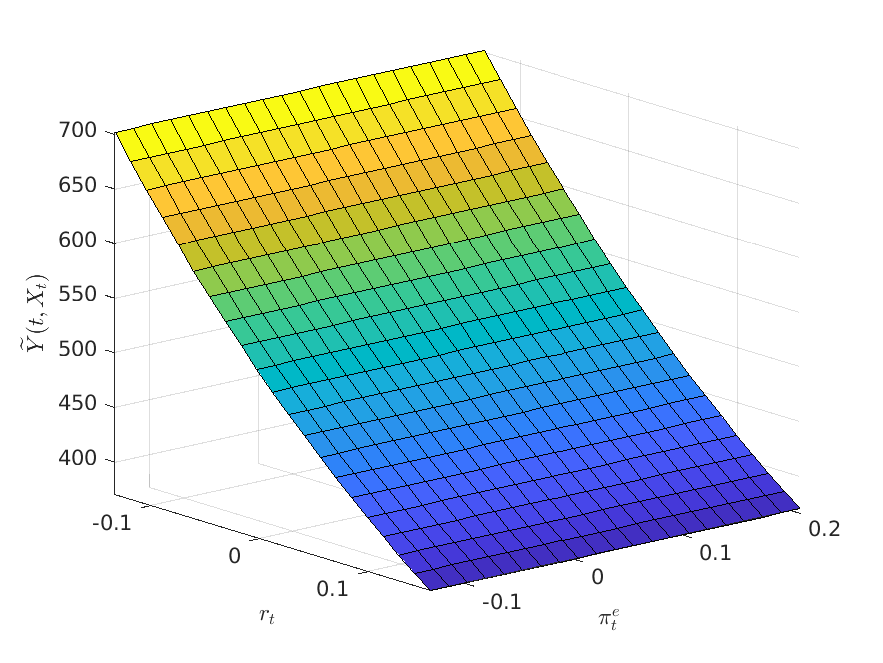}
    }
    %\hspace{-0.35in}
    \subfigure[future income (age 90)]
    {
    		\includegraphics[width=2.0in,height=2.0in]{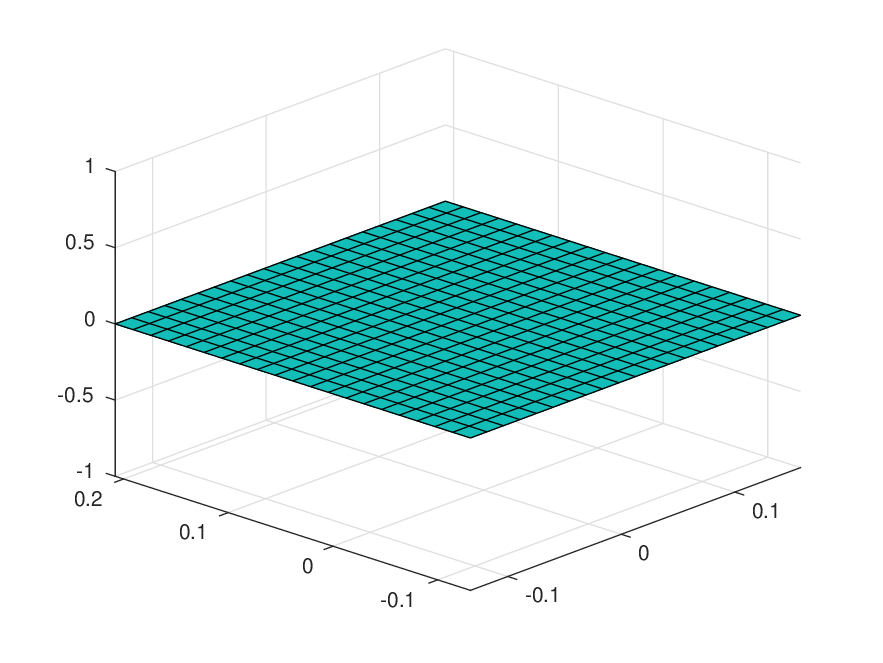}
    }
    \caption{Family's optimal consumption and insurance strategies when $\theta=0$. }
	\label{fig_I_3D_theta0}
\end{sidewaysfigure}

\begin{sidewaysfigure}[htbp]
    \centering
    \subfigure[$c^*_1$ (age 40)]
    {
        \includegraphics[width=2.0in,height=2.0in]{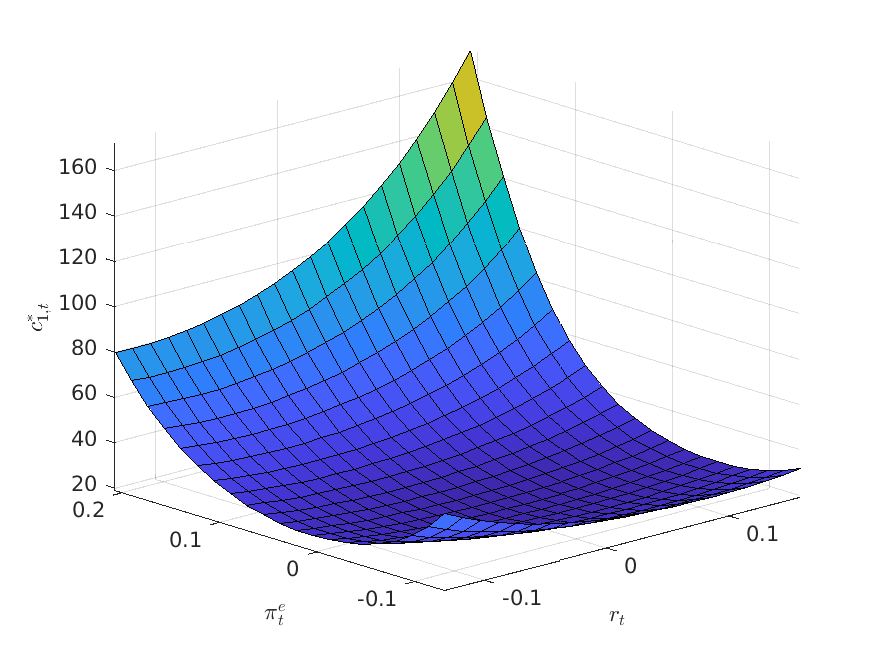}
    }
    %\hspace{-0.35in}	
    \subfigure[$c^*_1$ (age 90)]
    {
	\includegraphics[width=2.0in,height=2.0in]{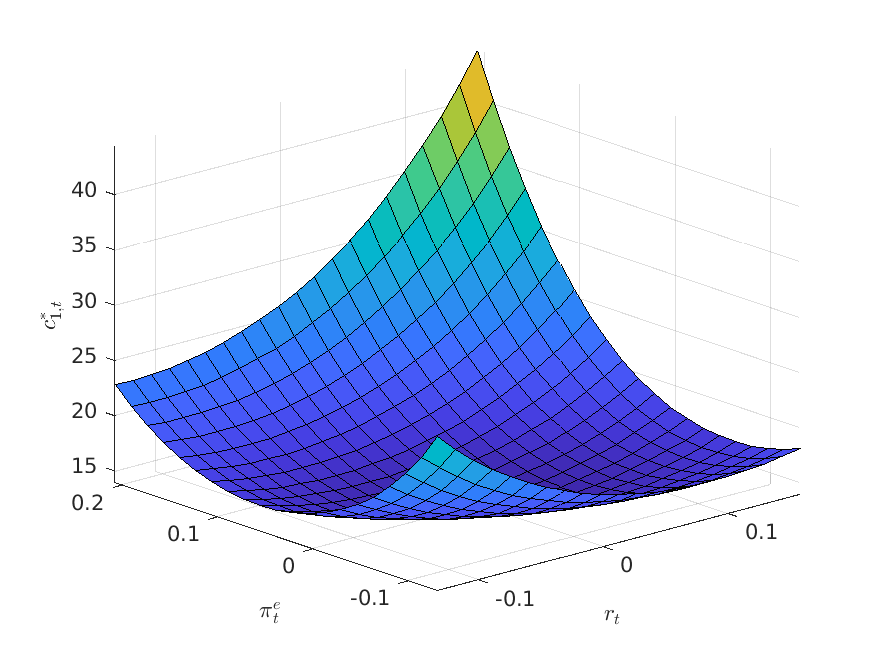}
    }
    %\hspace{-0.35in}
    \subfigure[$c^*_2$ (age 40)]
    {
	\includegraphics[width=2.0in,height=2.0in]{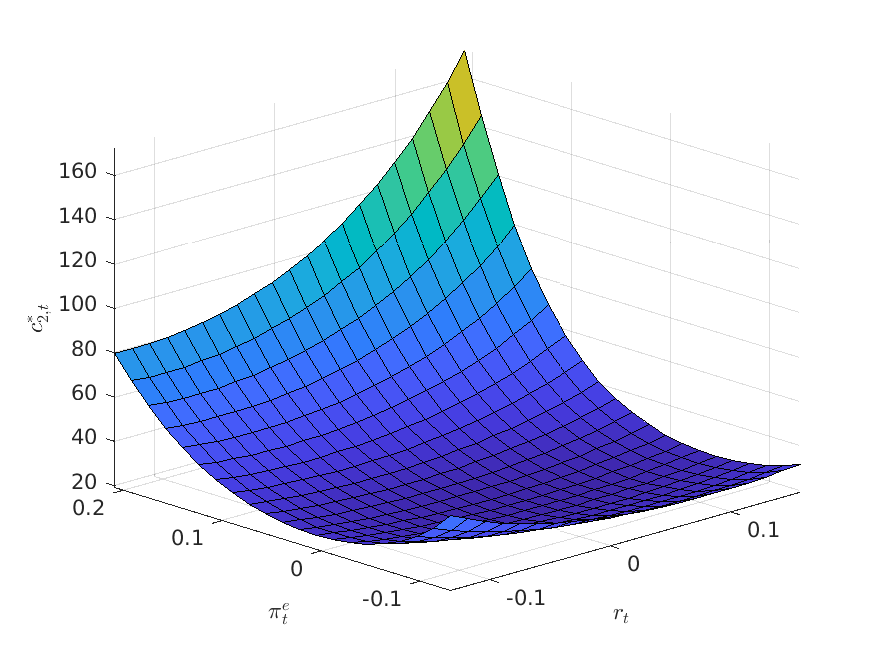}
    }
    %\hspace{-0.35in}
    \subfigure[$c^*_2$ (age 90)]
    {
	\includegraphics[width=2.0in,height=2.0in]{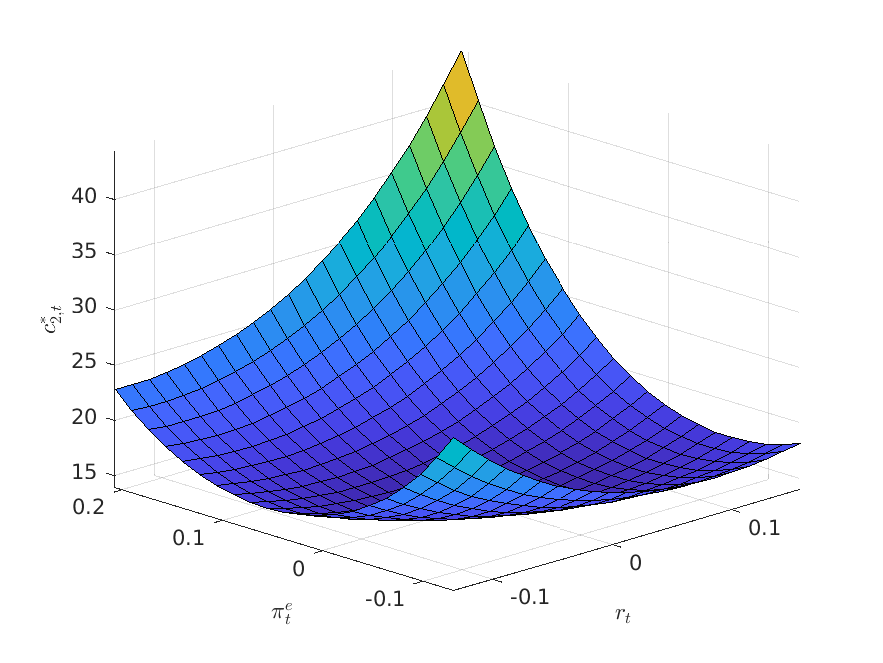}
    }
    \vspace{-0.15in}
    \subfigure[$I^*$ (age 40)]
    {
	\includegraphics[width=2.0in,height=2.0in]{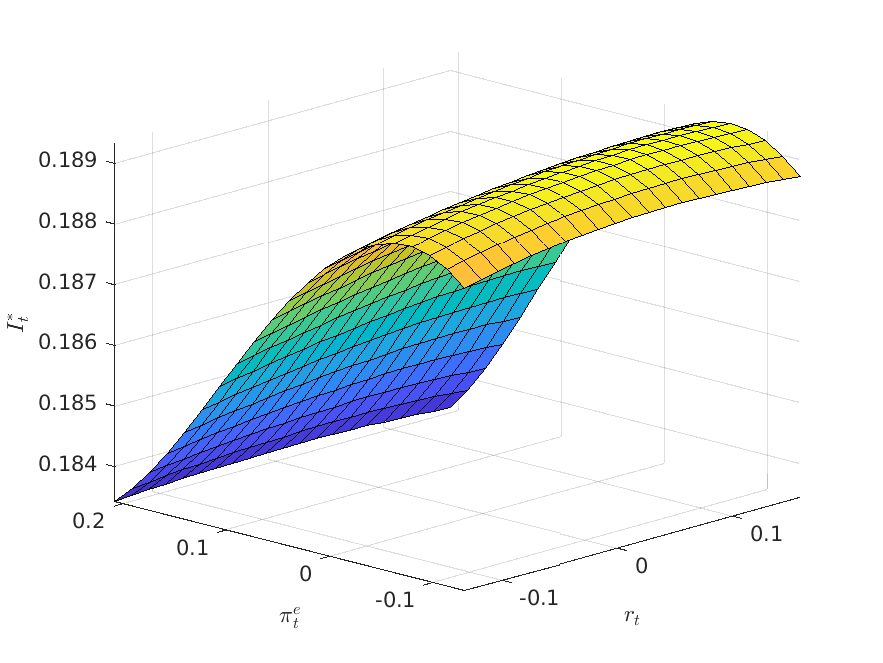}
    }
    %\hspace{-0.35in}
    \subfigure[$I^*$ (age 90)]
    {
	\includegraphics[width=2.0in,height=2.0in]{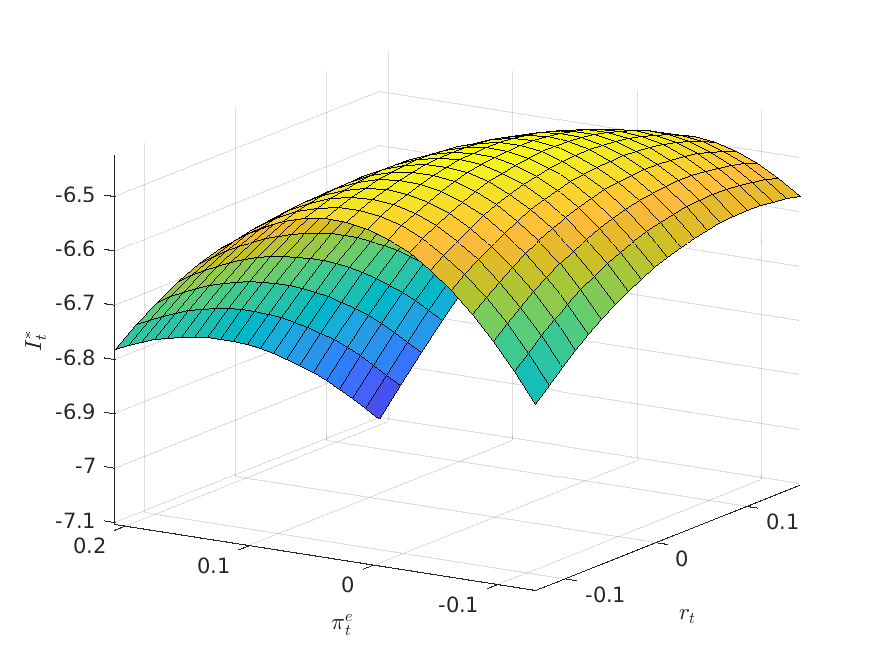}
    }
    %\hspace{-0.35in}
    \subfigure[Insurance face value (age 40)]
    {
    		\includegraphics[width=2.0in,height=2.0in]{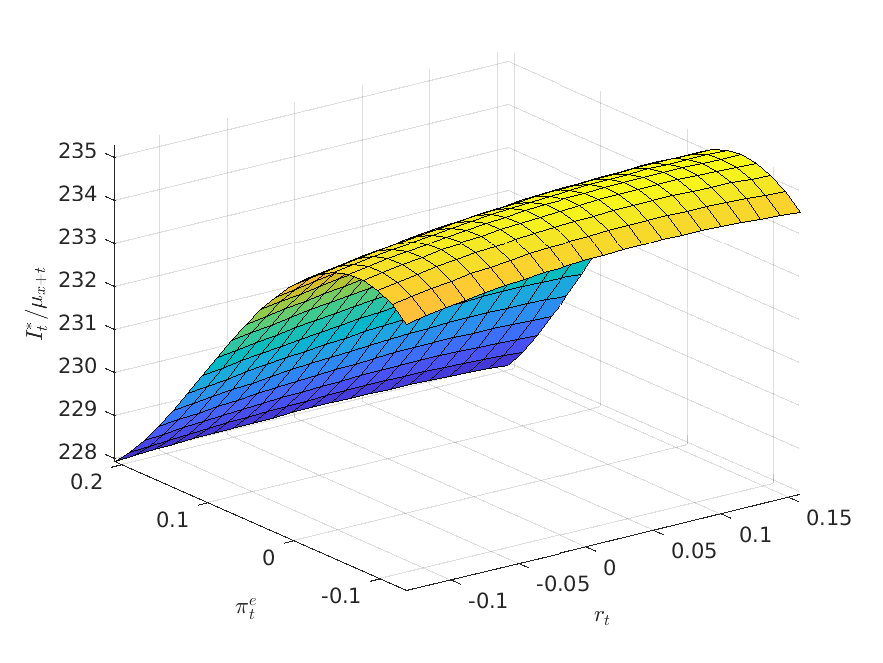}
    }
    %\hspace{-0.35in}    
    \subfigure[Insurance face value (age 90)]
    {
    		\includegraphics[width=2.0in,height=2.0in]{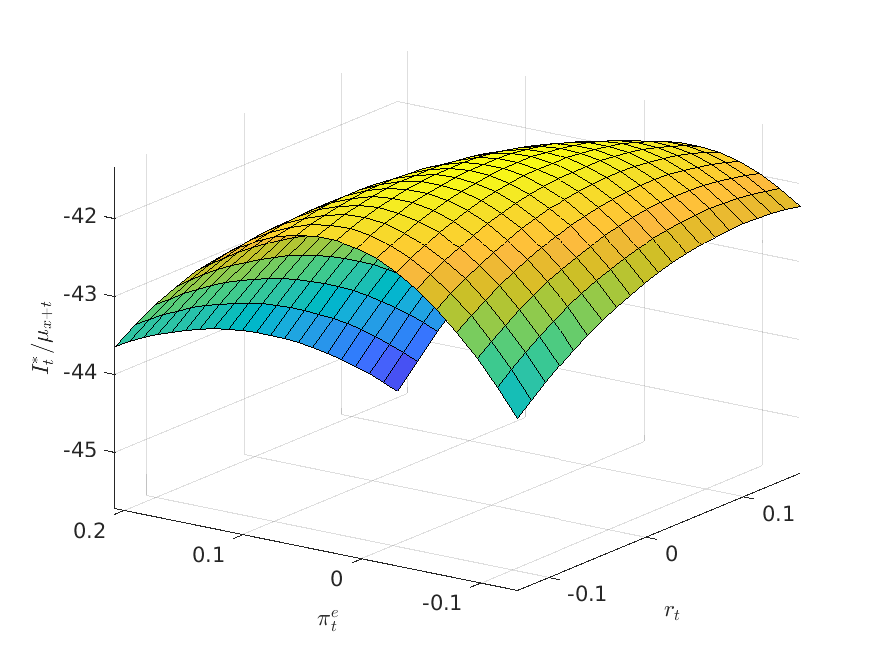}
    }
    \vspace{-0.15in}
    \subfigure[Bequest-wealth ratio (age 40)]
    {
    		\includegraphics[width=2.0in,height=2.0in]{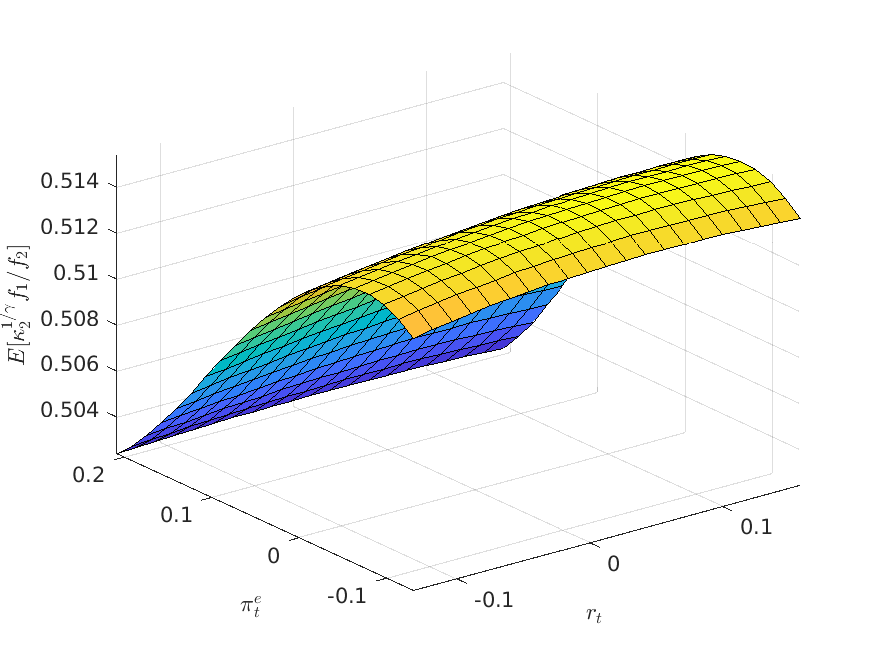}
    }
    %\hspace{-0.35in}
    \subfigure[Bequest-wealth ratio (age 90)]
    {
    		\includegraphics[width=2.0in,height=2.0in]{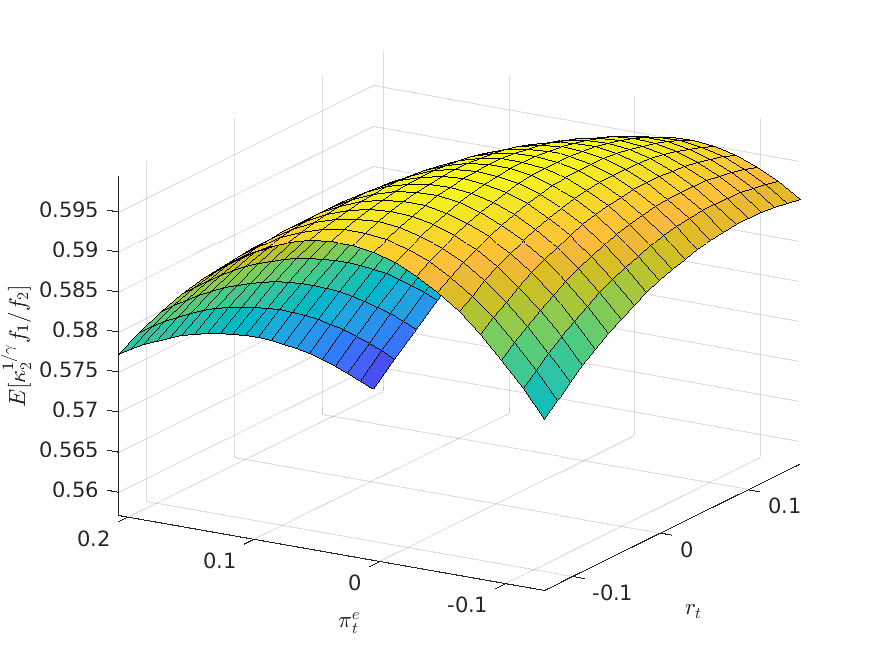}
    }
    %\hspace{-0.35in}
    \subfigure[future income (age 40)]
    {
    		\includegraphics[width=2.0in,height=2.0in]{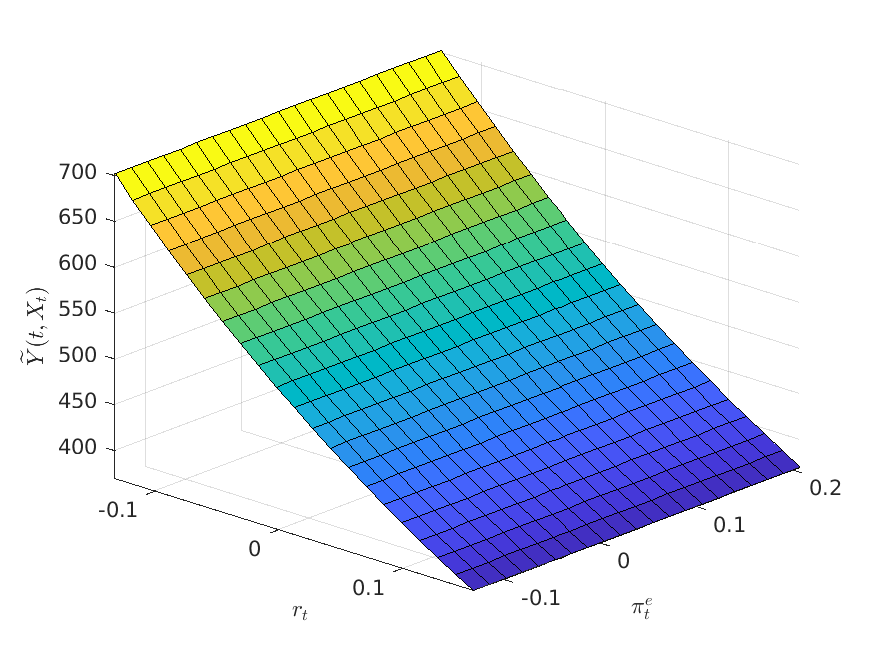}
    }
    %\hspace{-0.35in}
    \subfigure[future income (age 90)]
    {
    		\includegraphics[width=2.0in,height=2.0in]{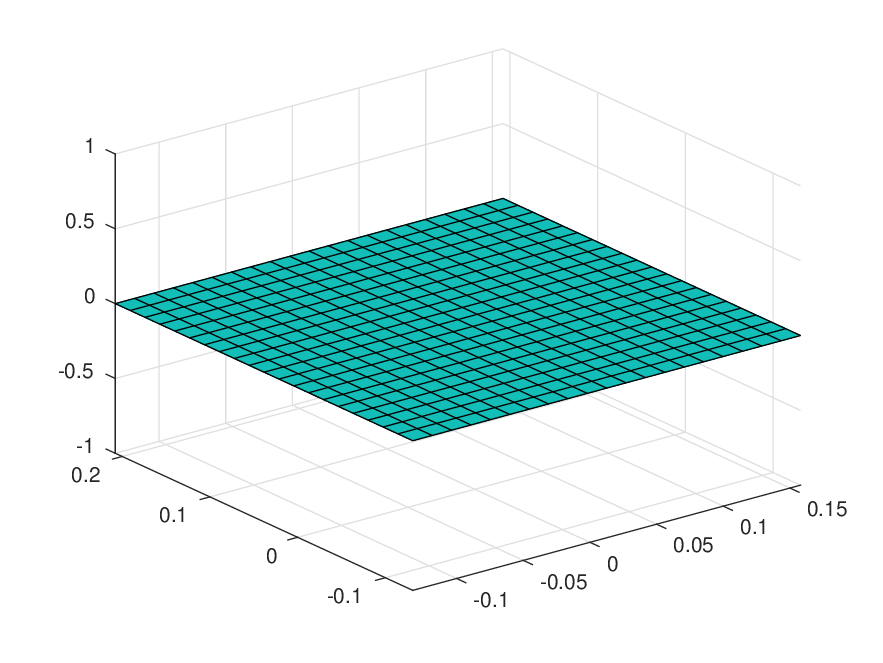}
    }
    \caption{Family's optimal consumption and insurance strategies when $\theta=0.8$. }
	\label{fig_I_3D_theta08}
\end{sidewaysfigure}

\subsection{Welfare loss}
This section evaluates the welfare loss resulting from the money illusion. We presume the family is non-illusioned, demonstrating rational preferences and full awareness of inflation risk. Specifically, the family's objective follows \eqref{objective} under $\theta = 0$
\begin{equation*}
 \sup \limits_{\alpha,c_1,c_2,I} E\left[ \kappa_1 \int_0^{T} {_tp_x} e^{-\delta t} \frac{c^{1-\gamma}_{1,t}}{1-\gamma}dt+ \kappa_2 \int_0^T e^{-\delta t}\frac{c^{1-\gamma}_{2,t}}{1-\gamma}dt\right].
\end{equation*}
We use $V(0,W^{\widetilde{Y}}_0,\Pi_0,X_0; \theta = 0)$ to denote the value function when non-illusioned family adopts strategies \eqref{primary_opt_c1_2} - \eqref{primary_opt_I_2} under $\theta=0$. Additionally, we denote $V^{\text{sub}}(0,W^{\widetilde{Y}}_0,X_0;\theta)$ for the non-illusioned family forced to adopt strategies \eqref{primary_opt_c1_2} - \eqref{primary_opt_I_2} under money-illusion degree $\theta$. Obviously, the non-illusioned family will obtain a lower value function under the money-illusioned strategies, i.e.,
\begin{equation*}
 V^{\text{sub}}(0,W^{\widetilde{Y}}_0,X_0;\theta)< V(0,W^{\widetilde{Y}}_0,\Pi_0,X_0; \theta = 0).
\end{equation*}
Inspired by \cite{basak2010equilibrium}, \cite{xue2019derivatives}, \cite{tan2020optimal}, and \cite{wei2023optimal}, we define the welfare loss $L(\theta)$ as
\begin{equation}
 V^{\text{sub}}(0,W^{\widetilde{Y}}_0,X_0;\theta)= V(0,W^{\widetilde{Y}}_0(1-L(\theta)),\Pi_0,X_0; \theta = 0).\label{welfare_loss_1}
\end{equation}
Substitute \eqref{adj_G} into \eqref{welfare_loss_1}, we have
\begin{equation*}
    V^{\text{sub}}(0,W^{\widetilde{Y}}_0,X_0;\theta)= \frac{1}{1-\gamma}\{W^{\widetilde{Y}}_0[1-L(\theta)]\}^{1-\gamma}[f_2(0,X_0;\theta=0)]^{\gamma},
\end{equation*}
where $f_2(0,X_0;\theta=0)$ is the function $f_2(0,X_0)$ under $\theta=0$. Solving it, we obtain the formula for welfare loss
\begin{equation*}
    L(\theta) = 1 - \frac{[(1-\gamma)V^{\text{sub}}(0,W^{\widetilde{Y}}_0,X_0;\theta)]^{\frac{1}{1-\gamma}}}{[f_2(0,X_0;\theta=0)]^{\frac{\gamma}{1-\gamma}} W^{\widetilde{Y}}_0}.
\end{equation*}
We employ the Monte Carlo method to simulate $V^{\text{sub}}(0,W^{\widetilde{Y}}_0,X_0;\theta)$, with the results depicted in Figure \ref{fig_welfare_loss}. The findings reveal that welfare loss escalates with money illusion, and the rate of increase varies with relative risk aversion. For a family with $\gamma=3$, the maximum welfare loss is less than $30\%$. For a family with $\gamma=5$, a $50\%$ welfare loss is experienced when the money illusion degree $\theta$ attains $0.8$. For a family with $\gamma=10$, a $50\%$ welfare loss is incurred earlier when $\theta$ reaches $0.37$. Generally, when a family becomes more risk-averse, they decrease their allocation to risky assets. The non-illusioned family perceives inflation-linked bonds as risk-free assets, while the illusioned family leans towards nominal bonds. Consequently, heightened risk aversion amplifies the welfare loss of money illusion.

\begin{figure}[htbp]	
	  \includegraphics[width=4.5in,height=3.7in]{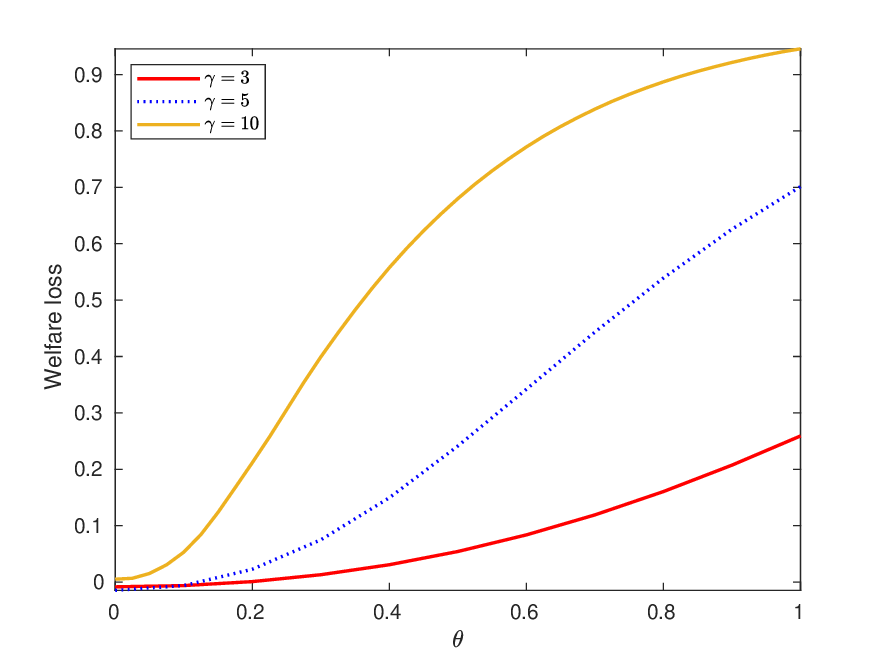}
	\caption{Welfare loss from money illusion.}\label{fig_welfare_loss}
\end{figure}

\section{Conclusion}\label{conclusion}
This paper investigates a life-cycle model under the money illusion, where households exhibit a preference for nominal over real money. The household can invest a part of the money in nominal bonds, inflation-linked bonds, a stock index, and a nominal cash account and use the other part of the money to purchase life insurance and annuities. We formulate this problem as a random-horizon utility maximization problem and derive its corresponding explicit solutions under CRRA utility. Our model, calibrated to U.S. data, illustrates that money illusion elevates life insurance demand for young adults while diminishing annuity demand for retirees. Sensitivity analysis reveals that annuity demand exhibits an upward ``U-shape" with respect to the real interest rate and expected inflation, consistent with the upward ``U-shape" of consumption but contrasting with the downward ``U-shape" of life insurance. Lastly, numerical simulations show that the welfare loss from the money illusion is significant, regardless of the risk aversion coefficient. In general, our paper enriches the existing literature by showing that money illusion can contribute to annuity puzzles. We recommend that insurance companies enhance their educational efforts regarding inflation risk to encourage voluntary annuity purchases among retirees.

\bibliographystyle{apalike}
\bibliography{ref}

\appendix

\section{Proofs of Proposition \ref{prop1} to Proposition \ref{prop3}}\label{appendix1}
\begin{proof}
    We first prove the Proposition \ref{prop3}, i.e.,  $G(t,W^{\widetilde{Y}}_t,\Pi_t,X_t)$ is the candidate solution to the HJB equation \eqref{primary_adj_HJB_1}. The derivatives of the candidate solution are given by
	\begin{align*}
		&\frac{\partial G}{\partial t} = \frac{\gamma}{1-\gamma} \bigg(\frac{f_2}{w^{\widetilde{Y}}}\bigg)^{\gamma-1}\frac{\partial f_2}{\partial t}\pi^{\theta(1-\gamma)}, ~\frac{\partial G}{\partial w^{\widetilde{Y}}} = \bigg(\frac{f_2}{w^{\widetilde{Y}}}\bigg)^{\gamma}\pi^{\theta(1-\gamma)},~\frac{\partial G}{\partial \pi} = \theta (w^{\widetilde{Y}})^{1-\gamma}\pi^{\theta(1-\gamma)-1}f^{\gamma}_2,\\
		&\frac{\partial G}{\partial X^{\top}} = \frac{\gamma}{1-\gamma} \bigg(\frac{f_2}{w^{\widetilde{Y}}}\bigg)^{\gamma-1}\frac{\partial f_2}{\partial X^{\top}}\pi^{\theta(1-\gamma)}, ~\frac{\partial^2 G}{(\partial w^{\widetilde{Y}})^2} = -\gamma (w^{\widetilde{Y}})^{-\gamma-1} f_2^{\gamma}\pi^{\theta(1-\gamma)}, \\
        &\frac{\partial^2 G}{(\partial \pi)^2} = \theta[\theta(1-\gamma)-1](w^{\widetilde{Y}})^{1-\gamma}\pi^{\theta(1-\gamma)-2}f^{\gamma}_2,\\
		&\frac{\partial^2 G}{\partial X^{\top}\partial X} = - \gamma (w^{\widetilde{Y}})^{1-\gamma}\pi^{\theta(1-\gamma)} f^{\gamma-2}_2 \frac{\partial f_2}{\partial X^{\top}}\frac{\partial f_2}{\partial X} + \frac{\gamma}{1-\gamma}(w^{\widetilde{Y}})^{1-\gamma} \pi^{\theta(1-\gamma)} f^{\gamma-1}_2\frac{\partial^2 f_2}{\partial X^{\top} \partial X},\\
  &\frac{\partial^2 G}{\partial w^{\widetilde{Y}}\partial \pi}= \theta (1-\gamma)(w^{\widetilde{Y}})^{-\gamma}\pi^{\theta(1-\gamma)-1}f^{\gamma}_2, ~\frac{\partial^2 G}{\partial w^{\widetilde{Y}} \partial X^{\top}} = \gamma (w^{\widetilde{Y}})^{-\gamma} \pi^{\theta(1-\gamma)}f^{\gamma-1}_2 \frac{\partial f_2}{\partial X^{\top}},\\
  &\frac{\partial^2 G}{\partial \pi \partial X^{\top}} = \theta \gamma (w^{\widetilde{Y}})^{1-\gamma} \pi^{\theta(1-\gamma)-1}f^{\gamma-1}_2 \frac{\partial f_2}{\partial X^{\top}}.
	\end{align*}
    Substitute these derivatives into the HJB equation \eqref{primary_adj_HJB_1}, we can verify that the equality holds. Thus, $G(t,W^{\widetilde{Y}}_t,\Pi_t,X_t)$ is indeed the candidate solution to the HJB equation \eqref{primary_adj_HJB_1}. By inserting $G(t,W^{\widetilde{Y}}_t,\Pi_t,X_t)$ into \eqref{primary_opt_c1_1}-\eqref{primary_opt_I_1}, we can then derive the optimal strategies \eqref{primary_opt_c1_2}-\eqref{primary_opt_I_2}.
    
    For Proposition \ref{prop1}, $G_1$ is the special case of $G$ when $\widetilde{Y}(t,X_t)\equiv 0$, $\kappa_1=0$, and $\kappa_2=1$. For Proposition \ref{prop2}, $G_2$ is the special case of $G$ when $\widetilde{Y}(t,X_t)\equiv 0$. The proofs for these two propositions can be established based on the derivation above. 
\end{proof} 

\section{Proof of Proposition \ref{global_exist_gamma_great_1}}\label{appendix2}
We can establish the global existence of $\Gamma_2(\tau)$ by utilizing Theorems 4.1.4. and 4.1.6. from \cite{abou2012matrix}. It is worth noting that their comparison theorem can be easily adapted from a semi-definite matrix case to a definite matrix case. Since $\Gamma_2(\tau)$ is both existent and negative for all $\tau \in (0,T]$, the candidate solution $G(t,W^{\widetilde{Y}}_t,\Pi_t,X_t)$ in \eqref{adj_G} also exists.

To prove the verification theorem, we define the value process for any $(\beta_t, c_{1,t}, c_{2,t}, I_t)\in \mathcal{A}_{\gamma}(0, T_R)$ and $s\in[t, T]$
\begin{align}\label{value_process}
      &g^{\beta,I,c_1,c_2}(s,W^{\widetilde{Y}}_s,\Pi_s,X_s):=\int_t^s {_{u-t}p_{x+t}}e^{-\delta(u-t)}\left[\kappa_1 U_1(c_{1,u},\Pi_u)+\kappa_2U_2(c_{2,u},\Pi_u)\right. \notag\\
      &\left.+\kappa_2\lambda_{x+u}G_1\bigg(u,W^{\widetilde{Y}}_u - \widetilde{Y}(u,X_u) + \frac{I_u}{\lambda_{x+u}},\Pi_u,X_u\bigg)\right] du + e^{-\delta(s-t)}{_{s-t}p_{x+t}}G(s,W^{\widetilde{Y}}_s,\Pi_s,X_s).
\end{align}
By Ito's formula, we have
\begin{align}\label{g_bI_SDE_gamma_01}
      &dg^{\beta,I,c_1,c_2}(s,W^{\widetilde{Y}}_s,\Pi_s,X_s) = {_{s-t}p_{x+t}}e^{-\delta(s-t)}\left\{\kappa_1 U_1(c_{1,s},\Pi_s) + \kappa_2 U_2(c_{2,s},\Pi_s) \right.\notag\\
      & \left.+\kappa_2 \lambda_{x+s} G_1\left(s,W^{\widetilde{Y}}_s-\widetilde{Y}(s,X_s)+\frac{I_s}{\lambda_{x+s}},\Pi_s,X_s\right) - (\lambda_{x+s}+\delta) G(s,W^{\widetilde{Y}}_s,\Pi_s,X_s)\right.\notag\\
      & \left. + \mathcal{D}^{\beta,I,c_1,c_2}G(s,W^{\widetilde{Y}}_s,\Pi_s,X_s)\right\}ds + g^{\beta,I,c_1,c_2}(s,W^{\widetilde{Y}}_s,\Pi_s,X_s) h^{\beta,I,c_1,c_2}(s,W^{\widetilde{Y}}_s,\Pi_s,X_s) dZ_s,
\end{align}
where $\mathcal{D}^{\beta,I,c_1,c_2}$ is the infinitesimal generator defined in \eqref{infi_generator_primal} and $h^{\beta,I,c_1,c_2}$ satisfies
\begin{align}
&h^{\beta,I,c_1,c_2}(s,W^{\widetilde{Y}}_s,\Pi_s,X_s) =\notag\\
&\frac{{_{s-t}p_{x+t}}e^{-\delta(s-t)}G(s,W^{\widetilde{Y}}_s,\Pi_s,X_s)}{g^{\beta,I,c_1,c_2}(s,W^{\widetilde{Y}}_s,\Pi_s,X_s)}\bigg[(1-\gamma)(\beta^{\top}_s\Sigma-\sigma^{\top}_{\Pi}) +(1-\gamma)\theta \sigma^{\top}_{\Pi} +\frac{\gamma}{f_2} \frac{\partial f_2}{\partial X}\Sigma_X  \bigg]. \label{h_bI}
\end{align}
Next, fix $(t,w^{\widetilde{Y}},\pi,X)\in[0,T]\times[0,\infty) \times[0,\infty)\times \mathbb{R}^2$ and denote the conditional expectation of the value process as
\begin{eqnarray*}\label{tilde_J}
      &&J(t,w^{\widetilde{Y}},\pi,X) :=  E_{t,w^{\widetilde{Y}},\pi,X}\left[ \int_t^{T_R} 
      {_{s-t}p_{x+t}} e^{-\delta(s-t)}\left[\kappa_1 U_1(c_{1,s},\Pi_s) + \kappa_2 U_2(c_{2,s},\Pi_s)+\kappa_2\lambda_{x+s}\right.\right.\notag\\
      &&\left.\left.\Phi_1\left(s,W^{\widetilde{Y}}_t - \widetilde{Y}(t,X_t) + \frac{I_t}{\mu_{x+t}},\Pi_s,X_s\right)\right]ds + e^{-\delta(T_R-t)}{_{T_R-t}p_{x+t}} \Phi_2\left(T_R,W^{\widetilde{Y}}_{T_R},\Pi_{T_R},X_{T_R}\right)\right],
\end{eqnarray*}
where $E_{t,w^{\widetilde{Y}},\pi,X}[\cdot]$ is short for $E[\cdot|W^{\widetilde{Y}}_t = w^{\widetilde{Y}}, \Pi_t = \pi, X_t=X]$.
Then, we have
\begin{equation}
      V(t,W^{\widetilde{Y}}_t,\Pi_t,X_t) = \sup \limits_{(\beta,I,c_1,c_2)\in \mathcal{A}_{\gamma}(0,T_R)}J(t,W^{\widetilde{Y}}_t,\Pi_t,X_t). \label{VJ_equality}
\end{equation}

For candidate solutions, $G_1$ is the special case of $G$ when $\widetilde{Y}(t,X_t)\equiv 0$, $\kappa_1=0$, and $\kappa_2=1$. Moreover, $G_2$ is the special case of $G$ when $\widetilde{Y}(t,X_t)\equiv 0$. One can easily verify $\Phi_1 = G_1$ and $\Phi_2 = G_2$ first, then use the same approach to verify $G=V$ by induction. We only verify $G=V$ in the following part since $G$ has the most complex form. 

The proof verification theorem has three steps:
\begin{enumerate}[Step 1:]
	
	\item Verify the optimal strategy $(\beta^*,I^*,c^*_1,c^*_2)$ is in the admissible set $\mathcal{A}_{\gamma}(0,T_R)$.
	
	Recall from \eqref{primary_opt_beta_2}
	\begin{equation*}
		   \beta^*_t = \frac{(\Sigma^{\top})^{-1}}{\gamma}\left[\Lambda_t -\sigma_{\Pi} + \theta(1-\gamma)\sigma_{\Pi}+ 
              \Sigma^{\top}_X \frac{\gamma}{f_2(t,X_t)}\frac{\partial f_2(t,X_t)}{\partial X}\right] + (\Sigma^{\top})^{-1}\sigma_{\Pi},
	\end{equation*}
	which satisfies a linear growth with $X_t$ due to the linear growth of $(\Lambda_t-\sigma_{\Pi})$ and $\frac{1} 
        {f_2}\frac{\partial f_2}{\partial X^{\top}}$. Next, substitute \eqref{primary_opt_c1_2} - \eqref{primary_opt_I_2} into \eqref{WY_1}, we have
	\begin{eqnarray}\label{WY_opt_SDE}
	       &&d(W^{\widetilde{Y}}_t)^* =  (W^{\widetilde{Y}}_t)^* \bigg[ r_t + \lambda_{x+t}\bigg(1-\kappa^{\frac{1} 
                {\gamma}}_2\frac{f_1}{f_2}\bigg)-\left(\kappa^{\frac{1}{\gamma}}_1+\kappa^{\frac{1}{\gamma}}_2\right)\frac{1}{f_2} + (\eta_t)^{\top}(\Lambda_t - \sigma_{\Pi})\bigg]dt \notag\\
              &&+ (W^{\widetilde{Y}}_t)^* (\eta_t)^{\top}dZ_t,
	\end{eqnarray}
	where $(\eta_t)^{\top} = \frac{1}{\gamma} [\Lambda^{\top}_t - \sigma^{\top}_{\Pi}+\theta(1-\gamma)\sigma^{\top}_{\Pi}] + 
        \frac{1}{f_2} \frac{\partial f_2}{\partial X}\Sigma_X$. Given that the drift term and volatility term of SDE \eqref{WY_opt_SDE} are almost surely sample continuous, we can utilize Proposition 1.1 in \cite{kraft2004optimal} to demonstrate the existence of a unique strong solution for SDE \eqref{WY_1} under $(\beta^*,I^*,c^*_1,c^*_2)$. As a result, we can conclude that $(\beta^*,I^*,c^*_1,c^*_2)\in \mathcal{A}_{\gamma}(0,T_R)$.
	
	\item Verify $J(t,W^{\widetilde{Y}}_t,\Pi_t,X_t) \leq G(t,W^{\widetilde{Y}}_t,\Pi_t,X_t)$ for any 
        $(\beta,I,c_1,c_2)\in \mathcal{A}_{\gamma}(0,T_R)$.

	We first introduce the following useful lemma.
	\begin{lemma} \label{martingale_lemma}
		Assume a $n$-dimensional stochastic process $\widetilde{X}_t$ is driven by a $m$-dimensional Brownian motion $\widetilde{Z}$
		\begin{eqnarray*}
			   d\widetilde{X}_t = \mu(t,\widetilde{X}_t) dt + \sigma(t) d\widetilde{Z}_t,~~ \widetilde{X}_0 = 
                  \widetilde{x}_0,
		\end{eqnarray*}
		where $\widetilde{x}_0$ is a constant n-dimensional vector, $\mu(t,\widetilde{X})$ is a borel function and $\sigma(t)$ a continuous function
		\begin{eqnarray*}
			   \mu(t,\widetilde{X}) : (0,\infty) \times \mathbb{R}^n \rightarrow \mathbb{R}^n,~~ \sigma(t) : (0,\infty) 
                  \rightarrow \mathbb{R}^n\times \mathbb{R}^m,
		\end{eqnarray*}
		satisfying
		\begin{eqnarray*}
			   ||\mu(t,\widetilde{X}_t)-\mu(t,\widetilde{Y}_t)||_2\leq k||\widetilde{X}_t-\widetilde{Y}_t||_2,\\
			   ||\mu(\cdot,0)||_2 + ||\sigma(\cdot)||_2 \in L^2(0,T;\mathbb{R}), \forall T>0,
		\end{eqnarray*}
		where $||\cdot||_2$ is the Euclidean norm and $L^2(0,T;\mathbb{R})$ represents the set of Lebesgue measurable function $\psi : [0,T]\rightarrow \mathbb{R}$, such that $\int_0^T |\psi(t)|^2 dt<\infty$. If a stochastic process $\widetilde{g}(t,\widetilde{X}_t)$, $\widetilde{g}:[0,T]\times \mathbb{R}^n \rightarrow \mathbb{R}^n$, grows linearly with respect to $\widetilde{X}_t$ ( $||\widetilde{g}(t,\widetilde{X}_t)||_2\leq c_0 + c_1 ||\widetilde{X}_t||_2$), then we have
		\begin{equation*}
			   E[\mathcal{E}(T,\widetilde{g})] = 1,
		\end{equation*}
		where
		\begin{equation*}
			   \mathcal{E}(t,\widetilde{g}) := \exp \bigg\{ \int_0^t [\widetilde{g}(s,\widetilde{X}_s)]^{\top}d\widetilde{Z}_s- \frac{1}{2}\int_0^t||\widetilde{g}(s,\widetilde{X}_s)||_2^2ds\bigg\}.
		\end{equation*}
	\end{lemma}
    \begin{proof}
    	The proof is an extension of Lemma 4.1.1. in \cite{bensoussan2004stochastic} to the case where $\widetilde{X}_t$ and $\mathcal{E}(t,\widetilde{g})$ share the same Brownian motion.
    \end{proof}

	Next, following \eqref{primary_adj_HJB_1}, \eqref{g_bI_SDE_gamma_01}, and $\Phi_1=G_1$, we have
	\begin{eqnarray}\label{g_bI_inequality_gamma_great_1}
		   g^{\beta,I,c_1,c_2}(T,W^{\widetilde{Y}}_{T},\Pi_{T},X_{T}) \leq g^{\beta,I,c_1,c_2}(t,W^{\widetilde{Y}}_t,\Pi_t,X_t) \frac{\mathcal{E}(T, h^{\beta,I,c_1,c_2})}{\mathcal{E}(t, h^{\beta,I,c_1,c_2})},
	\end{eqnarray}
	Recall from \eqref{h_bI}, for $s\in [t,T]$,
	\begin{eqnarray*}
              &&h^{\beta,I,c_1,c_2}(s,W^{\widetilde{Y}}_s,\Pi_s,X_s) =\notag\\
              &&\frac{{_{s-t}p_{x+t}}e^{-\delta(s-t)}G(s,W^{\widetilde{Y}}_s,\Pi_s,X_s)}{g^{\beta,I,c_1,c_2}(s,W^{\widetilde{Y}}_s,\Pi_s,X_s)}\bigg[(1-\gamma)(\beta^{\top}_s\Sigma-\sigma^{\top}_{\Pi}) +(1-\gamma)\theta \sigma^{\top}_{\Pi} +\frac{\gamma}{f_2} \frac{\partial f_2}{\partial X}\Sigma_X  \bigg].
	\end{eqnarray*}
	It is easy to prove that $h^{\beta,I,c_1,c_2}(s,W^{\widetilde{Y}}_s,\Pi_s,X_s)$ satisfies a linear growth with respect to $X_t$. By Lemma \ref{martingale_lemma}, $\mathcal{E}(t,h^{\beta,I,c_1,c_2})$ is a martingale. Given $\Phi_1 = G_1$, we can derive
	\begin{eqnarray}\label{G_inequlity2_gamma_great_1}
 	     &&J(t,w^{\widetilde{Y}},\pi,X) \notag\\
              &=&  E_{t,w^{\widetilde{Y}},\pi,X}\left[ \int_t^{T_R} 
                  {_{s-t}p_{x+t}} e^{-\delta(s-t)}\left[\kappa_1 U_1(c_{1,s},\Pi_s) + \kappa_2 U_2(c_{2,s},\Pi_s)\right.\right.\notag\\
              &&\left.\left.+\kappa_2\lambda_{x+s}\Phi_1\left(s,W^{\widetilde{Y}}_t - \widetilde{Y}(t,X_t) + \frac{I_t} 
                {\mu_{x+t}},\Pi_s,X_s\right)\right]ds\right.\notag\\
              &&\left.+ e^{-\delta(T_R-t)}{_{T_R-t}p_{x+t}} 
                 \Phi_2\left(T_R,W^{\widetilde{Y}}_{T_R},\Pi_{T_R},X_{T_R}\right)\right]\notag\\
		   &=&E_{t,w^{\widetilde{Y}},\pi,X}[g^{\beta,I,c_1,c_2}(T_R,W^{\widetilde{Y}}_{T_R},\Pi_{T_R},X_{T_R})]\notag\\
		   &\leq& E_{t,w^{\widetilde{Y}},\pi,X}\bigg[g^{\beta,I,c_1,c_2}(t,w^{\widetilde{Y}},\pi,X) \frac{\mathcal{E} 
                    (T_R, h^{\beta,I,c_1,c_2})}{\mathcal{E}(t, h^{\beta,I,c_1,c_2})}\bigg]\notag\\
		   &=& G(t,w^{\widetilde{Y}},\pi,X), ~\text{for}~ \forall (\beta,I,c_1,c_2)\in \mathcal{A}_{\gamma}(0,T_R).
	\end{eqnarray}
 
	\item Verify $V(t,W^{\widetilde{Y}}_t,\Pi_t,X_t) = G(t,W^{\widetilde{Y}}_t,\Pi_t,X_t)$ under the optimal strategy $(\beta^*,I^*,c^*_1,c^*_2)$.
	
	Since $(\beta^*_t,I^*_t,c^*_{1,t},c^*_{2,t})$ maximizes the HJB \eqref{primary_adj_HJB_1} and $G(t,W^{\widetilde{Y}}_t,\Pi_t,X_t)$ is the solution to \eqref{primary_adj_HJB_1}, the equality in \eqref{g_bI_inequality_gamma_great_1} holds
	\begin{equation*}\label{exp_g_gamma_01}
		   g^{\beta^*,I^*,c^*_1,c^*_2}(s',(W^{\widetilde{Y}}_{s'})^*,\Pi_{s'},X_{s'}) = g^{\beta^*,I^*,c^*_1,c^*_2}(s, 
              (W^{\widetilde{Y}}_s)^*,\Pi_s,X_s) \frac{\mathcal{E}(s', h^{\beta^*,I^*,c^*_1,c^*_2})}{\mathcal{E}(s, h^{\beta^*,I^*,c^*_1,c^*_2})}, 
	\end{equation*}
	for $s' \in [s,T]$, where
	\begin{eqnarray*}
              &&h^{\beta^*,I^*,c^*_1,c^*_2}(s,(W^{\widetilde{Y}}_s)^*,\Pi_s,X_s) =\notag\\
              &&\frac{{_{s-t}p_{x+t}}e^{-\delta(s-t)}G(s,(W^{\widetilde{Y}}_s)^*,\Pi_s,X_s)}{g^{\beta^*,I^*,c^*_1,c^*_2}(s,(W^{\widetilde{Y}}_s)^*,\Pi_s,X_s)}\bigg[\frac{1-\gamma}{\gamma}(\Lambda^{\top}_t-\sigma^{\top}_{\Pi}) +\frac{1-\gamma}{\gamma}\theta \sigma^{\top}_{\Pi} +\frac{1}{f_2} \frac{\partial f_2}{\partial X}\Sigma_X  \bigg].\label{opt_h_bI_star}
	\end{eqnarray*}
	It is straightforward to demonstrate that $h^{\beta^*,I^*,c^*_1,c^*_2}$ satisfies the linear growth condition. Thus, according to Lemma \ref{martingale_lemma}, $\mathcal{E}(t,h^{\beta^*,I^*,c^*_1,c^*_2})$ is a martingale. Considering $\Phi_1 = G_1$ and $\Phi_2 = G_2$, we can establish the following equality for $G(t,w^{\widetilde{Y}},\pi,X)$.
	\begin{eqnarray}\label{G_equality_gamma_great_1}
              &&V(t,w^{\widetilde{Y}},\pi,X)\notag\\
              &\geq& E_{t,w^{\widetilde{Y}},\pi,X}\left[ \int_t^{T_R} 
              {_{s-t}p_{x+t}} e^{-\delta(s-t)}\left[\kappa_1 U_1(c^*_{1,s},\Pi_s) + \kappa_2 U_2(c^*_{2,s},\Pi_s)\right.\right.\notag\\
              &&\left.\left.+\kappa_2\lambda_{x+s}\Phi_1\left(s,(W^{\widetilde{Y}}_t)^* - \widetilde{Y}(t,X_t) + \frac{I^*_t}{\mu_{x+t}},\Pi_s,X_s\right)\right]ds\right.\notag\\
              &&\left.+ e^{-\delta(T_R-t)}{_{T_R-t}p_{x+t}} \Phi_2\left(T_R,(W^{\widetilde{Y}}_{T_R})^*,\Pi_{T_R},X_{T_R}\right)\right]\notag\\
		   &=&E_{t,w^{\widetilde{Y}},\pi,X}[g^{\beta^*,I^*,c^*_1,c^*_2}(T_R, 
              (W^{\widetilde{Y}}_{T_R})^*,\Pi_{T_R},X_{T_R})]\notag\\
		   &=& E_{t,w^{\widetilde{Y}},\pi,X}\bigg[g^{\beta^*,I^*,c^*_1,c^*_2}(t,w^{\widetilde{Y}},\pi,X) \frac{\mathcal{E}(T_R, h^{\beta^*,I^*,c^*_1,c^*_2})}{\mathcal{E}(t, h^{\beta^*,I^*,c^*_1,c^*_2})}\bigg]\notag\\
		   &=& G(t,w^{\widetilde{Y}},\pi,X).
	\end{eqnarray}
By combining \eqref{G_inequlity2_gamma_great_1}, \eqref{G_equality_gamma_great_1}, and \eqref{VJ_equality}, we can verify that $V(t,W^{\widetilde{Y}}_t,\Pi_t,X_t)=$\\$G(t,W^{\widetilde{Y}}_t,\Pi_t,X_t)$, and $(\beta^*,I^*,c^*_1,c^*_2)$ obtained from \eqref{primary_opt_c1_2}-\eqref{primary_opt_I_2} represents the optimal strategy.
\end{enumerate}
The proof is complete.

\section{Proof of Proposition \ref{global_exist_gamma_01}}\label{appendix3}
\begin{proof}
By substituting $y = \lambda - \frac{b}{4}$ into \eqref{f_lambda}, we obtain a quartic equation
\begin{equation*}
      f_y(y) = y^4 + qy^2 + ry + s,
\end{equation*}
where the coefficients are given as follows
\begin{eqnarray*}
      q = \frac{8c - 3b^2}{8}, \quad r = \frac{b^3 - 4bc + 8d}{8}, \quad s =\frac{-3b^4 + 256j - 64bd + 16 b^2c}{256}.
\end{eqnarray*}
Moreover, the discriminant of $f_{y}(y)$ is defined as
\begin{eqnarray*}
      \widetilde{\Delta} &=& -4 q^3 r^2 - 27 r^4 + 256s^3 + 16q^4s + 144 qr^2s - 128q^2s^2.
\end{eqnarray*}
As stated in \cite{rees1922graphical}, if condition \eqref{root_condition} is satisfied, then \eqref{f_lambda} has four distinct real roots, which implies that the Hamiltonian matrix $H$ has four different real eigenvalues. This guarantees the diagonalizability of $H$ and the full rank of its eigenvector matrix $V$. According to Radon's lemma \citep[see Theorem 3.1.1. in][]{abou2012matrix}, we can express $\Gamma_2(\tau) = P(\tau)/Q(\tau)$, and the existence and negative definiteness of $\Gamma_2(\tau)$ can be derived from \eqref{det_condition}. Given that $\Gamma_2(\tau)$ exists and $\Gamma_2(\tau)<0$ for $\tau \in (0,T]$, the candidate solution $G(t,W^{\widetilde{Y}}_t,\Pi_t,X_t)$ in \eqref{adj_G} is globally existent. 

Similar to Appendix \ref{appendix2}, the verification theorem can be proven in three steps. However, there are two differences in this case. First, in Step 1, we need to verify that $(W^{\widetilde{Y}}_t)^*>0$. It is straightforward to observe that the solution to \eqref{WY_opt_SDE} satisfies
\begin{eqnarray*}\label{WY_bigger_0}
     (W^{\widetilde{Y}}_t)^* &=& W^{\widetilde{Y}}_0 \exp \left\{ \int_0^t \left[r_s + \lambda_{x+s}\bigg(1 - \kappa^{\frac{1}{\gamma}}_2\frac{f_1}{f_2}\bigg) -\left(\kappa^{\frac{1}{\gamma}}_1+\kappa^{\frac{1}{\gamma}}_2\right)\frac{1}{f_2}+ \eta^{\top}_s(\Lambda_s - \sigma_{\Pi})  \right.\right.\notag\\
     &&\left.\left. - \frac{1}{2} \eta^{\top}_s \eta_s \right] ds+ \int_0^t \eta^{\top}_s dZ_s \right\}>0,
\end{eqnarray*}
which satisfies the requirement of the admissible set \eqref{adset1_gamma_01}. The argument for the existence of a strong solution remains the same as Step 1 in Appendix \ref{appendix2}.

Second, in Step 2, we can adopt Fatou's lemma rather than Lemma \ref{martingale_lemma} to prove the inequality \eqref{G_inequlity2_gamma_great_1}, as the value process \eqref{value_process} is bounded below by zero when $0<\gamma<1$. Define
\begin{equation*}
      \Psi(s) := \int_t^s || g^{\beta,I,c_1,c_2}(u,W^{\widetilde{Y}}_u,\Pi_u,X_u) h^{\beta,I,c_1,c_2}(u,W^{\widetilde{Y}}_u,\Pi_u,X_u) ||^2_2du,
\end{equation*}
and let $\tau_n := T\wedge \inf\{s\in[t,T]|\Psi(s)\geq n\}$ for $n\in \mathbb{N}$. For $s\in [t,\tau_n]$, the stochastic integral $\int_t^s g^{\beta,I,c_1,c_2}(u,W^{\widetilde{Y}}_u,\Pi_u,X_u)h^{\beta,I,c_1,c_2}(u,W^{\widetilde{Y}}_u,\Pi_u,X_u)dZ_u$ is a martingale. Therefore, by using \eqref{primary_adj_HJB_1}, \eqref{g_bI_SDE_gamma_01}, and $\Phi_1=G_1$, we obtain
\begin{eqnarray}\label{chapter3_g_bI_inequality_gamma_01}
      &&g^{\beta,I,c_1,c_2}(\tau_n,W^{\widetilde{Y}}_{\tau_n},\Pi_{\tau_n},X_{\tau_n}) \leq g^{\beta,I,c_1,c_2}(t,W^{\widetilde{Y}}_t,\Pi_t,X_t)\notag\\
      &&\qquad \qquad \qquad \qquad \qquad +\int_t^{\tau_n} g^{\beta,I,c_1,c_2}(s,W^{\widetilde{Y}}_s,\Pi_s,X_s)h^{\beta,I,c_1,c_2}(s,W^{\widetilde{Y}}_s,\Pi_s,X_s)dZ_s.
\end{eqnarray}
Given that $\lim \limits_{n\rightarrow \infty} \tau_n = T$ and $g^{\beta,I,c_1,c_2}(t,W^{\widetilde{Y}}_t,\Pi_t,X_t)\geq0$ for any $t\in[0,T]$ under $0<\gamma<1$, we can prove
\begin{eqnarray}\label{chapter3_G_inequlity2_gamma_01}
      &&J(t,w^{\widetilde{Y}},\pi,X) \notag\\
      &=&  E_{t,w^{\widetilde{Y}},\pi,X}\left[ \int_t^{T_R} 
           {_{s-t}p_{x+t}} e^{-\delta(s-t)}\left[\kappa_1 U_1(c_{1,s},\Pi_s) + \kappa_2 U_2(c_{2,s},\Pi_s)\right.\right.\notag\\
      &&\left.\left.+\kappa_2\lambda_{x+s}\Phi_1\left(s,W^{\widetilde{Y}}_t - \widetilde{Y}(t,X_t) + \frac{I_t} 
          {\mu_{x+t}},\Pi_s,X_s\right)\right]ds\right.\notag\\
      &&\left.+ e^{-\delta(T_R-t)}{_{T_R-t}p_{x+t}} \Phi_2\left(T_R,W^{\widetilde{Y}}_{T_R},\Pi_{T_R},X_{T_R}\right)\right]\notag\\
      &=&E_{t,w^{\widetilde{Y}},\pi,X}[g^{\beta,I,c_1,c_2}(T_R,W^{\widetilde{Y}}_{T_R},\Pi_{T_R},X_{T_R})]\notag\\
      &\leq& \lim \limits_{n\rightarrow \infty} E_{t,w^{\widetilde{Y}},\pi,X}[g^{\beta,I,c_1,c_2} 
              (\tau_n,W^{\widetilde{Y}}_{\tau_n},\Pi_{\tau_n},X_{\tau_n})]\notag\\
      &\leq& g^{\beta,I,c_1,c_2}(t,w^{\widetilde{Y}},\pi,X)\notag\\
      &=& G(t,w^{\widetilde{Y}},\pi,X), ~\text{for}~ \forall (\beta,I)\in \mathcal{A}_{\gamma}(0,T),
\end{eqnarray}
where the first inequality follows from Fatou's lemma and the second inequality is obtained by taking conditional expectations on both sides of \eqref{chapter3_g_bI_inequality_gamma_01}. This completes the proof.
\end{proof}

\section{Estimation details for financial market}\label{appendix8}
Let $K_t = (X_{1,t}, X_{2,t},\log \Pi_t, \log S_t)^{\top}$ denote the vector of states in the financial market. Then, it evolves as follows
\begin{equation*}
	dK_t = (\theta_0 + \theta_1 K_t) dt + \Sigma_K dZ_t,
\end{equation*}
where 
\begin{equation*}
	\theta_0 = \begin{pmatrix}
		0_{2 \times 1}\\
		\delta_{\pi^e} - \frac{1}{2}\sigma^{\top}_{\Pi}\sigma_{\Pi}\\
		\delta_R + \mu_0 - \frac{1}{2}\sigma^{\top}_S\sigma_S
	\end{pmatrix},
	\theta_1 = \begin{pmatrix}
		-K_X       &0_{2 \times 2}\\
		e^{\top}_2 &0_{1 \times 2}\\
		\iota^{\top}_2-\sigma^{\top}_{\Pi}\Lambda_1 + \mu^{\top}_1 & 0_{1 \times 2}
	\end{pmatrix},
	\Sigma_K =  \begin{pmatrix}
		\Sigma_X\\
		\sigma^{\top}_{\Pi}\\
		\sigma^{\top}_S
	\end{pmatrix},
\end{equation*}
and $e_i$ represents the $i$-th unit vector in $\mathbb{R}^2$ and $\iota_2  = (1,1)^{\top}$. Applying Ito's formula, the transition equation for states satisfies
\begin{equation}\label{transition_equation}
	K_{t+\Delta t} = \Upsilon_1 + \Psi_1 K_t + \epsilon_{t+\Delta t}, ~~\epsilon_{t+\Delta t}\overset{i.i.d.}{\sim} N(0_{4\times 1}, \Sigma_{\epsilon}),
\end{equation}
where 
\begin{equation*}
	\Upsilon_1 = \int_0^{\Delta t} e^{\theta_1(\Delta t-s)}\theta_0 ds, \quad \Psi_1 = e^{\theta_1 \Delta t}, \quad \Sigma_{\epsilon} = \int_0^{\Delta t} e^{\theta_1 (\Delta t-s)}\Sigma_K \Sigma^{\top}_K (e^{\theta_1(\Delta t-s)})^{\top} ds.
\end{equation*}
For monthly data, we set $\Delta t = \frac{1}{12}$. Ten financial variables are observed each month, including the inflation index, equity index, and yield rates of nominal zero-coupon bonds across eight different maturities. Following \cite{koijen2011optimal}, we assume that the yield rates are observed with independent errors. Let $R^{Y}(t,\tau_i), i=1,2,...,8$ denote the yield rates of nominal zero-coupon bonds at time $t$ with maturity $\tau_i, i=1,2,...,8$. Then, the measurement equation for the states can be expressed as
\begin{equation}\label{measurement_equation}
	L_t = \Upsilon_2 + \Psi_2 K_t + \eta_{t}, ~  \eta_{t}\overset{i.i.d.}{\sim} N(0_{10\times 1}, \Sigma_{\eta}),	
\end{equation}
where $L_t = (R^{Y}(t,\tau_i)_{i=1,2,...,8}, \log \Pi_t, \log S_t)^{\top}$ is the observation vector. Furthermore, the coefficients in \eqref{measurement_equation} are
\begin{eqnarray*}
	\Upsilon_2 = \begin{pmatrix}
		-A_0(\tau_1)/\tau_1\\
		\vdots\\
		-A_0(\tau_8)/\tau_8\\
		0_{2 \times 1}
	\end{pmatrix}, ~
	\Psi_2 = \begin{pmatrix}
		-A^{\top}_1(\tau_1)/\tau_1 &0_{1 \times 2}\\
		\vdots &\vdots\\
		-A^{\top}_1(\tau_8)/\tau_8 &0_{1 \times 2}\\
		0_{2\times 2} &\widetilde{I}_2
	\end{pmatrix}, ~
	\Sigma_{\eta} = \begin{pmatrix}
		\chi_1 &        &         &     &\\
		& \ddots &         &     &\\
		&        & \chi_8  &     &\\
		&        &         & 0   &\\
		&        &         &     &0      
	\end{pmatrix},
\end{eqnarray*}
where $A_0$ and $A_1$ are given by \eqref{A0} and \eqref{A1} respectively, $\widetilde{I}_2$ is the 2nd-order identity matrix, and $\chi_i, i=1,2,...,8,$ are the measurement errors in yields to be estimated.

Let $\widetilde{L}_t$ denote the set of past observations $L_1, L_2, ..., L_t$ for $t=1,2,..., n$. We define the conditional means and variances as follows
\begin{eqnarray*}
    &&\bar{K}_{t|t} = E[K_t|\widetilde{L}_t], ~~\bar{K}_{t+1} = E[K_{t+1}|\widetilde{L}_t], ~~P_{t|t} = Var(K_t|\widetilde{L}_t), ~~P_{t+1} = Var(K_{t+1}|\widetilde{L}_t),\\
    &&v_t = L_t - E[L_t|\widetilde{L}_t] = L_t - \Upsilon_2 - \Psi_2 \bar{K}_t, ~~F_t = Var(v_t|\widetilde{L}_t) = \Psi_2 P_t \Psi^{\top}_2 + \Sigma_{\eta}.
\end{eqnarray*}
Next, we can express the Kalman filter iteration for \eqref{transition_equation} and \eqref{measurement_equation} as
\begin{eqnarray*}
    &&v_t = L_t  - \Upsilon_2 - \Psi_2 \bar{K}_t, ~F_t = \Psi_2 P_t \Psi^{\top}_2 + \Sigma_{\eta},\\
    &&\bar{K}_{t|t} = \bar{K}_t + P_t \Psi^{\top}_2 F^{-1}_t v_t, ~ P_{t|t} = P_t - P_t \Psi^{\top}_2 F^{-1}_t \Psi_2 P_t,\\
    &&\bar{K}_{t+1} = \Upsilon_1 + \Psi_1 \bar{K}_{t|t}, ~ P_{t+1} = \Psi_1 P_{t|t} \Psi^{\top}_1 + \Sigma_{\epsilon}.
\end{eqnarray*}
Let $\psi$ represent the vector of all model parameters, which encompasses 21 parameters in Table \ref{financial_market_estimate_table} and eight parameters in $\Sigma_{\eta}$. Then, we can derive the log-likelihood function from the ``prediction error decomposition" (see Chapter 3.4 in \cite{harvey1990forecasting})
\begin{eqnarray*}
    \text{logL}(\widetilde{L}_n|\psi) &=& \sum \limits_{t=1}^n \log p(L_t|\widetilde{L}_{t-1},\psi) = -\frac{10n}{2} \log(2\pi) - \frac{1}{2} \sum \limits_{t=1}^n (\log|F_t| + v^{\top}_t F^{-1}_t v_t).
\end{eqnarray*}
Finally, we can utilize the ``SSM" package in Matlab to estimate the maximum likelihood estimator (MLE) $\hat{\psi}$ of the unknown parameters $\psi$. Alternative estimation methodologies include the Expectation–maximization (EM) algorithm and Markov chain Monte Carlo (MCMC) algorithm, as detailed in Chapters 7.3.4 and 13.4 of \cite{durbin2012time}.

\end{document}